\documentclass[10pt]{article}
\makeatletter
\renewcommand{\thefootnote}{}

\usepackage{epsfig,amsmath,graphicx,amssymb,overpic}

\usepackage{graphicx}
\usepackage{dcolumn}
\usepackage{bm}
\usepackage{graphicx}
\usepackage{subfigure}
\usepackage{epsfig,amsmath,graphicx,amssymb,overpic,cite}

\usepackage{graphicx}
\usepackage{dcolumn}
\usepackage{bm}
\usepackage{graphicx}
\usepackage{subfigure}
\usepackage{amsthm}

\usepackage{pdfsync}
\usepackage{cite}
\usepackage{amscd}
\usepackage{amsmath}
\usepackage{latexsym}
\usepackage{amsfonts}
\usepackage{amssymb}
\usepackage{amsthm}
\usepackage{graphicx}
\usepackage{indentfirst}
\usepackage{enumerate}
\usepackage{amstext}
\usepackage[dvipdfm,
            pdfstartview=FitH,
            CJKbookmarks=true,
            bookmarksnumbered=true,
            bookmarksopen=true,
            colorlinks=true, 
            pdfborder=001,   
            citecolor=magenta,
            linkcolor=blue,
            ]{hyperref}

 \allowdisplaybreaks[4]

\newtheorem{corollary}{Corollary}[section]
\newtheorem{lemma}{Lemma}[section]
\newtheorem{prop}{Proposition}[section]
\newtheorem{remark}{Remark}[section]
\newtheorem{theorem}{Theorem}[section]

\usepackage{tikz}
\usetikzlibrary{shapes,arrows,positioning}
\usepackage{xcolor}

\def\be{\begin{equation}}
\def\ee{\end{equation}}
\def\bee{\begin{eqnarray}}
\def\ene{\end{eqnarray}}
\def\bes{\begin{subequations}}
\def\ees{\end{subequations}}
\def\det{{\rm det}}
\def\d{\displaystyle}
\def\v{\vspace{0.1in}}
\def\no{{\nonumber}}
\def\l{\left}
\def\r{\right}

\begin{document}

\baselineskip=13pt
\renewcommand {\thefootnote}{\dag}
\renewcommand {\thefootnote}{\ddag}
\renewcommand {\thefootnote}{ }

\pagestyle{plain}


\begin{center}
\baselineskip=16pt \leftline{} \vspace{-.3in} {\Large \bf The focusing complex mKdV equation with nonzero background:
 Large $N$-order asymptotics of multi-rational solitons and related Painlev\'{e}-III  hierarchy
} \\[0.2in]
\end{center}

\begin{center}
{\bf Weifang Weng$^{a}$, Guoqiang Zhang$^{b,c}$, Zhenya Yan}$^{b,c,*}$\footnote{$^{*}${\it Email address}: zyyan@mmrc.iss.ac.cn (Corresponding author)}  \\[0.1in]
{\it$^a$School of Mathematical Sciences,University of Electronic Science and Technology of China, Chengdu 611731, China} \\
\it $^b$KLMM,  Academy of Mathematics and Systems Science,  Chinese Academy of Sciences, Beijing 100190, China\\
\it $^c$School of Mathematical Sciences, University of Chinese Academy of Sciences, Beijing 100049, China \\[0.18in]
\end{center}

\noindent {\bf Abstract:}\, {\small In this paper, we investigate the large-order asymptotics of multi-rational solitons of the focusing complex modified Korteweg-de Vries (c-mKdV) equation with nonzero background via the Riemann-Hilbert problems. First, based on the Lax pair, inverse scattering transform, and a series of deformations, we construct a multi-rational soliton of the c-mKdV equation via a solvable Riemann-Hilbert problem (RHP). Then, through a scale transformation, we construct a RHP corresponding to the limit function which is a new solution of the c-mKdV equation in the rescaled variables $X,\,T$, and prove the existence and uniqueness of the RHP's solution. Moreover, we also find that the limit function satisfies the ordinary differential equations (ODEs) with respect to space $X$ and time $T$, respectively. The ODEs with respect to space $X$ are identified with certain members of the Painlev\'{e}-III hierarchy. We study the  large $X$ and transitional asymptotic behaviors of near-field limit solutions, and we provide some part results for the case of large $T$. These results will be useful to understand and apply the large-order rational solitons in the nonlinear wave equations.
}

\vspace{0.1in} \noindent {\bf Keywords}\, Complex mKdV equation; Nonzero background; Lax pair; Inverse scattering transform, Riemann-Hilbert problem;
Multi-rational solitons; Large-order asymptotics; Painlev\'{e}-III hierarchy

\vspace{0.1in} \noindent {\bf Mathematics Subject Classification}\, 35Q51, 35Q15, 37K40, 37K10


\baselineskip=17pt

\tableofcontents

\section{Introduction and main results}

In 1967, Gardner, Green, Kruskal and Miura~\cite{IST} found that  the Korteweg-de Vries (KdV) equation can be viewed as the compatibility condition of two linear partial differential equations (alias the Lax pair~\cite{lax}), and first presented the inverse scattering transform (IST) to exactly solve its $N$-soliton solutions. The powerful IST has been applied to other important
nonlinear integrable equations, such as the nonlinear Schr\"odinger (NLS) equation, modified Kdv equation, sine-Gordon equation, AKNS hierarchy, etc. (see, e.g., Ref.~\cite{soliton3,book87,yang,NMPZ} and references therein). One of most fundamental nonlinear integrable equations is the NLS equation~\cite{O1967}
\begin{equation}
\label{nls}
iq_t +q_{xx}+ 2\sigma |q|^2 q=0,\quad q=q(x,t): \mathbb{R}\times \mathbb{R}\to \mathbb{C}, \quad \sigma=\pm 1,
\end{equation}
which is used to describe the nonlinear waves in many fields, such as nonlinear optics~\cite{op1,op2,nls-1,nls-2}, deep ocean~\cite{Zak68,yuen82}, plasma physics~\cite{Zak72}, Bose-Einstein condensates  (alias the Gross-Pitaevskii equation~\cite{GP1,GP2}), and even finance~\cite{yanfrw}.
Zakharov and Shabat~\cite{zs72} first presented the Lax pair of the NLS equation,
and used the IST and Riemann-Hilbert problem to solve it. After that, another formal Riemann-Hilbert problem with IST was used to solve the NLS equation~\cite{NMPZ}.  In 1983,
Peregrine~\cite{Peregrine1983} considered the periodic parameter limit of the breathers~\cite{kuz77,kawata78,ma79} of the focusing ($\sigma=1$)  NLS equation to  first find its fundamental extreme  rational rogue wave (RW) solution (alias Peregrine soliton or Peregrine rogon):
\bee\label{p-soliton}
q_{nls}^{[1]}(x,t)=e^{2it}\left(1-\frac{4(1+4it)}{4x^2+16t^2+1}\right),
\ene
which is localized in both space and time, and rationally decays to nonzero background (see Fig.~\ref{1-2-order}(a,b)), $|q_{nls}^{[1]}|\to 1$ as $|x|,\,|t|\to \infty$, and whose maximum amplitude is three times that of the nonzero background. Moreover, The higher-order RWs of focusing NLS equation can be seen as the nonlinear superposition of a certain number of first-order rational solitons (e.g., see Figs.~\ref{1-2-order}(c,d) for the 2nd-order RW)~\cite{nail2009b,rw2010}, showing a very abundant structure, such as triangle, pentagon, and so on~\cite{ked2011,ling2017}.
The IST and RHP methods were used to exactly solve both focusing and defocusing NLS equations with nonzero backgrounds to find their breathers, dark solitons, and rational RW solutions~\cite{nls-nz,Demontis2013,bilman1,yang}. Moreover,  its near- and far-field limits of solutions of the focusing NLS equation with nonzero backgrounds were studied via IST and RHPs~\cite{bilman2,bilman3,bilman5}.  Moreover, a unified formula on the large-order asympotics was found for the infinite-order both solitons and RWs of the focusing NLS equation~\cite{Bilman21}. More recently, the asymptotics for large-order solitons was obtained for the coupled NLS equation~\cite{Ling24}.

Based on the IST, RHPs or/and nonlinear steepest descent methods~\cite{dz-1}, some works~~\cite{ZM76,its81,deift93,deift93b,vart02,deift03,zhou06,bert10,deift11,DM08,BM17,BLM21,BKS11,BIK09,bar-1,FLQ, BLS,BLS22,dnls22-fan,fan1-2,GL18} studied the long-time asymptotic behaviors of solutions of the focusing or defocusing NLS equation with various of different boundary conditions. Moreover, Cuccagna and Jenkins~\cite{bar-3} studied the asymptotic stability of $N$-solitons of the defocusing NLS equation. Fokas {\it et al}~\cite{fokas05,fokas12a,fokas12b,fokas12c} presented a unified method to show that the solutions of the  NLS equation with initial-boundary value conditions can be expressed via the RHPs. Recently, Koch{\it et al}~\cite{koch-1,koch-2,koch-3} established a family of conserved energies for the NLS equation.

As we have mentioned that though the NLS equation can be used to describe many physical phenomena, but the NLS equation has its boundedness.
For example, when the optical pulses become shorter (e.g., 100 fs~\cite{op1,op2}), the higher-order dispersive and nonlinear effects such as third-order dispersion, self-frequency shift, and self-steepening arising from the stimulated Raman scattering are
significant in the study of ultra-short optical pulse propagation~\cite{hnls,hnls2,yan13}. The Hirota equation and complex modified Korteweg-de Vries (c-mKdV) equation~\cite{hirota} are important higher-order extensions of the NLS equation.
The focusing c-mKdV equation~\cite{hirota}
\bee\label{cmkdv}
q_{t}+q_{xxx}+6|q|^2q_x=0,\quad q=q(x,t): \mathbb{R}\times \mathbb{R}\to \mathbb{C},
\ene
in fact, belongs to the integrable AKNS hierarchy~\cite{akns74}. The c-mKdV equation with a nonzero background (NZB) ($q(x,t)\to 1,\, {\rm as}\, |x|\to \infty$) has been verified to admit the multi-rational solitons ~\cite{hirota10}. For example, the expressions for the first-order and second-order rational solitons of  c-mKdV equation (\ref{cmkdv}) are given as follows (see  Figs.~\ref{1-2-order}(e-h)):
\begin{align} \label{s1}
\begin{aligned}
&q_{cmkdv}^{[1]}(x,t)=1-\dfrac{4}{4(x-6t)^2+1},\vspace{0.05in}\\
&q_{cmkdv}^{[2]}(x,t)=1-\frac{12\Delta_1(x,t)}{\Delta_2(x,t)},
\end{aligned}
\end{align}
where
\begin{align}\no
\begin{aligned}
\Delta_1(x,t)=&16x^4+24x^2+3168t^2+20736t^4-13824t^3x-384tx^3+3456t^2x^2-672tx-3,\vspace{0.05in}\\
\Delta_2(x,t)=& 48x^4+2985984t^6-2985984t^5x+1244160t^4x^2-276480t^3x^3-269568t^4+34560t^2x^4\vspace{0.05in}\\
&+124416t^3x-2304tx^5-17280t^2x^2+64x^6+384tx^3+20016t^2-2448tx+108x^2+9.
\end{aligned}
\end{align}
\begin{figure}[!t]
    \centering
 \vspace{-0.15in}
  {\scalebox{0.8}[0.8]{\includegraphics{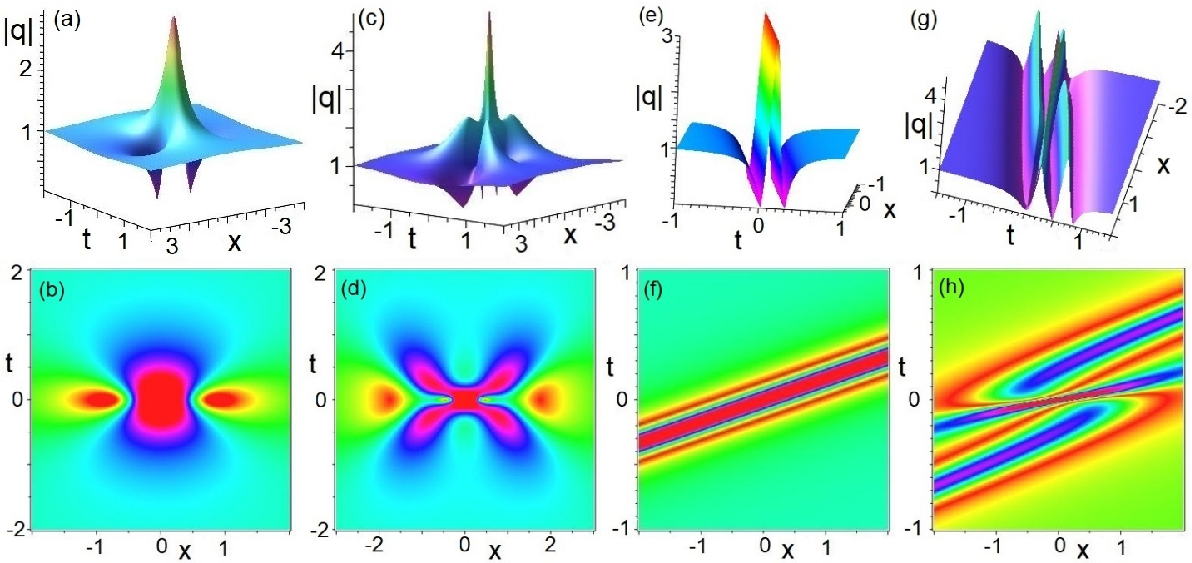}}}
\vspace{0.1in}
\caption{(a,b) 1-order rogue wave of NLS equation; (c, d) 2-order rogue wave of NLS equation;(e,f) 1-order rational soliton $q^{(1)}(x,t)$ of c-mKdV equation; (g, h) 2-order rational soliton $q^{(2)}(x,t)$ of c-mKdV equation.}
   \label{1-2-order}
\end{figure}

\begin{remark}
The W-shaped rational solitons of the c-mKdV equation differ from the rational rogue waves of the NLS equation even if they are both nonzero backgrounds. Moreover, these rational solitons also differ from the usual solitons with zero background.
\end{remark}

\begin{remark}
When the initial condition is taken as $q(x,0)=e^{i\sqrt{6}x}$, multi-rogue waves of the c-mKdV equation (\ref{cmkdv}) can be found~\cite{zha}
\bee
q_{cmkdv}(x,t)=\l[1-\frac{4(1-12\sqrt{6}it)}{4(x+12t)^2+432t^2+1}\r]e^{i\sqrt{6}x},
\ene
which differs from the W-shaped rational soliton $q_{cmkdv}^{[1]}(x,t)$ of the c-mKdV equation (\ref{cmkdv}) gvien by Eq.~(\ref{s1}).
\end{remark}

Recently, Chen and Yan~\cite{chen2019} found the $(2n-1,2n)$th-order RW solutions and rational solitons of both the Hirota equation and c-mKdV equation (\ref{cmkdv}) with non-zero boundary conditions, respectively, via the robust IST and RHPs~\cite{bilman1}.
Weng {\it et al}~\cite{weng2} found multi-rational solitons of $n$-component c-mKdV equations. Recently, Wang {\it et al}~\cite{fan-cmkdv22} and Zhang {\it et al}\cite{jmp2022longtime} studied the long-time asymptotics of the c-mKdV equation (\ref{cmkdv}) with step-like initial data and Schwartz decaying initial data, respectively.

Moreover, the Painlev\'{e} equation hierarchy is an important aspect of the study of integrable systems~\cite{soliton3}, and each Painlev\'{e} equation can be written as the following compatibility condition for the  Lax pair~\cite{Flaschka-1,Sakka-1}:
\bee\no
\phi_{x}(x,\lambda)=A(x,\lambda)\phi(x,\lambda),\quad \phi_{\lambda}(x,\lambda)=B(x,\lambda)\phi(x,\lambda),
\ene
where
\bee\no
A(x,\lambda)=\sum\limits_{k=0}^{L+l}A_k\lambda^{L-k},\quad B(x,\lambda)=\sum\limits_{k=0}^{N+n}B_k\lambda^{N-k},
\ene
and $A_k$ and $B_k$ are matrices with terms that depend on the solution of the Painlev\'{e}
equation.

In recent years, extreme rational soliton (e.g., rational rogue wave) phenomena with huge energies have been verified theoretically and/or experimentally in many fields~\cite{RW}, such as ocean~\cite{orw}, nonlinear optics~\cite{prw}, Bose-Einstein condensation~\cite{bec-rw}, finance~\cite{yanfrw}, capillary waves~\cite{Capillary-prl2010}, superfluid~\cite{superfluid}, cold atom~\cite{cold-a} and plasma physics~\cite{pp}. In particular,
rational solitons including rational rogue waves have been verified to appear in some nonlinear integrable systems~\cite{soliton,soliton2,soliton3,rw2017}, and nearly-integrable nonlinear wave equations~\cite{yan-rw19,yan-rw23}.

In this paper, we will focus on the large-order asymptotics of multi-rational solitons for the Cauchy problem of the focusing c-mKdV equation (\ref{cmkdv}) with finite density initial conditions 
\bee \label{id}
q(x,0)=q_0(x),\quad \lim_{x\to\pm\infty}q
_0(x)=q_{\pm},\quad \left|q_{\pm}\right|=q_0\ne 0
\ene
by analyzing the corresponding  Riemann-Hilbert problems.

The c-mKdV equation has the plane {\color{red} wave} solution
\bee \label{cw}
q_{cw}(x,t)=\alpha e^{i[kx-(k^3-6\alpha^2k)t+\vartheta]},\quad  \alpha,\, k,\, \vartheta\in \mathbb{R}
\ene
As $k=0$, the plane wave is simplified as a nonzero constant,
$q_{cw}(x,t)=\alpha e^{i \vartheta}$ with $\alpha\not=0$.  Without loss of generality, we here take $q_{\pm}=q_0=1$ in Eq.~(\ref{id}).
Moreover, we know that Eq.~(\ref{cmkdv}) has the solitary wave solution
\bee\label{sw}
q_{cmkdv}(x,t)=\alpha\, {\rm sech}[\alpha x-(\alpha^3-3k^2\alpha )t+c]e^{i[kx+(k^3-3\alpha^2k)t+\vartheta]},\quad \alpha,\, c,\, k,\, \vartheta\in \mathbb{R}
\ene
which approaches to $0$ as $|x|, |t|\to \infty$.
The phase and group velocities of the plane wave (\ref{cw}) are $v_{pv}=k^2-6\alpha^2$ and
$v_{gv}=3k^2-6\alpha^2$, respectively. However, phase and group velocities of the solitary wave (\ref{sw}) are $v_{pv}=-k^2+3\alpha^2$ and
$v_{gv}=-3k^2+3\alpha^2$, respectively.

\begin{remark} The c-mKdV equation (\ref{cmkdv}) is invariant under the
scaling transform
\bee
 x\to \gamma^{-1}x',\quad t\to \gamma^{-3}t',\quad q(x,t)\to \gamma e^{i\vartheta} q(x',t'),
\ene
where $\gamma\not=0$ is arbitrary real constants, and $\vartheta\in [0, 2\pi)$.
\end{remark}

For convenience, we introduce the following notations:
\bee\no
&\sigma_1=\!\left[\!\!\begin{array}{cc}
0& 1  \vspace{0.05in}\\
1 & 0
\end{array}\!\!\right],\quad \sigma_2=\!\left[\!\!\begin{array}{cc}
0& -i  \vspace{0.05in}\\
i & 0
\end{array}\!\!\right],
\quad\sigma_3=\!\left[\!\!\begin{array}{cc}
1& 0  \vspace{0.05in}\\
0 & -1
\end{array}\!\!\right],\quad Q=\!\dfrac{\sqrt{2}}{2}\left[\!\!\begin{array}{cc}
1& -1  \vspace{0.05in}\\
1 & 1
\end{array}\!\!\right].
\ene

Eq.~(\ref{cmkdv}) possesses the Lax pair:
\begin{align}\label{lax1}
\left\{\begin{aligned}
\psi_x=&\,
(-i\lambda \sigma_3+U)\psi=\left[\!\!\begin{array}{cc}
-i\lambda& q  \vspace{0.05in}\\
-q^* & i\lambda
\end{array}\!\!\right]\psi,\quad U=\!\left[\!\!\begin{array}{cc}
0& q(x,t)  \vspace{0.05in}\\
-q^*(x,t) & 0
\end{array}\!\!\right], \vspace{0.1in}\\
\psi_t=&\,(-4i\lambda^3\sigma_3+W)\psi=\left[\!\!\begin{array}{cc}
-4i\lambda^3+2i\lambda|q|^2+q_x^*q-q_xq^*& 4\lambda^2q+2i\lambda q_x-q_{xx}-2|q|^2q  \vspace{0.05in}\\
-4\lambda^2q^*+2i\lambda q_x^*+q_{xx}^*+2|q|^2q^* & 4i\lambda^3-2i\lambda|q|^2-q_x^*q+q_xq^*
\end{array}\!\!\right]\psi,\vspace{0.1in}\\
& W=4\lambda^2U-
2i\lambda (U_x+U^2)\sigma_3+[U_x, U]+2U^3-U_{xx},
\end{aligned}\right.
\end{align}
where $\psi=\psi(x,t;\lambda)$ is an unknown $2\times2$ matrix-valued eigenfunction, $\lambda\in\mathbb{C}$ is a spectral parameter, that is, Eq.~(\ref{cmkdv}) is viewed as the compatibility condition of the Lax pair.

We here review some basic properties about the robust IST of the c-mKdV equation (\ref{cmkdv})~\cite{akns74,chen2019}. According {\color{red} to } the boundary conditions $\lim\limits_{x\rightarrow\pm\infty}q(x,t)=1$, the Lax pair (\ref{lax1}) becomes
\begin{align}\label{lax1c}
\left\{\begin{aligned}
\psi_x^{bg}=&\,
(-i\lambda \sigma_3+U^{bg})\psi^{bg},\quad U^{bg}=\!\left[\!\!\begin{array}{cc}
0& 1  \vspace{0.05in}\\
-1 & 0
\end{array}\!\!\right], \vspace{0.1in}\\
\psi_t^{bg}=&\,\l(-4i\lambda^3\sigma_3+4\lambda^2U^{bg}-
2i\lambda (U^{bg})^2\sigma_3
+2(U^{bg})^3\r)\psi^{bg},
\end{aligned}\right.
\end{align}
which has the following fundamental solution matrix:
\bee
\psi^{bg}(\lambda; x,t)=E(\lambda)e^{-i\rho(\lambda)\big(x+(4\lambda^2-2)t\big)\sigma_3},
\ene
where
\bee
E(\lambda)=\!w(\lambda)\left[\!\!\begin{array}{cc}
1& i(\lambda-\rho(\lambda))  \vspace{0.05in}\\
i(\lambda-\rho(\lambda)) & 1
\end{array}\!\!\right],\quad \rho^2(\lambda)=\lambda^2+1,\quad w^2(\lambda)=\frac{\lambda+\rho(\lambda)}{2\rho(\lambda)},
\ene
which is a two-sheeted Riemann surface for $\lambda$, and $\rho(\lambda)=\lambda+\mathcal{O}(\lambda^{-1}),\quad \lambda\rightarrow\infty$ such that
$\lim_{\lambda\to \infty}w(\lambda)=1$.
Then, we have
$\det(E(\lambda))=\det(Q)=1.$

It can be easy to find the unique Jost solutions $\psi_{\pm}(\lambda;x,t)$ of Lax pair (\ref{lax1}):
\bee
\psi_{\pm}(\lambda;x,t)=\mu_{\pm}(\lambda;x,t) e^{-i\rho(\lambda)(x+(4\lambda^2-2)t)\sigma_3},
\ene
with
\bee
\psi_{\pm}(\lambda;x,t)e^{i\rho(\lambda)(x+(4\lambda^2-2)t)\sigma_3}=E(\lambda)+o(1),\qquad
x\to \pm\infty,
\ene
where
\bee\label{mu}
\mu_{\pm}(\lambda;x,t)=E(\lambda)+\int_{\pm\infty}^{x}E(\lambda)e^{-i\rho(\lambda)(x-y)\widehat{\sigma}_3}(E(\lambda)^{-1}\Delta q(y,t)\mu_{\pm}(\lambda;{\color{red} y},t))dy
\ene
with $\Delta q(x,t)=\left[\!\!\begin{array}{cc}
0& q(x,t)-1 \vspace{0.05in}\\
1-q^{*}(x,t) & 0
\end{array}\!\!\right]$ and $e^{\widehat{\sigma}_3}\bullet=e^{\sigma_3}\bullet e^{-\sigma_3}$.

Let $\mu_{\pm j}$ {\color{red} stand} for their $j$th column. Then, according to Eq.~(\ref{mu}), it can be shown that $\mu_{-1}(\lambda;x,t),\mu_{+2}(\lambda;x,t)$ are bounded and analytic in $\mathbb{C}^+\setminus\Sigma_c$ and continuous up to the boundary ($\mathbb{C}^+$ is the upper half-plane of the complex $\lambda$-plane and $\Sigma_c=[-i,i]$), and $\mu_{+1}(\lambda;x,t),\mu_{-2}(\lambda;x,t)$ are bounded and analytic in $\mathbb{C}^-\setminus\Sigma_c$ and continuous up to the boundary ($\mathbb{C}^-$ is the lower half-plane of the complex $\lambda$-plane).

When $\lambda\in\mathbb{R}\cup\Sigma_c$, the Jost solutions $\psi_{\pm}(\lambda;x,t)$ have a relationship by a scattering matrix:
\bee
\psi_{+}(\lambda;x,t)=\psi_{-}(\lambda;x,t)S(\lambda),
\qquad
S(\lambda)=\left[\!\!\begin{array}{cc}
s_{11}^*(\lambda^*)& s_{12}^*(\lambda^*) \vspace{0.05in}\\
-s_{12}(\lambda) & s_{11}(\lambda)
\end{array}\!\!\right].
\ene

\begin{lemma} ~\cite{chen2019,bilman1} For the fixed $L\in\mathbb{R}$, suppose that $q(x,t)$ is a bounded classical solution of Eq.~(\ref{cmkdv}) defined
for $(x,t)$ in a simply connected domain $\Omega\subseteq\mathbb{R}$ which contains the point $(L,0)$. Then for each $\lambda\in\mathbb{C}$
there exists a unique simultaneous fundamental solution matrix $\psi=\psi^{\mathrm{in}}(x,t;\lambda)$, which satisfies Lax pairs (\ref{lax1}) and the initial
condition $\psi^{\mathrm{in}}(L,0;\lambda)=\mathbb{I}$. For each $(x,t)\in\Omega$, $\psi=\psi^{\mathrm{in}}(x,t;\lambda)$ is an entire function of $\lambda$ and $\det(\psi^{\mathrm{in}}(x,t;\lambda))=1$.
\end{lemma}

Let $\Sigma_0$ represent a clockwise circular contour centered on the origin and having a radius $r_0$ greater than 1. And let
\bee
M(x,t;\lambda)=\psi(x,t;\lambda)e^{i\rho(\lambda)(x+(4\lambda^2-2)t)\sigma_3},\quad \lambda\in\mathbb{C}\setminus(\Sigma_0\cup\Sigma_c), \quad (x,t)\in\mathbb{R}\times\mathbb{R}^+,
\ene
with
\bee\no
\psi(x,t;\lambda)=\begin{cases}
\begin{array}{ll}\left(\!\!\begin{array}{cc}
\dfrac{\psi_{-1}(x,t,\lambda)}{s_{11}(\lambda)}, & \psi_{+2}(x,t,\lambda)
\end{array}\!\!\right), & \lambda\in D_+,\vspace{0.05in}\\
\left(\!\!\begin{array}{cc}
\psi_{+1}(x,t,\lambda), & \dfrac{\psi_{-2}(x,t,\lambda)}{s_{11}^*(\lambda^*)}
\end{array}\!\!\right), & \lambda\in D_-,\vspace{0.05in}\\
\psi^{\mathrm{in}}(x,t;\lambda),&\lambda\in D_0,
\end{array}
\end{cases}
\ene
where $\Sigma_+:=\Sigma_0\cap\mathbb{C}^+$,~$\Sigma_-:=\Sigma_0\cap\mathbb{C}^-$,~$\Sigma_l:=\{\lambda\in\mathbb{R}\,\,|\,\,\lambda\leq-r_0\}$,~
$\Sigma_{r_0}:=\{\lambda\in\mathbb{R}\,\,|\,\,\lambda\geq r_0\}$,~$D_0:=\{\lambda\in\mathbb{C}\,\,|\,\,|\lambda|<r_0\}$,~$D_+:=\{\lambda\in\mathbb{C}\,\,|\,|\lambda|>r_0\}\cap\mathbb{C}^+$,
~$D_-:=\{\lambda\in\mathbb{C}\,\,|\,\,|\lambda|>r_0\}\cap\mathbb{C}^-$.

\begin{figure}[!t]
    \centering
 \vspace{-0.15in}
  {\scalebox{0.4}[0.4]{\includegraphics{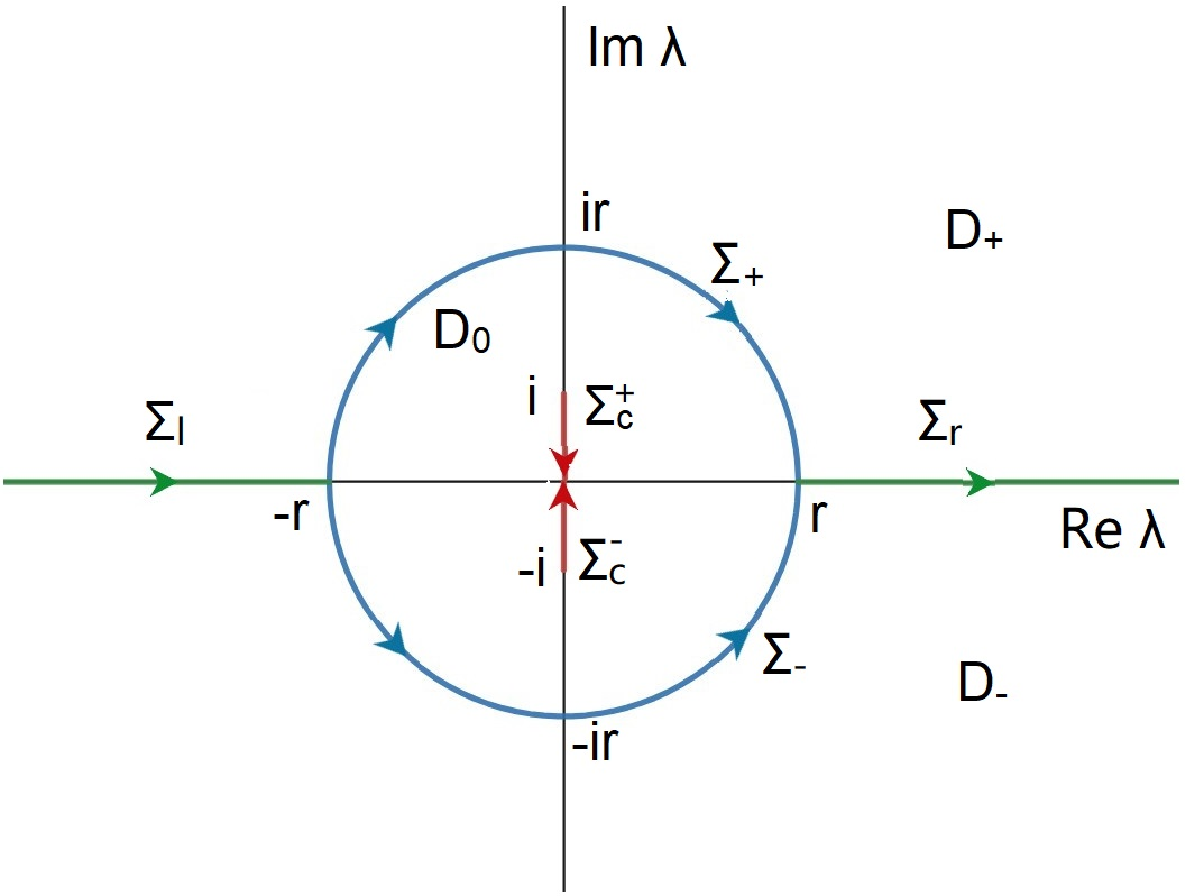}}}\hspace{-0.35in}
\vspace{0.05in}
\caption{The jump contour $\Sigma_l\cup\Sigma_r\cup\Sigma_+\cup\Sigma_-\cup\Sigma_c$ for $M(x,t,\lambda)$~\cite{chen2019}.}
   \label{D+D-}
\end{figure}

\begin{lemma} ~\cite{chen2019}  The modified Riemann-Hilbert problem of the c-mKdV equation (\ref{cmkdv}), that is,
\begin{itemize}
    \item{} Analyticity: $M(x,t;\lambda)$ is a meromorphic function in $\lambda\in\mathbb{C}\setminus(\Sigma_l\cup\Sigma_r\cup\Sigma_0\cup\Sigma_c)$;

 \item{} Jump conditions: The boundary values on the jump contour $\Sigma_l\cup\Sigma_r\cup\Sigma_0\cup\Sigma_c$ are related as:
\bee\no
M_+(x,t;\lambda)=M_-(x,t;\lambda)V(x,t;\lambda).
\ene

 \item{} Normalization: $M(x,t;\lambda)={\color{red}\mathbb{I}}+\mathcal{O}(\lambda^{-1})$ as $\lambda\rightarrow\infty$.
\end{itemize}
where the jump contour is given in Fig.~\ref{D+D-}, and the jump matrix $V(x,t;\lambda)$ is defined by
\bee\no
V(x,t;\lambda)=\begin{cases}
e^{-i\rho(\lambda)(x+(4\lambda^2-2)t)\widehat{\sigma}_3}V_1(\lambda),\quad \lambda\in \Sigma_0,\vspace{0.05in}\\
e^{2i\rho_+(\lambda)(x+(4\lambda^2-2)t)\sigma_3},\quad \lambda\in \Sigma_c,
\end{cases}
\ene
with
\bee\no
V_1(\lambda)=\begin{cases}
\left(\!\!\begin{array}{cc}
\dfrac{\psi_{-1}(L,0;\lambda)}{s_{11}(\lambda)}, & \psi_{+2}(L,0;\lambda)
\end{array}\!\!\right),\quad \lambda\in D_+,\vspace{0.1in}\\
\left(\!\!\begin{array}{cc}
\psi_{+1}(L,0;\lambda), & \dfrac{\psi_{-2}(L,0;\lambda)}{s_{11}^*(\lambda^*)}
\end{array}\!\!\right)^{-1},\quad\lambda\in D_-,\vspace{0.1in}\\
\left[\!\!\begin{array}{cc}
1
+|\hat{r}(\lambda)|^2& \hat{r}^*(\lambda) \vspace{0.05in}\\
\hat{r}(\lambda) & 1
\end{array}\!\!\right],\qquad\lambda\in \Sigma_l\cup\Sigma_r,\quad \hat{r}(\lambda)=\dfrac{s_{12}(\lambda)}{s_{11}(\lambda)},
\end{cases}
\ene
has a unique solution $M(x,t;\lambda)$ for $(x,t)\in\mathbb{R}\times\mathbb{R}^+$.
\end{lemma}

Then, the solution of c-mKdV equation (\ref{cmkdv}) can be given by the solution $M(x,t;\lambda)$:
\bee\no
q(x,t)=2i\lim_{\lambda\rightarrow\infty}\lambda M_{12}(x,t;\lambda).
\ene

We consider the following gauge transformation:
\bee\no
\widehat{\psi}(x,t;\lambda)=\begin{cases}
G(x,t;\lambda)\psi(x,t;\lambda),\quad \lambda\in D_+\cup D_-,\vspace{0.05in}\\
G(x,t;\lambda)\psi(x,t;\lambda)G(L,0;\lambda)^{-1},\quad \lambda\in D_0.
\end{cases}
\ene
By adjusting the parameters in function $G(x,t;\lambda)$ and performing multiple gauge transformations, higher-order rational soliton solutions of c-mKdV equation (\ref{cmkdv}) can be obtained. Thus one can find the $(2n-1,\, 2n)$th-order rational soliton solution of Eq.~(\ref{cmkdv}) in the form \cite{chen2019}
\bee \label{solu-n} \begin{array}{l}
\tilde{q}_{0/e}(x,t)
=1+2i\dfrac{\rm{det}({\mathcal P}_{o/e,4n}(x,t))+\rm{det}({\mathcal P}_{o/e,4n-1}(x,t))}{\rm{det}({\mathcal P}_{o/e}(x,t))},
\end{array}
\ene
where $\mathcal{P}_{o/e}(x,t)=\mathcal{P}_{(2n-1)/(2n)}(x,t)$ is a coefficient matrix that consists of $2n\times 2n$ blocks:
{\small\begin{align}\no
&\mathcal{P}_{o/e}(x,t)=\\\no
&\left(\!\!\!\begin{array}{cccccccc}
                  A^{[1]}_0& A^{[2]}_0&0&0&\cdots&\cdots&\cdots&0 \v\\
                  A^{[1]}_1& A^{[2]}_1& A^{[1]}_0& A^{[2]}_0&0&\cdots&\cdots&0 \v\\
                 \vdots & \vdots & \vdots & \vdots & \vdots & \vdots & \vdots & \vdots  \v\\
                  A^{[1]}_{n-1}&A^{[2]}_{n-1}&A^{[1]}_{n-2}&A^{[2]}_{n-2}&\cdots&\cdots& A^{[1]}_{0}&A^{[2]}_{0} \v\\
                  A^{[1]}_{n}+B^{[1]}_{0,n}& A^{[2]}_{n}+B^{[2]}_{0,n}&A^{[1]}_{n-1}+B^{[1]}_{0,n-1}&A^{[2]}_{n-1}+B^{[2]}_{0,n-1}&\cdots&\cdots&A^{[1]}_{1}+B^{[1]}_{0,1}
                  &A^{[2]}_{1}+B^{[2]}_{0,1} \v\\
                  A^{[1]}_{n+1}+B^{[1]}_{1,n}&A^{[2]}_{n+1}+B^{[2]}_{1,n}&A^{[1]}_{n}+B^{[1]}_{1,n-1}&A^{[2]}_{n}+B^{[2]}_{1,n-1}&\cdots&\cdots&A^{[1]}_{2}+B^{[1]}_{1,1}
                  &A^{[2]}_{2}+B^{[2]}_{1,1}\v\\
                  A^{[1]}_{n+2}+B^{[1]}_{2,n}&A^{[2]}_{n+2}+B^{[2]}_{2,n}&A^{[1]}_{n+1}+B^{[1]}_{2,n-1}&A^{[2]}_{n+1}+B^{[2]}_{2,n-1}&\cdots&\cdots&A^{[1]}_{3}+B^{[1]}_{2,1}
                  &A^{[2]}_{3}+B^{[2]}_{2,1}\v\\
                \vdots&\vdots & \vdots & \vdots & \vdots & \vdots & \vdots & \vdots \v\\
                A^{[1]}_{2n-1}\!+\!B^{[1]}_{n-1,n}& A^{[2]}_{2n-1}\!+\!B^{[2]}_{n-1,n}& A^{[1]}_{2n-2}\!+\!B^{[1]}_{n-1,n-1}&
                A^{[2]}_{2n-2}\!+\!B^{[2]}_{n-1,n-1}&\cdots&\cdots&
                A^{[1]}_{n}\!+\!B^{[1]}_{n-1,1}&
                 A^{[2]}_{n}\!+\!B^{[2]}_{n-1,1}
                 \end{array}\!\!\!
                 \right)
\end{align}}
with
\bee \no A_k^{[j]}=\left(\begin{array}{cc}
(w^+_{o/e,k}(x,t))_j&0\\
0&(w^-_{o/e,k}(x,t))_j
\end{array}
                 \right),\quad j=1,2,
\ene
\bee\no  B_{m,k}^{[j]}=:\left(\begin{array}{cc}
0&\sum_{l=0}^{m}\gamma_{(m-l)k}(w^+_{o/e,l}(x,t))_j\\
\sum_{l=0}^{m}(-1)^{k+m-l}\gamma_{(m-l)k}(w^-_{o/e,l}(x,t))_j&0
\end{array}\right),\quad j=1,2,\ene
$\gamma_{km}=\frac{(-1)^k}{(2i)^{m+k}}\left(\begin{array}{c}
m+k-1\\ k \end{array}  \right)$,\, $(m=1,2,\ldots;\, k=0,1,\ldots),$
$w^+_{o/e,k}(x,t)=D_k^+(x,t){\bf c}_{\infty,o/e}$ and $w^-_{o/e,k}(x,t)=D_k^-(x,t)\sigma_3{\bf c}_{\infty,o/e},\, k=0,1,2,\ldots$ with $D_k^+(x,t)$, $D_k^-(x,t)$ being the Taylor coefficients of the analytic function $\psi_{bg}(\lambda;x,t)\psi_{bg}(\lambda;0,0)^{-1}$ in $\lambda=\pm i$, i.e., $
\Psi^{in}_{bg}(\lambda;x,t)=\sum_{j=0}^{\infty}D_j^{\pm}(x,t)(\lambda\mp i)^j,$
and $\mathcal{P}_{o/e,k}$ stands for the matrix $\mathcal{P}_{o/e}$ with its $k$th column replaced by the vector $(0,\cdots, 0, -(w^+_{o/e,0}(x,t))_1, -(w^-_{o/e,0}(x,t))_1, -(w^+_{o/e,1}(x,t))_1, -(w^-_{o/e,1}(x,t))_1$, $\cdots$,
$-(w^+_{o/e,n-1}(x,t))_1, -(w^-_{o/e,n-1}(x,t))_1)^T$, the signs $(\cdot)_{o}$ and $(\cdot)_{e}$ correspond to the data ${\bf c}_{\infty}=(1, -1)^T$ and  ${\bf c}_{\infty}=(1, 1)^T$, respectively (see Ref.~ \cite{chen2019} for specific parameters).


\v It should be worth noting that though the c-mKdV equation (\ref{cmkdv}) and the NLS equation possesses the same ZS spectral problem, but the temporal parts of their Lax pairs are distinct such that the NLS equation admits the multi-rational RWs with finite numbers of extreme points, but the c-mKdV equation possesses the multi-rational solitons with infinitely many extreme points (see. Fig.~\ref{1-2-order}). Compared with the NLS equation, the phase function of the c-mKdV equation is different and has a higher degree, resulting in more asymptotic regions. The properties of limiting functions require more techniques and models to be proposed.  Some  highlights for the higher-order rational solitons are found for the Cauchy problem of the c-mKdV equation:


\begin{itemize}

\item [(1)]
The form of the obtained multi-rational soliton of the c-mKdV equation given by Eq.~(\ref{solu-n}) is very difficult to directly analyze the large $N$-order asymptotics. We here use the Riemann-Hilbert problems to  reconstruct multi-rational soliton solutions of the c-mKdV equation through a series of deformations to analyze the relevant properties of solution (see Proposition ~\ref{prop3} with RHP 2).

\item [(2)] Through the scale transform of spatio-tempotal variables and spectral parameter: $\{X:=nx,\, T:=n^3t,\, \Lambda:=\frac{\lambda}{n}\}$, we construct the Riemann-Hilbert problem corresponding to the limit functions of the multi-rational soliton solutions of the c-mKdV equation (\ref{cmkdv}). In order to analyze the PDE satisfied by the limit functions, we need to overcome the difficulties caused by the time part of the Riemann-Hilbert problem corresponding to the limit functions. At this point, we need to use higher-order constraints on the spatial part and take many computational techniques. As a result, Eq.~(\ref{ab-26-1}) show that the limit functions satisfy the new complex mKdV equation (\ref{cmkdv-1}) in the $(X, T)$-space.

\item [(3)] The Riemann-Hilbert problem corresponding to the limit functions of the c-mKdV equation is of odd power with respect to time variable, which affects the properties of the limit functions. In order to  analyze the symmetries of the limit functions (see Propositions \ref{sym-prop2} and \ref{prop-real}), we need to provide a new deformation for the Riemann-Hilbert problem. As a result, we find that the limit functions of the multi-rational soliton solution of c-mKdV equation are real (see Proposition \ref{prop-real}), which can also be reflected in the asymptotic part of the solution.

\item [(4)] Due to the difficulty caused by the term $e^{-i(\Lambda X+4\Lambda^3T)\sigma_3}$ in the Riemann-Hilbert problem, it is no longer easy to derive the ODEs satisfied by the limit functions. By using a series of analysis, we not only provide the ODEs (\ref{eq8}) with respect to space $X$ for fixed variable $T$ satisfied by the limit functions, but also provide the ODEs (\ref{5eq}) with respect to times $T$ for fixed variable $X$ satisfied by the limit functions. The ODEs with respect to space variable $X$ are identified with certain members of the Painlev\'{e}-III hierarchy (see Theorem \ref{new-th-1}).

\item [(5)] Because  the phase function contains the high-power parameter $z$ related to the spectral parameter, the c-mKdV equation has five asymptotic regions. And due to the limit function has special symmetry (\ref{27-real}) (the limit function is real), we need to use the  Painlev\'e-II differential equation (\ref{27-painleve-2}) to solve the leading term and error estimation of the limit function.

\end{itemize}

\v\noindent {\bf Main results}

Our main conclusions of this paper are summarized as follows:

\begin{theorem}\label{prop5} The Riemann-Hilbert Problem $3$ given by Proposition \ref{prop4} has a unique solution $N^{\pm}(\Lambda;X,T)$ with $\det(N^{\pm}(\Lambda;X,T))=1$. For each compact subset $K\subset\mathbb{R}^2$,
\bee\label{C-k}
\sup\limits_{|\Lambda|\neq1,(X,T)\in K}||N^{^\pm}(\Lambda;x,t)||=C(K)<\infty.
\ene
The function
\bee\label{fanyan-q1}
\widehat{q}^{\pm}(X,T)=2i\lim_{\Lambda\rightarrow\infty}\Lambda N^{\pm}_{12}(\Lambda;X,T)
\ene
is a global solution of the complex modified Korteweg-de Vries (c-mKdV) equation:
\bee\label{cmkdv-1}
\frac{\partial\widehat{q}}{\partial T}+\frac{\partial^3\widehat{q}}{\partial X^3}+6|\widehat{q}|^2\frac{\partial\widehat{q}}{\partial X}=0,\quad \widehat{q}=\widehat{q}(X,T).
\ene
\end{theorem}

\begin{theorem}\label{th1}
Let $q_k(x,t)$ {\color{red} be} the
$k$th-order rational soliton of the c-mKdV equation (\ref{cmkdv}) given by the formula (\ref{solu-n}). Then, if $k=2n,n\in\mathbb{N}$, one has
\bee
\frac{1}{n}q_{2n}\l(x,\,t\r)=\frac{1}{n}q_{2n}\l(\frac{X}{n},\,\frac{T}{n^3}\r)=\widehat{q}^{+}(X,T)+\mathcal{O}\l(\frac{1}{n}\r),\quad n\rightarrow\infty,
\ene
if $k=2n-1,n\in\mathbb{N}_+$, one has
\bee
\frac{1}{n}q_{2n-1}\l(x,\,t\r)=\frac{1}{n}q_{2n-1}\l(\frac{X}{n},\,\frac{T}{n^3}\r)=\widehat{q}^{-}(X,T)
+\mathcal{O}\l(\frac{1}{n}\r),\quad n\rightarrow\infty,
\ene
uniformly for $(X,T)$ in compact subsets of $\mathbb{R}^2$, where $\widehat{q}^{\pm}(X,T)$ satisfies the c-mKdV Eq.~(\ref{cmkdv-1}).
\end{theorem}

\begin{theorem}\label{new-th-1}
The limit function $\widehat{q}^{\pm}(X,T)$ as a function of $X$ satisfies the following ordinary differential equations for fixed variable $T$:
\begin{align}\label{eq8}
\begin{aligned}
&3\widehat{q}\widehat{q}_{XX}-2\widehat{q}_X^2+4\widehat{q}^4+X(\widehat{q}\widehat{q}_{XXX}-\widehat{q}_X\widehat{q}_{XX}+4\widehat{q}^3\widehat{q}_X)\vspace{0.05in}\\
&+3T\bigg(\widehat{q}_X\widehat{q}_{XXXX}-\widehat{q}\widehat{q}_{XXXXX}
-6\widehat{q}^2(\widehat{q}\widehat{q}_{XX})_x-24\widehat{q}^2\widehat{q}_X
\widehat{q}_{XX}-4\widehat{q}^3\widehat{q}_{XXX}-24\widehat{q}^5\widehat{q}_X\bigg)=0,\vspace{0.05in}\\
&\bigg(\left(\widehat{q}+X\widehat{q}_X-3T\widehat{q}_{XXX}-18T\widehat{q}^2\widehat{q}_X\right)_X\bigg)^2+4\widehat{q}^2\left(\widehat{q}+X\widehat{q}_X-3T\widehat{q}_{XXX}-18T\widehat{q}^2\widehat{q}_X\right)^2-64\widehat{q}^2=0.
\end{aligned}
\end{align}

When $T=0$, Eq.~(\ref{lax-new-1}) becomes
\begin{align}\label{lax1-2g}
\begin{aligned}
&H_X(\Lambda; X,0)=\left[\!\!\begin{array}{cc}
-i\Lambda& \widehat{q}(X,0)  \vspace{0.05in}\\
-\widehat{q}(X,0) & i\Lambda
\end{array}\!\!\right]H(\Lambda; X,0),\vspace{0.05in}\\
&H_\Lambda(\Lambda; X,0)=\left[\!\!\begin{array}{cc}
-iX-\dfrac{i(2\widehat{q}_X+X\widehat{q}_{XX})}{4\widehat{q}}\Lambda^{-2}& X\widehat{q}\Lambda^{-1}+\dfrac{i}{2}(\widehat{q}+X\widehat{q}_X)\Lambda^{-2}  \vspace{0.05in}\\
-X\widehat{q}\Lambda^{-1}+\dfrac{i}{2}(\widehat{q}+X\widehat{q}_X)\Lambda^{-2} & iX+\dfrac{i(2\widehat{q}_X+X\widehat{q}_{XX})}{4\widehat{q}}\Lambda^{-2}
\end{array}\!\!\right]H(\Lambda; X,0),
\end{aligned}
\end{align}
which is the Lax pair of first member of the Painlev\'{e}-III equation hierarchy.

And when $T=0$ and $f:=\partial_X\ln\widehat{q}(X,0)$, Eq.~(\ref{eq8}) can be written as
\bee
X^2(ff_{XX}-f_X^2-f^4)+X(f_{XX}+ff_X-3f^3)+3f_X-3f^2+64=0{\color{red}.}
\ene
\end{theorem}

\begin{theorem}\label{new-th-2}
The limit function $\widehat{q}^{\pm}(X,T)$ as a function of $T$ satisfies the following ordinary differential equations for fixed variable $X$:
\begin{align}\label{6eq}
\begin{aligned}
&\widehat{q}+X\widehat{q}_X+3T\widehat{q}_T-2\beta=0,\vspace{0.05in}\\
&4\alpha \widehat{q}+2\beta_X=0,\vspace{0.05in}\\
&X\widehat{q}_T-3(\widehat{q}_{XX}+2\widehat{q}^3)-3T(\widehat{q}_{XXT}+6\widehat{q}_T\widehat{q}^2)+4\beta \widehat{q}^2-4\alpha \widehat{q}_X=0,\vspace{0.05in}\\
&\alpha_T+2\beta(\widehat{q}_{XX}+2\widehat{q}^3)=0,\vspace{0.05in}\\
&\beta_T-2\alpha(\widehat{q}_{XX}+2\widehat{q}^3)=0,\vspace{0.05in}\\
&\alpha(X,T)^2+\beta(X,T)^2=4.
\end{aligned}
\end{align}

\end{theorem}

\begin{theorem}\label{Big-X-q} Let $-1/96<a\leq0$ and $\Theta(z;a):=z+4az^3+2z^{-1}$. Then we know that the solution of Eq.~(\ref{cmkdv-1}) has the large $X$ asymptotics ($\widehat{q}^+(X,T)=-\widehat{q}^-(X,T)$)
\bee
\widehat{q}^+(X,aX^2)=\frac{\sqrt{2p}\cos(\phi(X,a))}{X^{\frac34}\sqrt{6ab+b^{-3}}}+\mathcal{O}(X^{-1}),\quad X\rightarrow+\infty,
\ene
where
\bee\no
\phi(X,a)=2X^{\frac12}\Theta(b;a)-p\ln(24ab+4b^{-3})-\frac{p\ln(X)}{2}-2p\ln(2b)+\mathrm{arg}(\Gamma(ip))-2\pi p^2-\frac{\pi}{4}.
\ene

\end{theorem}

\begin{theorem}\label{new-prop} Let $a\rightarrow-1/96=:a_c$ and $\Theta(z;a):=4az^3+z+2z^{-1}$. If we assume that $a-a_c=\mathcal{O}(X^{-\frac13})$, then we have
the solution of Eq.~(\ref{cmkdv-1}) has the large $X$ asymptotics
\bee
\widehat{q}^+(X,aX^2)=-\frac{4\cdot3^{\frac13}|\mathcal{V}_1(-96\cdot3^\frac13(a+\frac{1}{96}) X^{\frac13})|}{X^{\frac{2}{3}}}\sin(\phi(X,a))+\mathcal{O}(X^{-\frac56}),\quad X\rightarrow+\infty,
\ene
where
\bee\no
\phi(X,a)=-\frac{\ln(2)}{6\pi}\ln(X)+\frac{16}{3}X^{\frac12}+64(a+\frac{1}{96})X^{\frac12}+\mathrm{arg}(\mathcal{V}_1(-96\cdot3^\frac13(a+\frac{1}{96}) X^{\frac13}))-\frac{\ln(2)}{\pi}\ln(4\cdot3^{-\frac13}).
\ene

\end{theorem}

The rest of this paper is arranged as follows.
In Sec. 2, we perform a series of deformations on the Riemann-Hilbert problem corresponding to infinite-order rational solitons of the c-mKdV equation (\ref{cmkdv}), and we construct a Riemann-Hilbert problem corresponding to the limit function which is a new solution of the c-mKdV equation (\ref{cmkdv}) in the rescaled variables, $X$ and $T$, and prove the existence and uniqueness of the Riemann-Hilbert problem's solution. In Sec. 3, we analyze the properties of the limit functions $\widehat{q}^{\pm}(X,T)$, including the symmetry and ordinary differential equation it satisfies. In Sec. 4, we study the large $X$ and transitional asymptotic behaviors of $\widehat{q}^{\pm}(X,T)$. And we provide some results under the condition of large $T$. Finally, we give some conclusions in Sec. 5.

\section{Large-order asymptotics of multi-rational solitons under NZB}

Before we study the near-field asymptotics of multi-rational solitons of the c-mKdV equation with nonzero background (\ref{cmkdv}),
we firstly present some two Riemann-Hilbert problems (RHPs), and the multi-rational solutions of the  c-mKdV equation (\ref{cmkdv}) expressed by using the RHP 2.


\begin{prop}\label{prop1} Riemann-Hilbert Problem 1.\\
Let $(x,t)\in\mathbb{R}^2$ and $k\in\mathbb{N}$. Find a $2\times 2$ matrix $M^{(k)}(\lambda;x,t)$ that satisfies the following properties:

\begin{itemize}

 \item {} Analyticity: $M^{(k)}(\lambda; x,t)$ is analytic in $\{\lambda|\lambda\in\mathbb{C}\setminus(\Sigma_c\cup\Sigma_0)\}$ and takes
continuous boundary values on $\Sigma_c\cup\Sigma_0$.

 \item {} Jump condition: The boundary values on the jump contour $\Sigma_c\cup\Sigma_0$ are defined as
 \bee
M^{(k)}_+(\lambda; x,t)=M^{(k)}_-(\lambda; x,t)V^{(k)}(\lambda; x,t),\quad \lambda\in\Sigma_c\cup\Sigma_0,
\ene
where the jump matrix is defined as
\bee \no
V^{(k)}
=\l\{
\begin{array}{ll}
\!\!\! e^{2i\rho_+(\lambda)(x+(4\lambda^2-2)t)\sigma_3},& \lambda\in\Sigma_c, \v\\
\!\!\! e^{-i\rho(\lambda)(x+(4\lambda^2-2)t)\sigma_3}Q\left(\dfrac{\lambda-i}{\lambda+i}\right)^{n\sigma_3}
Q^{-1}E(\lambda)e^{i\rho(\lambda)(x+(4\lambda^2-2)t)\sigma_3},& \lambda\in\Sigma_0,\,\, k=2n,\, n\in\mathbb{N}, \v\\
\!\!\! e^{-i\rho(\lambda)(x+(4\lambda^2-2)t)\sigma_3}Q\left(\dfrac{\lambda+i}{\lambda-i}\right)^{n\sigma_3}Q^{-1}E(\lambda)e^{i\rho(\lambda)(x+(4\lambda^2-2)t)\sigma_3},&\lambda\in\Sigma_0,
\,\, k=2n-1,\, n\in\mathbb{N}_+.
\end{array}\r.
\ene

 \item {} Normalization: $M^{(k)}(\lambda; x,t)$ tends to the identity matrix as $\lambda\rightarrow\infty$.
\end{itemize}
\end{prop}

Then, the $k$th-order rational soliton solution of the c-mKdV equation (\ref{cmkdv}) with with finite density initial conditions (\ref{id}) can be given by the formula:
\bee\label{fanyan}
q(x,t)=q_k(x,t)=2i\lim_{\lambda\rightarrow\infty}\lambda M_{12}^{(k)}(\lambda; x,t),\quad k\in\mathbb{N}.
\ene

In particular, if $k=0$, then the solution of Riemann-Hilbert Problem 1 is
\begin{align}
\begin{aligned}
M^{(0)}(\lambda;x,t)
=\begin{cases}
E(\lambda),\quad \mathrm{if}~\lambda~\mathrm{exterior}~\mathrm{to}~\Sigma_0,\vspace{0.05in}\\
E(\lambda)e^{-i\rho(\lambda)(x+(4\lambda^2-2)t)\widehat\sigma_3}E(\lambda)^{-1},
\quad \mathrm{if}~\lambda~\mathrm{in}~\mathrm{the}~\mathrm{interior}~\mathrm{of}~\Sigma_0.
\end{cases}
\end{aligned}
\end{align}

Let
\begin{align}\label{E-hat}
\begin{aligned}
\widehat{E}(\lambda;x,t)&:=E(\lambda)e^{-i\rho(\lambda)(x+(4\lambda^2-2)t)\sigma_3}E(\lambda)^{-1}\vspace{0.05in}\\
&=\frac{(x+(4\lambda^2-2)t)\sin(\theta)}{\theta}\left[\!\!\begin{array}{cc}
-i\lambda& 1  \vspace{0.05in}\\
-1 & i\lambda
\end{array}\!\!\right]+\cos(\theta)\mathbb{I},
\end{aligned}
\end{align}
with $\theta=\rho(\lambda)(x+(4\lambda^2-2)t)$. Then according to the formula (\ref{fanyan}), we have
\bee\no
q_0(x,t)=2i\lim_{\lambda\rightarrow\infty}\lambda M_{12}^{(0)}(\lambda; x,t)=2i\lim_{\lambda\rightarrow\infty}\lambda E_{12}(\lambda)=1.
\ene

The Riemann-Hilbert Problem 1 without jump across $\Sigma_c$ is easily formulated. To this end, we will make the following transformation:
\bee\label{M-hat}
\widehat{M}^{(k)}(\lambda;x,t):=M^{(k)}(\lambda;x,t)M^{(0)}(\lambda;x,t)^{-1}.
\ene

\begin{prop}\label{prop2} Riemann-Hilbert Problem 2.\\
Let $(x,t)\in\mathbb{R}$ and $k\in\mathbb{N}$. Find a $2\times 2$ matrix $\widehat{M}^{(k)}(\lambda;x,t)$ defined by Eq.~(\ref{M-hat}) that satisfies the following properties:

\begin{itemize}

\item{} Analyticity: $\widehat{M}^{(k)}(\lambda; x,t)$ is analytic in $\{\lambda|\lambda\in\mathbb{C}\setminus\Sigma_0\}$ and takes
continuous boundary values on $\Sigma_0$ from both the exterior and interior.

\item{} Jump conditions: The boundary values on the jump contour $\Sigma_0$ are defined as:\\
if $k=2n,\,n\in\mathbb{N}$, then
\begin{align}
\begin{aligned}
\widehat{M}^{(k)}_+(\lambda; x,t)=&\widehat{M}^{(k)}_-(\lambda; x,t)\widehat{E}(\lambda; x,t)Q\left(\frac{\lambda-i}{\lambda+i}\right)^{n\sigma_3}Q^{-1}\widehat{E}(\lambda; x,t)^{-1},\quad \lambda\in\Sigma_0,
\end{aligned}
\end{align}
and if $k=2n-1,\, n\in\mathbb{N}_+$, then
\begin{align}
\begin{aligned}
\widehat{M}^{(k)}_+(\lambda; x,t)=&\widehat{M}^{(k)}_-(\lambda; x,t)\widehat{E}(\lambda; x,t)Q\left(\frac{\lambda+i}{\lambda-i}\right)^{n\sigma_3}Q^{-1} \widehat{E}(\lambda; x,t)^{-1},\quad \lambda\in\Sigma_0,
\end{aligned}
\end{align}
where the matrix $\widehat{E}(\lambda; x,t)$ is given by Eq.~(\ref{E-hat}).

\item{} Normalization: $\widehat{M}^{(k)}(\lambda; x,t)$ tends to the identity matrix  $\mathbb{I}$ as $\lambda\rightarrow\infty$.
\end{itemize}

\end{prop}

\begin{proof}
If $\lambda\in\Sigma_c$, it follows from Eq.~(\ref{M-hat}) that we have
\begin{align}
\begin{aligned}
\widehat{M}^{(k)}_+(\lambda;x,t)&=M^{(k)}_+(\lambda;,x,t)M^{(0)}_+(\lambda;,x,t)^{-1}\vspace{0.05in}\\
&=M^{(k)}_-(\lambda; x,t)e^{2i\rho_+(\lambda)(x+(4\lambda^2-2)t)\sigma_3}\l(M^{(0)}_-(\lambda; x,t)e^{2i\rho_+(\lambda)(x+(4\lambda^2-2)t)\sigma_3}\r)^{-1}\vspace{0.05in}\\
&=M^{(k)}_-(\lambda; x,t)M^{(0)}_-(\lambda; x,t)^{-1}\vspace{0.05in}\\
&=\widehat{M}^{(k)}_-(\lambda;x,t)\mathbb{I}.
\end{aligned}
\end{align}
This indicates that $\widehat{M}^{(k)}(\lambda;x,t)$ has no jump across $\Sigma_c$.
According to the Riemann-Hilbert Problem 1 and Eq.~(\ref{M-hat}), we know that $\widehat{M}^{(k)}(\lambda; x,t)$ is analytic in $\{\lambda|\lambda\in\mathbb{C}\setminus\Sigma_0\}$ and takes
continuous boundary values on $\Sigma_0$ from both the exterior and interior.

If $\lambda\in\Sigma_0,k=2n,n\in\mathbb{N}$, we have
\begin{align}
\begin{aligned}
\widehat{M}^{(k)}_+(\lambda;x,t)=&M^{(k)}_+(\lambda;,x,t)M^{(0)}_+(\lambda;,x,t)^{-1}\vspace{0.05in}\\
=&M^{(k)}_-(\lambda; x,t)M^{(0)}_-(\lambda; x,t)^{-1}M^{(0)}_-(\lambda; x,t)e^{-i\rho(\lambda)(x+(4\lambda^2-2)t)\sigma_3}\vspace{0.05in}\\
&\times Q \left(\frac{\lambda-i}{\lambda+i}\right)^{n\sigma_3}Q^{-1}E(\lambda)e^{i\rho(\lambda)(x+(4\lambda^2-2)t)\sigma_3}\vspace{0.05in}\\
&\times \bigg(M^{(0)}_-(\lambda; x,t)e^{-i\rho(\lambda)(x+(4\lambda^2-2)t)\sigma_3}E(\lambda)e^{i\rho(\lambda)(x+(4\lambda^2-2)t)\sigma_3}\bigg)^{-1}\vspace{0.05in}\\
=&\widehat{M}^{(k)}_-(\lambda; x,t)\widehat{E}(\lambda; x,t)Q\left(\frac{\lambda-i}{\lambda+i}\right)^{n\sigma_3}Q^{-1}\widehat{E}(\lambda; x,t)^{-1}.
\end{aligned}
\end{align}

If $\lambda\in\Sigma_0,k=2n-1,n\in\mathbb{N}_+$, then we have
\begin{align}
\begin{aligned}
\widehat{M}^{(k)}_+(\lambda;x,t)=&M^{(k)}_+(\lambda;,x,t)M^{(0)}_+(\lambda;,x,t)^{-1}\vspace{0.05in}\\
=&M^{(k)}_-(\lambda; x,t)M^{(0)}_-(\lambda; x,t)^{-1}M^{(0)}_-(\lambda; x,t)e^{-i\rho(\lambda)(x+(4\lambda^2-2)t)\sigma_3}\vspace{0.05in}\\
&\times Q \left(\frac{\lambda+i}{\lambda-i}\right)^{n\sigma_3}Q^{-1}E(\lambda)e^{i\rho(\lambda)(x+(4\lambda^2-2)t)\sigma_3}\vspace{0.05in}\\
&\times \bigg(M^{(0)}_-(\lambda; x,t)e^{-i\rho(\lambda)(x+(4\lambda^2-2)t)\sigma_3}E(\lambda)e^{i\rho(\lambda)(x+(4\lambda^2-2)t)\sigma_3}\bigg)^{-1}\vspace{0.05in}\\
=&\widehat{M}^{(k)}_-(\lambda; x,t)\widehat{E}(\lambda; x,t)Q\left(\frac{\lambda+i}{\lambda-i}\right)^{n\sigma_3}Q^{-1}\widehat{E}(\lambda; x,t)^{-1}.
\end{aligned}
\end{align}

Finally, according to Eq.~(\ref{M-hat}), we know that
$\widehat{M}^{(k)}(\lambda; x,t)\to \mathbb{I}$ as $\lambda\rightarrow\infty$.
Thus the proof is completed.
\end{proof}

\begin{prop}\label{prop3} The solution of Eq.~(\ref{cmkdv}) with finite density initial conditions (\ref{id}) can be expressed as:
\bee\label{qk-1}
q_k(x,t)=1+2i\lim_{\lambda\rightarrow\infty}\lambda \widehat{M}_{12}^{(k)}(\lambda; x,t),\quad k\in\mathbb{N},
\ene
where $\widehat{M}^{(k)}(\lambda; x,t)$ satisfies the Riemann-Hilbert Problem 2.
\end{prop}

\begin{proof}
According to Eq.~(\ref{fanyan}), we obtain
\begin{align}
\begin{aligned}
q_k(x,t)&=2i\lim_{\lambda\rightarrow\infty}\lambda M_{12}^{(k)}(\lambda; x,t),\quad k\in\mathbb{N}\vspace{0.05in}\\
&=2i\lim_{\lambda\rightarrow\infty}\lambda \left(\widehat{M}^{(k)}(\lambda; x,t)M^{(0)}(\lambda;,x,t)\right)_{12}\vspace{0.05in}\\
&=2i\lim_{\lambda\rightarrow\infty}\lambda \left(\widehat{M}_{11}^{(k)}(\lambda; x,t)M^{(0)}_{12}(\lambda;,x,t)+
\widehat{M}_{12}^{(k)}(\lambda; x,t)M^{(0)}_{22}(\lambda;,x,t)\right)\vspace{0.05in}\\
&=1+2i\lim_{\lambda\rightarrow\infty}\lambda \widehat{M}_{12}^{(k)}(\lambda; x,t).
\end{aligned}
\end{align}
Thus the proof is completed.
\end{proof}

Now we consider the near-field asymptotics of rational solitons of the c-mKdV equation (\ref{cmkdv}) with finite density initial conditions  (\ref{id}). We will perform the following scale transforms for the Riemann-Hilbert Problem 2:
\bee\label{trans}
X:=nx,\qquad T:=n^3t,\qquad \Lambda:=\frac{\lambda}{n}.
\ene
We choose a radius of $n$ for the contour $\Sigma_0$, then the jump matrix possesses the following asymptotic behavior:
\begin{align}
\begin{aligned}
&\widehat{E}(\lambda; x,t)Q\left(\dfrac{\lambda-i}{\lambda+i}\right)^{\pm n\sigma_3}Q^{-1}\widehat{E}(\lambda; x,t)^{-1}\vspace{0.05in}\\
=~&E(\lambda)e^{-i\rho(\lambda)(x+(4\lambda^2-2)t)\sigma_3}E(\lambda)^{-1}Q
\left(\dfrac{\lambda-i}{\lambda+i}\right)^{\pm n\sigma_3}Q^{-1}
E(\lambda)e^{i\rho(\lambda)(x+(4\lambda^2-2)t)\sigma_3}E(\lambda)^{-1}\vspace{0.05in}\\
=~&\left(\mathbb{I}+\mathcal{O}\left(\frac{1}{n}\right)\right)e^{-i(\Lambda X+4\Lambda^3T)\sigma_3}Qe^{\mp2i\Lambda^{-1}\sigma_3}Q^{-1}e^{i(\Lambda X+4\Lambda^3T)\sigma_3},\quad n\to \infty,
\end{aligned}
\end{align}
which holds uniformly for $|\Lambda|=1$ and $(X,T)$ in compact subsets of $\mathbb{R}^2$.

\begin{prop}\label{prop4} Riemann-Hilbert Problem 3.\\
Let $(X,T)$ be in compact subsets of $\mathbb{R}^2$. Find a $2\times 2$ matrix $N^{\pm}(\Lambda;X,T)$ that satisfies the following properties:

\begin{itemize}
    \item{} Analyticity: $N^{\pm}(\Lambda; X,T)$ is analytic in $\mathbb{C}\setminus\Sigma_1:=\{\Lambda|~|\Lambda|\neq1,\Lambda\in\mathbb{C}\}$ and takes
continuous boundary values on $\Sigma_1:=\{\Lambda|~|\Lambda|=1,\Lambda\in\mathbb{C}\}$.

\item{} Jump condition: Assuming clockwise orientation of $\Sigma_1$, the boundary values on the jump contour $\Sigma_1$ are related as:
\begin{align}
\begin{aligned}
N^{\pm}_+(\Lambda; X,T)=&N^{\pm}_-(\Lambda; X,T)e^{-i(\Lambda X+4\Lambda^3T)\sigma_3}Qe^{\mp2i\Lambda^{-1}\sigma_3}Q^{-1}e^{i(\Lambda X+4\Lambda^3T)\sigma_3},\quad \Lambda\in\Sigma_1,
\end{aligned}
\end{align}

\item{} Normalization: $N^{\pm}(\Lambda; X,T)$ tends to the identity matrix $\mathbb{I}$ as $\Lambda\rightarrow\infty$.
\end{itemize}
\end{prop}

Then, we will solve the Riemann-Hilbert Problem 3 and provide the proof of Theorem \ref{prop5}.

\begin{proof}
To show solvability and uniqueness, we will prove that the jump matrices and the jump conditions in Riemann-Hilbert problem 3 satisfy the hypotheses of Zhou's vanishing lemma~\cite{zhou1989-siam}. For this reason, we reorient the jump contour $\Sigma_1$ to the upper half-plane in a clockwise orientation and the lower half-plane in a counterclockwise orientation. Then, we define the following piecewise function:
\begin{align}
\begin{aligned}
&V^{\pm}(\Lambda;X,T)=\begin{cases}
e^{-i(\Lambda X+4\Lambda^3T)\sigma_3}Qe^{\mp2i\Lambda^{-1}\sigma_3}Q^{-1}e^{i(\Lambda X+4\Lambda^3T)\sigma_3},\quad |\Lambda|=1,\,\,\mathrm{Im}(\Lambda)>0,\vspace{0.05in}\\
e^{-i(\Lambda X+4\Lambda^3T)\sigma_3}Qe^{\pm2i\Lambda^{-1}\sigma_3}Q^{-1}e^{i(\Lambda X+4\Lambda^3T)\sigma_3},\quad |\Lambda|=1,\,\,\mathrm{Im}(\Lambda)<0.
\end{cases}
\end{aligned}
\end{align}

For $|\Lambda|=1,\mathrm{Im}(\Lambda)>0$, since $Q^{-1}=Q^{T}$, thus we have
\begin{align}
\begin{aligned}
V^{\pm}(\Lambda;X,T)&=e^{-i(\Lambda X+4\Lambda^3T)\sigma_3}Qe^{\mp2i\Lambda^{-1}\sigma_3}Q^{-1}e^{i(\Lambda X+4\Lambda^3T)\sigma_3}\vspace{0.05in}\\
&=\left(e^{i(\Lambda^* X+4\Lambda^{*3}T)\sigma_3}Qe^{\pm2i\Lambda^{*-1}\sigma_3}Q^{-1}e^{-i(\Lambda^* X+4\Lambda^{*3}T)\sigma_3}\right)^*\vspace{0.05in}\\
&=\left(e^{-i(\Lambda^* X+4\Lambda^{*3}T)\sigma_3}Qe^{\pm2i\Lambda^{*-1}\sigma_3}Q^{-1}e^{i(\Lambda^* X+4\Lambda^{*3}T)\sigma_3}\right)^{\dag}\vspace{0.05in}\\
&=\l(V^{\pm}(\Lambda^*;X,T)\r)^{\dag},
\end{aligned}
\end{align}
where the superscript $\dag$ denotes the conjugate transpose of the matrix. Considering the normalization condition $N^{\pm}(\Lambda;X,T)\rightarrow\mathbb{I}$ as $\Lambda\rightarrow\infty$, we have satisfied all the hypotheses of the vanishing lemma. Consequently, the Riemann-Hilbert problem 3 is uniquely solvable. The matrix $N^{\pm}(\Lambda;X,T)$ has the following Laurent expansion.
\bee\label{N}
N^{\pm}(\Lambda;X,T)=\mathbb{I}+\sum_{j=1}^{\infty}N^{\pm[j]}(X,T)\Lambda^{-j},\quad |\Lambda|>1,
\ene
and analytic Fredholm theory implies that each $N^{\pm[j]}(X,T)$ is real analytic on $\mathbb{R}^2$. We will show that $\widehat{q}^{\pm}(X,T)$ is a global solution of the complex modified Korteweg-de Vries equation. To this end, we define
\bee\no
P^{\pm}(\Lambda;X,T):=N^{\pm}(\Lambda;X,T)e^{-i(\Lambda X+4\Lambda^3T)\sigma_3}.
\ene
Then, $P^{\pm}(\Lambda;X,T)$ satisfies the following Riemann-Hilbert problem:

\begin{itemize}
\item{} Analyticity: $P^{\pm}(\lambda; X,T)$ is analytic in $\mathbb{C}\setminus\Sigma_1$ and takes continuous boundary values on $\Sigma_1$.

\item{} Jump conditions: Assuming clockwise orientation of $\Sigma_1$, the boundary values on the jump contour $\Sigma_1$ are related as:
\bee
P^{\pm}_+(\Lambda; X,T)=P^{\pm}_-(\Lambda; X,T)Qe^{\mp2i\Lambda^{-1}\sigma_3}Q^{-1},\quad \Lambda\in\Sigma_1,
\ene

\item{} Normalization: $P^{\pm}(\Lambda; X,T)$ tends to $e^{-i(\Lambda X+4\Lambda^3T)}$ as $\Lambda\rightarrow\infty$.
\end{itemize}

We define two matrices:
\begin{align}\label{N-lax0}
\left\{\begin{aligned}
&A^{\pm}(\Lambda; X,T):=P^{\pm}_X(\Lambda; X,T)P^{\pm}(\Lambda; X,T)^{-1},\vspace{0.05in}\\
&B^{\pm}(\Lambda; X,T):=P^{\pm}_T(\Lambda; X,T)P^{\pm}(\Lambda; X,T)^{-1}.
\end{aligned}\right.
\end{align}

They can be defined by continuity for $\Sigma_1$ and they are entire functions of $\Lambda$. Firstly, we provide the expression for $N^{\pm}(\Lambda;X,T)^{-1}$.
\begin{align}\label{N-1}
\begin{aligned}
N^{\pm}(\Lambda;X,T)^{-1}=\,&\mathbb{I}-N^{\pm[1]}(X,T)\Lambda^{-1}+(N^{\pm[1]}(X,T)^2-N^{\pm[2]}(X,T))\Lambda^{-2}\vspace{0.05in}\\
&+(N^{\pm[1]}(X,T)N^{\pm[2]}(X,T)+N^{\pm[2]}(X,T)N^{\pm[1]}(X,T)\vspace{0.05in}\\
&-N^{\pm[3]}(X,T)-N^{\pm[1]}(X,T)^3)\Lambda^{-3}+\mathcal{O}(\Lambda^{-4}), \quad \Lambda\rightarrow\infty.
\end{aligned}
\end{align}

According to Eqs.~(\ref{N}), (\ref{N-lax0}) and (\ref{N-1}), we have
\bee\label{A}
A^{\pm}(\Lambda;X,T)=-i\Lambda\sigma_3+i[\sigma_3,N^{\pm[1]}(X,T)]+\mathcal{O}(\Lambda^{-1}),\quad \Lambda\rightarrow\infty,
\ene
and
\begin{align}\label{B}
\begin{aligned}
B^{\pm}(\Lambda;X,T)=&-4i\Lambda^3\sigma_3+4i\Lambda^2[\sigma_3,N^{\pm[1]}(X,T)]+4i\Lambda\bigg([\sigma_3.N^{\pm[2]}(X,T)]\vspace{0.05in}\\
&+[N^{\pm[1]}(X,T),\sigma_3N^{\pm[1]}(X,T)])+4i\bigg([\sigma_3,N^{\pm[3]}(X,T)]+[N^{\pm[2]}(X,T),\sigma_3N^{\pm[1]}(X,T)]\vspace{0.05in}\\
&+[\sigma_3N^{\pm[1]}(X,T)^2,N^{\pm[1]}(X,T)]+[N^{\pm[1]}(X,T),\sigma_3N^{\pm[2]}(X,T)]\bigg)+\mathcal{O}(\Lambda^{-1}),\quad \Lambda\rightarrow\infty.
\end{aligned}
\end{align}
Using Liouville's theorem, we get
\begin{align}\label{AB1}
\begin{aligned}
A^{\pm}(\Lambda;X,T)=&-i\Lambda\sigma_3+i[\sigma_3,N^{\pm[1]}(X,T)],\vspace{0.05in}\\
B^{\pm}(\Lambda;X,T)=&-4i\Lambda^3\sigma_3+4i\Lambda^2[\sigma_3,N^{\pm[1]}(X,T)]+4i\Lambda\bigg([\sigma_3,N^{\pm[2]}(X,T)]\vspace{0.05in}\\
&+[N^{\pm[1]}(X,T),\sigma_3N^{\pm[1]}(X,T)]\bigg)+4i\bigg([\sigma_3,N^{\pm[3]}(X,T)]+[N^{\pm[2]}(X,T),\sigma_3N^{\pm[1]}(X,T)]\vspace{0.05in}\\
&+[\sigma_3N^{\pm[1]}(X,T)^2,N^{\pm[1]}(X,T)]+[N^{\pm[1]}(X,T),\sigma_3N^{\pm[2]}(X,T)]\bigg).
\end{aligned}
\end{align}

To eliminate matrices $N^{\pm[3]}(X,T)])$, we provide the error terms for $\Lambda^{-1}$ and $\Lambda^{-2}$ in Eq.~(\ref{A}).
\bee\label{Lambda-1}
P_{\Lambda^{-1}}:\quad &i[\sigma_3,N^{\pm[2]}(X,T)]+i[N^{\pm[1]}(X,T),\sigma_3N^{\pm[1]}(X,T)]+N_X^{\pm[1]}(X,T),
\ene
and
\begin{align}\label{Lambda-2}
\begin{aligned}
P_{\Lambda^{-2}}:\quad &i[\sigma_3,N^{\pm[3]}(X,T)]+i[N^{\pm[2]}(X,T),\sigma_3N^{\pm[1]}(X,T)]+i[\sigma_3N^{\pm[1]}(X,T)^2,N^{\pm[1]}(X,T)]\vspace{0.05in}\\
&+i[N^{\pm[1]}(X,T),\sigma_3N^{\pm[2]}(X,T)]+N_X^{\pm[2]}(X,T)-N_X^{\pm[1]}(X,T)N^{\pm[1]}(X,T).
\end{aligned}
\end{align}

Because of the existence and uniqueness of the solution to the Riemann-Hilbert problem 3, we have the following symmetric relation:
\bee\label{sym}
N^{\pm}(\Lambda;X,T)=\sigma_2N^{\pm}(\Lambda^*;X,T)^*\sigma_2.
\ene

Note that $\widehat{q}^{\pm}(X,T)=2iN^{\pm[1]}_{12}(X,T)$ and solving $\{P_{\Lambda^{-1},11}=0,\, P_{\Lambda^{-1},22}=0\}$, we obtain
\begin{align}\label{N-lax}
\left\{\begin{aligned}
&N^{\pm[1]}_{11,X}(X,T)=-\frac{i}{2}|\widehat{q}^{\pm}(X,T)|^2,\vspace{0.05in}\\
&N^{\pm[1]}_{22,X}(X,T)=\frac{i}{2}|\widehat{q}^{\pm}(X,T)|^2.
\end{aligned}\right.
\end{align}

Using Eqs.~(\ref{B}),(\ref{Lambda-1}) and(\ref{N-lax}), we have
\begin{align}\label{B-new1}
\begin{aligned}
B^{\pm}(\Lambda;X,T)=&-4i\Lambda^3\sigma_3+4i\Lambda^2[\sigma_3,N^{\pm[1]}(X,T)]-4\Lambda N^{\pm[1]}_X(X,T)\vspace{0.05in}\\
&+4i\bigg([\sigma_3,N^{\pm[3]}(X,T)]+[N^{\pm[2]}(X,T),\sigma_3N^{\pm[1]}(X,T)]\vspace{0.05in}\\
&+[\sigma_3N^{\pm[1]}(X,T)^2,N^{\pm[1]}(X,T)]+[N^{\pm[1]}(X,T),\sigma_3N^{\pm[2]}(X,T)\bigg).
\end{aligned}
\end{align}

To solve $\{P_{\Lambda^{-1},12}=0,\, P_{\Lambda^{-1},21}=0\}$, we have
\begin{align}\label{N2}
\begin{aligned}
&N^{\pm[2]}_{12}(X,T)=N^{\pm[1]}_{12}(X,T)N^{\pm[1]}_{22}(X,T)+\frac{i}{2}N^{\pm[1]}_{12,X}(X,T),\vspace{0.05in}\\
&N^{\pm[2]}_{21}(X,T)=N^{\pm[1]}_{21}(X,T)N^{\pm[1]}_{11}(X,T)-\frac{i}{2}N^{\pm[1]}_{21,X}(X,T).
\end{aligned}
\end{align}

Let
\begin{align}\no
\begin{aligned}
M^0:=&[\sigma_3,N^{\pm[3]}(X,T)]+[N^{\pm[2]}(X,T),\sigma_3N^{\pm[1]}(X,T)]+[\sigma_3N^{\pm[1]}(X,T)^2,N^{\pm[1]}(X,T)]\vspace{0.05in}\\
&+[N^{\pm[1]}(X,T),\sigma_3N^{\pm[2]}(X,T)].
\end{aligned}
\end{align}

Extracting the elements of matrix $M^0$ at the diagonal, we have
\begin{align}\label{M-diag}
\begin{aligned}
M^0_{11}=&2N^{\pm[1]}_{12}(X,T)(N^{\pm[1]}_{21}(X,T)N^{\pm[1]}_{11}(X,T)+N^{\pm[1]}_{22}(X,T)N^{\pm[1]}_{21}(X,T))\vspace{0.05in}\\
&-N^{\pm[2]}_{12}(X,T)N^{\pm[1]}_{21}(X,T)-N^{\pm[1]}_{12}(X,T)N^{\pm[2]}_{21}(X,T)\vspace{0.05in}\\
&-N^{\pm[1]}_{12}(X,T)N^{\pm[2]}_{21}(X,T)-N^{\pm[2]}_{12}(X,T)N^{\pm[1]}_{21}(X,T),\vspace{0.05in}\\
M^0_{22}=&-2N^{\pm[1]}_{21}(X,T)(N^{\pm[1]}_{11}(X,T)N^{\pm[1]}_{12}(X,T)+N^{\pm[1]}_{12}(X,T)N^{\pm[1]}_{22}(X,T))\vspace{0.05in}\\
&+N^{\pm[2]}_{21}(X,T)N^{\pm[1]}_{12}(X,T)+N^{\pm[1]}_{21}(X,T)N^{\pm[2]}_{12}(X,T)\vspace{0.05in}\\
&+N^{\pm[1]}_{21}(X,T)N^{\pm[2]}_{12}(X,T)+N^{\pm[2]}_{21}(X,T)N^{\pm[1]}_{12}(X,T).
\end{aligned}
\end{align}

Substituting Eqs.~(\ref{N-lax}) and (\ref{N2}) into Eq.~(\ref{M-diag}), we obtain
\begin{align}\label{M-diag-1}
\begin{aligned}
&M^0_{11}=-\frac{i}{4}\widehat{q}^{\pm*}_X\widehat{q}^{\pm}
+\frac{i}{4}\widehat{q}^{\pm}_X\widehat{q}^{\pm*},\vspace{0.05in}\\
&M^0_{22}=\frac{i}{4}\widehat{q}^{\pm*}_X\widehat{q}^{\pm}-\frac{i}{4}\widehat{q}^{\pm}_X\widehat{q}^{\pm*}.
\end{aligned}
\end{align}

Solving $\{P_{\Lambda^{-2},12}=0,\,P_{\Lambda^{-2},21}=0\}$, we obtain
\begin{align}\label{M-offdiag}
\begin{aligned}
&M^0_{12}=-iN^{\pm[1]}_{11,X}(X,T)N^{\pm[1]}_{12}(X,T)-iN^{\pm[1]}_{12,X}(X,T)N^{\pm[1]}_{22}(X,T)+iN^{\pm[2]}_{12,X}(X,T),\vspace{0.05in}\\
&M^0_{21}=-iN^{\pm[1]}_{21,X}(X,T)N^{\pm[1]}_{11}(X,T)-iN^{\pm[1]}_{22,X}(X,T)N^{\pm[1]}_{21}(X,T)+iN^{\pm[2]}_{21,X}(X,T).
\end{aligned}
\end{align}

Substituting Eqs.~(\ref{N-lax}) and (\ref{N2}) into Eq.~(\ref{M-offdiag}), we obtain
\begin{align}\label{M-offdiag-1}
\begin{aligned}
&M^0_{12}=\frac{i}{4}\widehat{q}^{\pm}_{XX}+\frac{i}{2}|\widehat{q}^{\pm}|^2\widehat{q}^{\pm},\vspace{0.05in}\\
&M^0_{21}=-\frac{i}{4}\widehat{q}^{\pm*}_{XX}-\frac{i}{2}|\widehat{q}^{\pm}|^2\widehat{q}^{\pm*}.
\end{aligned}
\end{align}

Substituting Eqs.~(\ref{N-lax}), (\ref{N2}), (\ref{M-diag-1}), and (\ref{M-offdiag-1}) into Eq.~(\ref{AB1}), we have
\begin{align}\label{N-lax-1}
\begin{aligned}
&A^{\pm}(\Lambda;X,T)=\left[\!\!\begin{array}{cc}
-i\Lambda& \widehat{q}^{\pm}(X,T)  \vspace{0.05in}\\
-\widehat{q}^{\pm}(X,T)^* & i\Lambda
\end{array}\!\!\right],\vspace{0.05in}\\
&B^{\pm}(\Lambda;X,T)=\left[\!\!\begin{array}{cc}
-4i\Lambda^3+2i\Lambda|\widehat{q}^{\pm}|^2+\widehat{q}_X^{\pm*}\widehat{q}^{\pm}-\widehat{q}^{\pm}_X\widehat{q}^{\pm*}& 4\Lambda^2\widehat{q}^{\pm}+2i\Lambda \widehat{q}^{\pm}_X-\widehat{q}^{\pm}_{XX}-2|\widehat{q}^{\pm}|^2\widehat{q}^{\pm}  \vspace{0.05in}\\
-4\Lambda^2\widehat{q}^{\pm*}+2i\Lambda \widehat{q}_X^{\pm*}+\widehat{q}_{XX}^{\pm*}+2|\widehat{q}^{\pm}|^2\widehat{q}^{\pm*} & 4i\Lambda^3-2i\Lambda|\widehat{q}^{\pm}|^2-\widehat{q}_X^{\pm*}\widehat{q}^{\pm}+\widehat{q}^{\pm}_X\widehat{q}^{\pm*}
\end{array}\!\!\right].
\end{aligned}
\end{align}

Using Eqs.~(\ref{N-lax}) and (\ref{N-lax-1}), we get the system of equations about the matrix functions $P^{\pm}(\Lambda; X,T)$:
\begin{align}\label{ab-26-1}
\begin{aligned}
&P_X^{\pm}(\Lambda; X,T)=A^{\pm}(\Lambda;X,T)P^{\pm}(\Lambda; X,T),\vspace{0.05in}\\
&P_T^{\pm}(\Lambda; X,T)=B^{\pm}(\Lambda;X,T)P^{\pm}(\Lambda; X,T),
\end{aligned}
\end{align}
which constitute the Lax pair for the c-mKdV equation. According to compatibility condition $A_T^{\pm}(\Lambda;X,T)-B_X^{\pm}(\Lambda;X,T)+[A^{\pm}(\Lambda;X,T),B^{\pm}(\Lambda;X,T)]=0$, we show that the function $q^{\pm}(X,T)$ is a solution of Eq.~(\ref{cmkdv-1}). Thus the proof is completed.

\end{proof}

Below we present the large-order asymptotic behavior of rational soliton solution $q_k(x,t)$ of the c-mKdV equation (\ref{cmkdv}) with finite density initial conditions (\ref{id}) after making a transformation (\ref{trans}). Our main result is given by the following Theorem. Firstly, we present a small-norm Riemann-Hilbert Problem.

Let
\bee
F(\Lambda;X,T):=\widehat{M}^{(k)}\l(n\Lambda;\frac{X}{n},\frac{T}{n^3}\r)N^{\pm}(\Lambda;X,T)^{-1},
\ene
where if $k=2n,n\in\mathbb{N}$, we choose the $+$ sign, and if $k=2n-1,n\in\mathbb{N}_+$, we choose the $-$ sign.

\begin{lemma}\label{le1} Small-norm Riemann-Hilbert Problem.\\
Let $(X,T)$ be in compact subsets of $\mathbb{R}^2$. Find a $2\times 2$ matrix $F(\Lambda;X,T)$ that satisfies the following properties:

\begin{itemize}
  \item{}
 Analyticity: $F(\Lambda;X,T)$ is analytic in $\mathbb{C}\setminus\Sigma_1$ and takes continuous boundary values on $\Sigma_1$.

\item{} Jump conditions: Assuming clockwise orientation of $\Sigma_1$, the boundary values on the jump contour $\Sigma_1$ are related as:
\bee
F_+(\Lambda;X,T)=F_-(\Lambda;X,T)N_-^{\pm}(\Lambda;X,T)\l(\mathbb{I}+\mathcal{O}\l(\frac{1}{n}\r)\r)N_-^{\pm}(\Lambda;X,T)^{-1},\quad \Lambda\in\Sigma_1,
\ene

\item{} Normalization: $F(\Lambda;X,T)$ tends to the identity matrix as $\Lambda\rightarrow\infty$.
\end{itemize}
\end{lemma}

\begin{proof}
If $\Lambda\in\Sigma_1$, we have
\begin{align}\no
\begin{aligned}
F_+(\Lambda;X,T)&=\widehat{M}_+^{(k)}\left(n\Lambda;\frac{X}{n},\frac{T}{n^3}\right)N_+^{\pm}(\Lambda;X,T)^{-1}\vspace{0.05in}\\
&=\widehat{M}^{(k)}_-\l(n\Lambda;\frac{X}{n},\frac{T}{n^3}\r)\l(\mathbb{I}+\mathcal{O}\l(\frac{1}{n}\r)\r)e^{-i(\Lambda X+4\Lambda^3T)\sigma_3}Qe^{\mp2i\Lambda^{-1}\sigma_3}Q^{-1}e^{i(\Lambda X+4\Lambda^3T)\sigma_3}\vspace{0.05in}\\
&\qquad \times \left(N^{\pm}_-(\Lambda; x,t)e^{-i(\Lambda X+4\Lambda^3T)\sigma_3}Qe^{\mp2i\Lambda^{-1}\sigma_3}Q^{-1}e^{i(\Lambda X+4\Lambda^3T)\sigma_3}\right)^{-1}\vspace{0.05in}\\
&=\widehat{M}^{(k)}_-\l(n\Lambda;\frac{X}{n},\frac{T}{n^3}\r)\l(\mathbb{I}+\mathcal{O}\l(\frac{1}{n}\r)\r)N^{\pm}_-(\Lambda; x,t)^{-1}\vspace{0.05in}\\
&=F_-(\Lambda;X,T)N_-^{\pm}(\Lambda;X,T)\l(\mathbb{I}+\mathcal{O}\l(\frac{1}{n}\r)\r)N_-^{\pm}(\Lambda;X,T)^{-1},\quad \Lambda\in\Sigma_1.
\end{aligned}
\end{align}

Thus the proof is completed.
\end{proof}

According to Lemma \ref{le1}, we provide the proof of Theorem \ref{th1}.

\begin{proof}
To prove this theorem, we only need to solve the small-norm Riemann-Hilbert Problem in Lemma \ref{le1}. If $k=2n,n\in\mathbb{N}$, we choose the $+$ sign, and if $k=2n-1,n\in\mathbb{N}_+$, we choose the $-$ sign. Selecting a compact $K\subset\mathbb{R}^2$ and using $\det(N^{\pm}(\Lambda;x,t))=1$ along with Eq.~(\ref{C-k}), we can obtain
\bee
F_+(\Lambda;X,T)=F_-(\Lambda;X,T)\l(\mathbb{I}+\mathcal{O}\l(\frac{1}{n}\r)\r),\quad \Lambda\in\Sigma_1,
\ene
uniformly for (X,T) in compact subsets $K$. According to the standard theory of small-norm Riemann-Hilbert Problem, it follows that
\bee
F(\Lambda;X,T)=\mathbb{I}+\mathcal{O}\l(\frac{1}{n}\r),\quad \Lambda\in\Sigma_1,
\ene
uniformly for $(X,T)\in K$. Moreover, the matrix function $F_(\Lambda;X,T)$ has the following Laurent expansion.
\bee\label{F}
F(\Lambda;X,T)=\mathbb{I}+\sum_{j=1}^{\infty}F^{[j]}(X,T)\Lambda^{-j},\quad |\Lambda|>1,
\ene
uniformly for $(X,T)\in K$. Then, every coefficient $F^{[j]}(X,T)$ of Eq.~(\ref{F}) is $\mathcal{O}(\frac{1}{n})$. Using Eq.~(\ref{qk-1}), we have
\begin{align}\no
\begin{aligned}
\frac{1}{n}q_k\l(\frac{X}{n},\frac{T}{n^3}\r)&=\frac{1}{n}+\frac{2i}{n}\lim_{\lambda\rightarrow\infty}\lambda\widehat{M}^{(k)}_{12}\l(\lambda;\frac{X}{n},\frac{T}{n^3}\r)\vspace{0.05in}\\
&=\frac{1}{n}+2i\lim_{\Lambda\rightarrow\infty}\Lambda\widehat{M}^{(k)}_{12}\l(n\Lambda;\frac{X}{n},\frac{T}{n^3}\r)\vspace{0.05in}\\
&=\frac{1}{n}+2i\lim_{\Lambda\rightarrow\infty}\Lambda\left(F_{11}(\Lambda;X,T)N^{\pm}_{12}(\Lambda;X,T)+F_{12}(\Lambda;X,T)N^{\pm}_{22}(\Lambda;X,T)
\right)\vspace{0.05in}\\
&=\frac{1}{n}+2i\lim_{\Lambda\rightarrow\infty}\Lambda N^{\pm}_{12}(\Lambda;X,T)+2iF^{[1]}_{12}(X,T)\vspace{0.05in}\\
&=\widehat{q}^{\pm}(X,T)+\mathcal{O}\l(\frac{1}{n}\r)
\end{aligned}
\end{align}
holds uniformly for $(X,T)\in K$. Thus the proof is completed.
\end{proof}

\section{Some basic properties of the near-field limit in $(X, T)$-space}

As mentioned earlier, performing a transformation $(\ref{trans})$ can provide the solution $\widehat{q}^{\pm}(X,T)$, satisfying the c-mKdV equation (\ref{cmkdv-1}). To better analyze the properties of $\widehat{q}^{\pm}(X,T)$, let
\bee\label{D}
D^{\pm}(\Lambda;X,T)
=\begin{cases}
N^{\pm}(\Lambda;X,T)e^{-i(\Lambda X+4\Lambda^3T)\sigma_3}Qe^{i(\Lambda X+4\Lambda^3T)\sigma_3},\quad |\Lambda|<1,\vspace{0.05in}\\
N^{\pm}(\Lambda;X,T)e^{\pm2i\Lambda^{-1}\sigma_3},\quad |\Lambda|>1.
\end{cases}
\ene
Then we construct a Riemann-Hilbert Problem about the matrix functions $D^{\pm}(\Lambda;X,T)$.
\begin{prop} Riemann-Hilbert Problem 4.\\
Let $(X,T)$ be in compact subsets of $\mathbb{R}^2$. Find a $2\times 2$ matrix $D^{\pm}(\Lambda;X,T)$ that satisfies the following properties:

\begin{itemize}
    \item{} Analyticity: $D^{\pm}(\Lambda; X,T)$ is analytic in $\mathbb{C}\setminus\Sigma_1$ and takes continuous boundary values on $\Sigma_1$.

\item{} Jump conditions: Assuming clockwise orientation of $\Sigma_1$, the boundary values on the jump contour $\Sigma_1$ are related as:
\bee\label{RH4}
D^{\pm}_+(\Lambda; X,T)=D^{\pm}_-(\Lambda; X,T)e^{-i(\Lambda X+4\Lambda^3T\pm2\Lambda^{-1})\sigma_3}Q^{-1}e^{i(\Lambda X+4\Lambda^3T\pm2\Lambda^{-1})\sigma_3},\quad \Lambda\in\Sigma_1.
\ene

\item{} Normalization: $D^{\pm}(\Lambda; x,t)$ tends to the identity matrix $\mathbb{I}$ as $\Lambda\rightarrow\infty$.
\end{itemize}
\end{prop}

\begin{proof}
If $\Lambda\in\Sigma_1$, we have
\begin{align}\no
\begin{aligned}
D^{\pm}_+(\Lambda;X,T)&=N_+^{\pm}(\Lambda;X,T)e^{\pm2i\Lambda^{-1}\sigma_3}\vspace{0.05in}\\
&=N^{\pm}_-(\Lambda; x,t)e^{-i(\Lambda X+4\Lambda^3T)\sigma_3}Qe^{\mp2i\Lambda^{-1}\sigma_3}Q^{-1}e^{i(\Lambda X+4\Lambda^3T\pm2\Lambda^{-1})\sigma_3}\vspace{0.05in}\\
&=N_-^{\pm}(\Lambda;X,T)e^{-i(\Lambda X+4\Lambda^3T)\sigma_3}Qe^{i(\Lambda X+4\Lambda^3T)\sigma_3}\bigg(e^{-i(\Lambda X+4\Lambda^3T)\sigma_3}Q^{-1}
e^{i(\Lambda X+4\Lambda^3T)\sigma_3}\bigg)\vspace{0.05in}\\
&~~\times e^{-i(\Lambda X+4\Lambda^3T)\sigma_3}Qe^{\mp2i\Lambda^{-1}\sigma_3}Q^{-1}e^{i(\Lambda X+4\Lambda^3T\pm2\Lambda^{-1})\sigma_3}\vspace{0.05in}\\
&=D^{\pm}_-(\Lambda; x,t)e^{-i(\Lambda X+4\Lambda^3T\pm2i\Lambda^{-1})\sigma_3}Q^{-1}e^{i(\Lambda X+4\Lambda^3T\pm2i\Lambda^{-1})\sigma_3},\quad \Lambda\in\Sigma_1.
\end{aligned}
\end{align}

Since $e^{\pm2i\Lambda^{-1}\sigma_3}\rightarrow\mathbb{I},N^{\pm}(\Lambda; x,t)\rightarrow\mathbb{I}$ as $\Lambda\rightarrow\infty$, thus we have $D^{\pm}(\Lambda; x,t)\rightarrow\mathbb{I}$ as $\Lambda\rightarrow\infty$.
Thus the proof is completed.
\end{proof}

According to Eq.~(\ref{fanyan-q1}), we can recover $\widehat{q}^{\pm}(X,T)$ by the following formula:
\bee\label{fanyan-q2}
\widehat{q}^{\pm}(X,T)=2i\lim_{\Lambda\rightarrow\infty}\Lambda D^{\pm}_{12}(\Lambda;X,T).
\ene

\subsection{Symmetries}
$D^+(\Lambda;X,T)$ and $D^-(\Lambda;X,T)$ given by Eq.~(\ref{D}) have the following symmetric relations.

\begin{prop}\label{prop-sch}
Schwarz symmetry:
\bee\label{sch-sym}
D^{\pm}(\Lambda;X,T)=\sigma_2D^{\pm}(\Lambda^*;X,T)^*\sigma_2,\quad |\Lambda|\neq1,\quad (X,T)\in\mathbb{R}^2.
\ene
\end{prop}

\begin{proof}
Using Eq.~(\ref{RH4}), we have
\begin{align}\label{R-sym-1}
\begin{aligned}
\sigma_2D^{\pm}_+(\Lambda^*; X,T)^*\sigma_2=&\sigma_2D^{\pm}_-(\Lambda^*; X,T)^*\sigma_2\sigma_2e^{i(\Lambda X+4\Lambda^3T\pm2\Lambda^{-1})\sigma_3}Q^{-1}\vspace{0.05in}\\
&\times e^{-i(\Lambda X+4\Lambda^3T\pm2\Lambda^{-1})\sigma_3}\sigma_2,\quad \Lambda\in\Sigma_1.
\end{aligned}
\end{align}

Using $\sigma_2Q^{-1}\sigma_2=Q^{-1}$, we have
\begin{align}\label{R-sym-2}
\begin{aligned}
&\sigma_2e^{i(\Lambda X+4\Lambda^3T\pm2\Lambda^{-1})\sigma_3}Q^{-1}e^{-i(\Lambda X+4\Lambda^3T\pm2\Lambda^{-1})\sigma_3}\sigma_2\vspace{0.05in}\\
=&e^{-i(\Lambda X+4\Lambda^3T\pm2\Lambda^{-1})\sigma_3}\sigma_2Q^{-1}\sigma_2e^{i(\Lambda X+4\Lambda^3T\pm2\Lambda^{-1})\sigma_3}\vspace{0.05in}\\
=&e^{-i(\Lambda X+4\Lambda^3T\pm2\Lambda^{-1})\sigma_3}Q^{-1}e^{i(\Lambda X+4\Lambda^3T\pm2\Lambda^{-1})\sigma_3}
\end{aligned}
\end{align}

Substituting Eq.~(\ref{R-sym-2}) into Eq.~(\ref{R-sym-1}), we obtain
\begin{align}\no
\begin{aligned}
\sigma_2D^{\pm}_+(\Lambda^*; X,T)^*\sigma_2=&\,\sigma_2D^{\pm}_-(\Lambda^*; X,T)^*\sigma_2e^{-i(\Lambda X+4\Lambda^3T\pm2\Lambda^{-1})\sigma_3}Q^{-1} e^{i(\Lambda X+4\Lambda^3T\pm2\Lambda^{-1})\sigma_3},\quad \Lambda\in\Sigma_1.
\end{aligned}
\end{align}
Then we obtain the symmetry (\ref{sch-sym}).
Thus the proof is completed.
\end{proof}

\begin{prop}\label{sym1}
The following identities (symmetries) hold:
\bee
D^{\mp}(\Lambda;X,T)
=\begin{cases}
-\sigma_3D^{\pm}(\Lambda;X,T)e^{-2i(\Lambda X+4\Lambda^3T)\sigma_3}\sigma_1,\quad |\Lambda|<1,\vspace{0.05in}\\
\sigma_3D^{\pm}(\Lambda;X,T)e^{\mp 4i\Lambda^{-1}\sigma_3}\sigma_3,\qquad
\qquad |\Lambda|>1.
\end{cases}
\ene
\end{prop}

\begin{proof}
To simplify the representation, let
\begin{align}\no
\begin{aligned}
V^{+}(\Lambda; X,T)=e^{-i(\Lambda X+4\Lambda^3T+2\Lambda^{-1})\sigma_3}Q^{-1}e^{i(\Lambda X+4\Lambda^3T+2\Lambda^{-1})\sigma_3},\vspace{0.05in}\\
V^{-}(\Lambda; X,T)=e^{-i(\Lambda X+4\Lambda^3T-2\Lambda^{-1})\sigma_3}Q^{-1}e^{i(\Lambda X+4\Lambda^3T-2\Lambda^{-1})\sigma_3}.
\end{aligned}
\end{align}

Using $-\sigma_1Q^{-1}\sigma_3=Q^{-1}$, we have
\begin{align}\label{V+V-}
\begin{aligned}
V^{+}(\Lambda; X,T)&=e^{-i(\Lambda X+4\Lambda^3T+2\Lambda^{-1})\sigma_3}Q^{-1}e^{i(\Lambda X+4\Lambda^3T+2\Lambda^{-1})\sigma_3}\vspace{0.05in}\\
&=-e^{-i(\Lambda X+4\Lambda^3T+2\Lambda^{-1})\sigma_3}\sigma_1Q^{-1}\sigma_3e^{i(\Lambda X+4\Lambda^3T+2\Lambda^{-1})\sigma_3}\vspace{0.05in}\\
&=-e^{-i(\Lambda X+4\Lambda^3T)\sigma_3}e^{-2i\Lambda^{-1}\sigma_3}\sigma_1Q^{-1}\sigma_3e^{i(\Lambda X+4\Lambda^3T)\sigma_3}e^{2i\Lambda^{-1}\sigma_3}\vspace{0.05in}\\
&=-e^{-i(\Lambda X+4\Lambda^3T)\sigma_3}\sigma_1e^{2i\Lambda^{-1}\sigma_3}Q^{-1}e^{-2i\Lambda^{-1}\sigma_3}e^{i(\Lambda X+4\Lambda^3T)\sigma_3}\sigma_3e^{4i\Lambda^{-1}\sigma_3}\vspace{0.05in}\\
&=-e^{-2i(\Lambda X+4\Lambda^3T)\sigma_3}\sigma_1e^{-i(\Lambda X+4\Lambda^3T)\sigma_3}e^{2i\Lambda^{-1}\sigma_3}Q^{-1}e^{-2i\Lambda^{-1}\sigma_3}e^{i(\Lambda X+4\Lambda^3T)\sigma_3}\sigma_3e^{4i\Lambda^{-1}\sigma_3}\vspace{0.05in}\\
&=-e^{-2i(\Lambda X+4\Lambda^3T)\sigma_3}\sigma_1e^{-i(\Lambda X+4\Lambda^3T-2i\Lambda^{-1})\sigma_3}Q^{-1}e^{i(\Lambda X+4\Lambda^3T-2i\Lambda^{-1})\sigma_3}\sigma_3e^{4i\Lambda^{-1}\sigma_3}\vspace{0.05in}\\
&=-e^{-2i(\Lambda X+4\Lambda^3T)\sigma_3}\sigma_1V^-(\Lambda; X,T)\sigma_3e^{4i\Lambda^{-1}\sigma_3}.
\end{aligned}
\end{align}

According to Eq.~(\ref{RH4}), we have
\begin{align}\label{D+D+}
\begin{aligned}
D^{+}_+(\Lambda; X,T)&=D^{+}_-(\Lambda; x,t)V^+(\Lambda; X,T)\vspace{0.05in}\\
&=-D^{+}_-(\Lambda; X,T)e^{-2i(\Lambda X+4\Lambda^3T)\sigma_3}\sigma_1V^-(\Lambda; X,T)\sigma_3e^{4i\Lambda^{-1}\sigma_3}.
\end{aligned}
\end{align}

Rewriting Eq.~(\ref{D+D+}), we have
\bee\label{D+D+1}
D^{+}_+(\Lambda; X,T)e^{-4i\Lambda^{-1}\sigma_3}\sigma_3=-D^{+}_-(\Lambda; X,T)e^{-2i(\Lambda X+4\Lambda^3T)\sigma_3}\sigma_1V^-(\Lambda; X,T).
\ene

To ensure the normalization of the Riemann-Hilbert Problem 4, both sides of Eq.~(\ref{D+D+1}) are multiplied by matrix $\sigma_3$ simultaneously,
\bee\label{D+D+12}
\sigma_3D^{+}_+(\Lambda; X,T)e^{-4i\Lambda^{-1}\sigma_3}\sigma_3=-\sigma_3D^{+}_-(\Lambda; X,T)e^{-2i(\Lambda X+4\Lambda^3T)\sigma_3}\sigma_1V^-(\Lambda; X,T).
\ene

Then, we have
\bee
D^{-}(\Lambda;X,T)
=\begin{cases}
-\sigma_3D^{+}(\Lambda;X,T)e^{-2i(\Lambda X+4\Lambda^3T)\sigma_3}\sigma_1,\quad |\Lambda|<1,\vspace{0.05in}\\
\sigma_3D^{+}(\Lambda;X,T)e^{-4i\Lambda^{-1}\sigma_3}\sigma_3,\quad |\Lambda|>1.
\end{cases}
\ene

On the other hand, using Eq.~(\ref{V+V-}), we obtain
\bee
V^{-}(\Lambda; X,T)=-e^{-2i(\Lambda X+4\Lambda^3T)\sigma_3}\sigma_1V^{+}(\Lambda; X,T)\sigma_3e^{-4i\Lambda^{-1}\sigma_3}.
\ene

According to Eq.~(\ref{RH4}), we have
\begin{align}\label{D-D-}
\begin{aligned}
D^{-}_+(\Lambda; X,T)&=D^{-}_-(\Lambda; X,T)V^-(\Lambda; X,T)\vspace{0.05in}\\
&=-D^{-}_-(\Lambda; X,T)e^{-2i(\Lambda X+4\Lambda^3T)\sigma_3}\sigma_1V^{+}(\Lambda; X,T)\sigma_3e^{-4i\Lambda^{-1}\sigma_3}.
\end{aligned}
\end{align}

Rewriting Eq.~(\ref{D-D-}), we have
\bee
\sigma_3D^{-}_+(\Lambda; X,T)e^{4i\Lambda^{-1}\sigma_3}\sigma_3=-\sigma_3D^{-}_-(\Lambda; X,T)e^{-2i(\Lambda X+4\Lambda^3T)\sigma_3}\sigma_1V^+(\Lambda; X,T).
\ene

Then, we have
\bee
D^{+}(\Lambda;X,T)
=\begin{cases}
-\sigma_3D^{-}(\Lambda;X,T)e^{-2i(\Lambda X+4\Lambda^3T)\sigma_3}\sigma_1,\quad |\Lambda|<1,\vspace{0.05in}\\
\sigma_3D^{-}(\Lambda;X,T)e^{4i\Lambda^{-1}\sigma_3}\sigma_3,\quad |\Lambda|>1.
\end{cases}
\ene

Thus the proof is completed.
\end{proof}

From Proposition \ref{sym1}, we know that $\widehat{q}^+(X,T)=-\widehat{q}^-(X,T)$ holds for all $(X,T)\in\mathbb{R}^2$.

\begin{prop}\label{sym-prop2}
The following relations hold:
\begin{align}\no
\begin{aligned}
&D^{\mp}(-\Lambda;-X,-T)=D^{\pm}(\Lambda;X,T),\vspace{0.05in}\\
&\widehat{q}^{\pm}(-X,-T)=\widehat{q}^{\pm}(X,T).
\end{aligned}
\end{align}
\end{prop}

\begin{proof}
Using Eq.~(\ref{RH4}) and replacing $\Lambda$ with $-\Lambda$, $X$ with $-X$, $T$ with $-T$, we have
\bee\label{RH42}
D^{\pm}_+(-\Lambda; -X,-T)=D^{\pm}_-(-\Lambda; -X,-T)e^{-i(\Lambda X+4\Lambda^3T\mp2\Lambda^{-1})\sigma_3}Q^{-1}e^{i(\Lambda X+4\Lambda^3T\mp2\Lambda^{-1})\sigma_3},\quad \Lambda\in\Sigma_1.
\ene
Then we obtain
\bee\no
D^{\mp}(-\Lambda;-X,-T)=D^{\pm}(\Lambda;X,T).
\ene

According to Eq.~(\ref{RH4}) and (\ref{fanyan-q2}), we have
\begin{align}\no
\begin{aligned}
\widehat{q}^{\pm}(-X,-T)&=2i\lim_{\Lambda\rightarrow\infty}\Lambda D^{\pm}_{12}(\Lambda;-X,-T)\vspace{0.05in}\\
&=2i\lim_{\Lambda\rightarrow\infty}\Lambda D^{\mp}_{12}(-\Lambda;X,T)\vspace{0.05in}\\
&=-2i\lim_{\Lambda\rightarrow\infty}\Lambda D^{\mp}_{12}(\Lambda;X,T)\vspace{0.05in}\\
&=-\widehat{q}^{\mp}(X,T)=\widehat{q}^{\pm}(X,T).
\end{aligned}
\end{align}

Thus the proof is completed.
\end{proof}

\begin{prop}\label{prop-real}
The following relations hold:
\begin{align}\label{27-real}
\begin{aligned}
&D^{\pm}(-\Lambda^*;X,T)^*=D^{\pm}(\Lambda;X,T),\vspace{0.05in}\\
&\widehat{q}^{\pm}(X,T)=\widehat{q}^{\pm}(X,T)^*.
\end{aligned}
\end{align}
\end{prop}

\begin{proof}
Using Eq.~(\ref{RH4}) and replacing $\Lambda$ with $-\Lambda^*$, we have
\bee\label{D+-}
D^{\pm}_+(-\Lambda^*; X,T)=D^{\pm}_-(-\Lambda^*; X,T)e^{i(\Lambda^* X+4\Lambda^{*3}T\pm2\Lambda^{*-1})\sigma_3}Q^{-1}e^{-i(\Lambda^* X+4\Lambda^{*3}T\pm2\Lambda^{*-1})\sigma_3},\quad \Lambda\in\Sigma_1.
\ene

For Eq.~(\ref{D+-}), taking conjugation on both sides simultaneously, we have
\bee\label{RH43}
D^{\pm}_+(-\Lambda^*; X,T)^*=D^{\pm}_-(-\Lambda^*; X,T)^*e^{-i(\Lambda X+4\Lambda^3T\pm2\Lambda^{-1})\sigma_3}Q^{-1}e^{i(\Lambda X+4\Lambda^3T\pm2\Lambda^{-1})\sigma_3},\quad \Lambda\in\Sigma_1.
\ene

Then we obtain
\bee\no
D^{\pm}(-\Lambda^*;X,T)^*=D^{\pm}(\Lambda;X,T).
\ene

According to Eqs.~(\ref{RH4}) and (\ref{fanyan-q2}), we have
\begin{align}\no
\begin{aligned}
\widehat{q}^{\pm}(X,T)&=2i\lim_{\Lambda\rightarrow\infty}\Lambda D^{\pm}_{12}(\Lambda;X,T)\vspace{0.05in}\\
&=2i\lim_{\Lambda\rightarrow\infty}\Lambda D^{\pm}_{12}(-\Lambda^*;X,T)^*\vspace{0.05in}\\
&=\left(-2i\lim_{\Lambda\rightarrow\infty}\Lambda^* D^{\pm}_{12}(-\Lambda^*;X,T)\right)^*\vspace{0.05in}\\
&=\left(2i\lim_{\Lambda\rightarrow\infty}\Lambda D^{\pm}_{12}(\Lambda;X,T)\right)^*\vspace{0.05in}\\
&=\widehat{q}^{\pm}(X,T)^*,
\end{aligned}
\end{align}
which impies that $\widehat{q}^{\pm}(X,T)$ is real.
Thus the proof is completed.
\end{proof}

\subsection{Ordinary differential equations in $X$ or $T$}

Proposition \ref{prop4} indicates that the function $\widehat{q}^{\pm}(X,T)$ satisfies the c-mKdV equation (\ref{cmkdv-1}). The purpose of this subsection is to prove that these special solutions $\widehat{q}^{\pm}(X,T)$ of the c-mKdV equation also satisfy some ordinary differential equation about $X$ ($T$) for each fixed $T$ ($X$).

According to Proposition \ref{sym1}, we show that $\widehat{q}^+(X,T)=-\widehat{q}^-(X,T)$ holds for all $(X,T)\in\mathbb{R}^2$, then it suffices to consider the case $D(\Lambda;X,T):=D^+(\Lambda;X,T),\, \widehat{q}(X,T):=\widehat{q}^{+}(X,T)$ of Riemann-Hilbert Problem 4. Let
\bee\label{H}
H(\Lambda;X,T):=D(\Lambda;X,T)e^{-i(\Lambda X+4\Lambda^3T+2\Lambda^{-1})\sigma_3}.
\ene

Then we construct a Riemann-Hilbert Problem about matrix function $H(\Lambda;X,T)$.
\begin{prop} Riemann-Hilbert Problem 5.\\
Let $(X,T)$ be in compact subsets of $\mathbb{R}^2$. Find a $2\times 2$ matrix $H(\Lambda;X,T)$ that satisfies the following properties:

\begin{itemize}
    \item{} Analyticity: $H(\Lambda; X,T)$ is analytic in $\mathbb{C}\setminus\Sigma_1$ and takes continuous boundary values on $\Sigma_1$.

\item{} Jump conditions: Assuming clockwise orientation of $\Sigma_1$, the boundary values on the jump contour $\Sigma_1$ are related as:
\bee\label{RH44}
H_+(\Lambda; X,T)=H_-(\Lambda; X,T)Q^{-1},\quad \Lambda\in\Sigma_1.
\ene

\item{} Normalization: $H(\Lambda; X,T)\rightarrow e^{-i(\Lambda X+4\Lambda^3T)\sigma_3}$as $\Lambda\rightarrow\infty$.
\end{itemize}
\end{prop}

As we discussed in the proof of Proposition \ref{prop5}, the matrix function $H(\Lambda; X,T)$ possesses the Lax pair:
\begin{align}\label{lax1-1}
\left\{\begin{aligned}
&H_X(\Lambda; X,T)=A(\Lambda;X,T)H(\Lambda; X,T),\vspace{0.05in}\\
&H_T(\Lambda; X,T)=B(\Lambda;X,T)H(\Lambda; X,T),
\end{aligned}\right.
\end{align}
where
\begin{align}\no
\begin{aligned}
&A(\Lambda;X,T)=\left[\!\!\begin{array}{cc}
-i\Lambda& \widehat{q}(X,T)  \vspace{0.05in}\\
-\widehat{q}^{*}(X,T) & i\Lambda
\end{array}\!\!\right],\vspace{0.05in}\\
&B(\Lambda;X,T)=-4i\Lambda^3\sigma_3+B^{[2]}(X,T)\Lambda^2+B^{[1]}(X,T)\Lambda+B^{[0]}(X,T),
\end{aligned}
\end{align}
with
\begin{align}\label{B2B1B0}
\begin{aligned}
&B^{[2]}(X,T)=\left[\!\!\begin{array}{cc}
0& 4\widehat{q}(X,T) \vspace{0.05in}\\
-4\widehat{q}^{*}(X,T) & 0
\end{array}\!\!\right],\vspace{0.05in}\\
&B^{[1]}(X,T)=\left[\!\!\begin{array}{cc}
2i|\widehat{q}(X,T)|^2& 2i\widehat{q}_X(X,T) \vspace{0.05in}\\
2i\widehat{q}_X^{*}(X,T) & -2i|\widehat{q}(X,T)|^2
\end{array}\!\!\right],\vspace{0.05in}\\
&B^{[0]}(X,T)=\left[\!\!\begin{array}{cc}
\widehat{q}_X^{*}(X,T)\widehat{q}(X,T)-\widehat{q}_X(X,T)\widehat{q}^{*}(X,T)& -\widehat{q}_{XX}(X,T)-2|\widehat{q}(X,T)|^2\widehat{q}(X,T)  \vspace{0.05in}\\
\widehat{q}_{XX}^*(X,T)+2|\widehat{q}(X,T)|^2\widehat{q}^{*}(X,T) & -\widehat{q}_X^{*}(X,T)\widehat{q}(X,T)+\widehat{q}_X(X,T)\widehat{q}^{*}(X,T)
\end{array}\!\!\right].
\end{aligned}
\end{align}

Rewrite $A(\Lambda;X,T)$ as $-i\Lambda\sigma_3+A^{[0]}(X,T)$, where
\bee\no
&A^{[0]}(X,T)=\left[\!\!\begin{array}{cc}
0& \widehat{q}(X,T) \vspace{0.05in}\\
-\widehat{q}^{*}(X,T) & 0
\end{array}\!\!\right]=\dfrac{1}{4}B^{[2]}(X,T).
\ene

Let the matrix $D(\Lambda;X,T)$ have the following Laurent expansion.
\bee\label{D-lau}
D(\Lambda;X,T)=\mathbb{I}+\sum_{j=1}^{\infty}D^{[j]}(X,T)\Lambda^{-j},\quad |\Lambda|>1.
\ene
Then, we have
\bee\label{q-D-fanyan}
q(X,T)=2iD^{[1]}_{12}(X,T),\quad q(X,T)^*=2iD^{[1]}_{21}(X,T).
\ene

Since the jump matrix of $H(\Lambda; X,T)$ is a constant matrix $Q$, there exists another Lax equation about $\Lambda$. To this end, we define
\bee\label{lax-L}
C(\Lambda; X,T):=H_{\Lambda}(\Lambda; X,T)H(\Lambda; X,T)^{-1},
\ene
where the subscripts denote the partial derivatives with respect to the variable $\Lambda$. $C(\Lambda; X,T)$ has no jump across $\Sigma_1$ and hence may be considered to be analytic in the whole complex plane about $\Lambda$.

According to Eq.~(\ref{H}), we have
\bee\label{C-no}
C(\Lambda; X,T)=D_{\Lambda}(\Lambda; X,T)D(\Lambda; X,T)^{-1}-i(X+12T\Lambda^2-2\Lambda^{-2})D(\Lambda; X,T)\sigma_3D(\Lambda; X,T)^{-1}.
\ene

Firstly, we provide the expression for $D(\Lambda;X,T)^{-1}$.
\begin{align}\label{N-1-2}
\begin{aligned}
D(\Lambda;X,T)^{-1}=&\mathbb{I}-D^{[1]}(X,T)\Lambda^{-1}+\left(D^{[1]}(X,T)^2-D^{[2]}(X,T)\right)\Lambda^{-2}+\bigg(D^{[1]}(X,T)D^{[2]}(X,T)\vspace{0.05in}\\
&+D^{[2]}(X,T)D^{[1]}(X,T)-D^{[3]}(X,T)-D^{[1]}(X,T)^3\bigg)\Lambda^{-3}+\bigg(D^{[1]}(X,T)D^{[3]}(X,T)\vspace{0.05in}\\
&-D^{[1]}(X,T)D^{[2]}(X,T)D^{[1]}(X,T)-D^{[1]}(X,T)^2D^{[2]}(X,T)+D^{[1]}(X,T)^4\vspace{0.05in}\\
&+D^{[2]}(X,T)^2-D^{[2]}(X,T)D^{[1]}(X,T)^2+D^{[3]}(X,T)D^{[1]}(X,T)-D^{[4]}(X,T)\bigg)\Lambda^{-4}\vspace{0.05in}\\
&+\mathcal{O}(\Lambda^{-5}),\quad \Lambda\rightarrow\infty.
\end{aligned}
\end{align}

Using Eq.~(\ref{D-lau}), we obtain
\begin{align}\label{C-Lau}
\begin{aligned}
C(\Lambda; X,T)=&C^{[2]}(X,T)\Lambda^2+C^{[1]}(X,T)\Lambda+C^{[0]}(X,T)\vspace{0.05in}\\
&+C^{[-1]}(X,T)\Lambda^{-1}+C^{[-2]}(X,T)\Lambda^{-2}+\mathcal{O}(\Lambda^{-3}),\quad \Lambda\rightarrow\infty,
\end{aligned}
\end{align}
where
\begin{align}\no
C^{[2]}(X,T)=&-12iT\sigma_3,\vspace{0.05in}\\
C^{[1]}(X,T)=&12iT[\sigma_3,D^{[1]}(X,T)],\vspace{0.05in} \no\\
C^{[0]}(X,T)=&-iX\sigma_3+12iT\left([\sigma_3,D^{[2]}(X,T)]+[D^{[1]}(X,T),\sigma_3D^{[1]}(X,T)]\right),\vspace{0.05in}\no\\
C^{[-1]}(X,T)=&iX[\sigma_3,D^{[1]}(X,T)]+12iT\bigg([\sigma_3,D^{[3]}(X,T)]+[D^{[2]}(X,T),\sigma_3D^{[1]}(X,T)]\vspace{0.05in}\no\\
&+[\sigma_3D^{[1]}(X,T)^2,D^{[1]}(X,T)]+[D^{[1]}(X,T),\sigma_3D^{[2]}(X,T)]\bigg),\vspace{0.05in}\no\\
C^{[-2]}(X,T)=&-D^{[1]}(X,T)+2i\sigma_3+iX\left([\sigma_3,D^{[2]}(X,T)]+[D^{[1]}(X,T),\sigma_3D^{[1]}(X,T)]\right)\vspace{0.05in}\no\\
&+12iT\bigg([\sigma_3,D^{[4]}(X,T)]+[D^{[3]}(X,T),\sigma_3D^{[1]}(X,T)]\vspace{0.05in}\no\\
&+[\sigma_3D^{[1]}(X,T)D^{[2]}(X,T),D^{[1]}(X,T)]+[\sigma_3D^{[1]}(X,T)^2,D^{[2]}(X,T)]\vspace{0.05in}\no\\
&+[D^{[1]}(X,T),\sigma_3D^{[1]}(X,T)^3]+[D^{[2]}(X,T),\sigma_3D^{[2]}(X,T)]\vspace{0.05in}\no\\
&+[\sigma_3D^{[2]}(X,T)D^{[1]}(X,T),D^{[1]}(X,T)]+[D^{[1]}(X,T),\sigma_3D^{[3]}(X,T)]\bigg).\no
\end{align}

Moreover, we have the Taylor expansions at the origin of $D(\Lambda;X,T)$ and $D(\Lambda;X,T)^{-1}$:
\begin{align}\label{D-taylor}
\begin{aligned}
&D(\Lambda;X,T)=D(0;X,T)+D_\Lambda(0;X,T)\Lambda+\mathcal{O}(\Lambda^2),\quad \Lambda\rightarrow0,\vspace{0.05in}\\
&D(\Lambda;X,T)^{-1}=D(0;X,T)^{-1}-D(0;X,T)^{-1}D_\Lambda(0;X,T)D(0;X,T)^{-1}\Lambda+\mathcal{O}(\Lambda^2),\quad \Lambda\rightarrow0.
\end{aligned}
\end{align}

Substituting Eq.~(\ref{D-taylor}) into Eq.~(\ref{C-no}), we obtain
\begin{align}\label{C-taylor}
\begin{aligned}
C(\Lambda;X,T)=&2iD(0;X,T)\sigma_3D(0;X,T)^{-1}\Lambda^{-2}\vspace{0.05in}\\
&+2i[D_{\Lambda}(0;X,T)D(0;X,T)^{-1},D(0;X,T)\sigma_3D(0;X,T)^{-1}]\Lambda^{-1}+\mathcal{O}(1),\quad \Lambda\rightarrow0.
\end{aligned}
\end{align}

Comparing Eqs.~(\ref{C-Lau}) and (\ref{C-taylor}), we have
\bee\label{C-final}
C(\Lambda; X,T)=C^{[2]}(X,T)\Lambda^2+C^{[1]}(X,T)\Lambda+C^{[0]}(X,T)+C^{[-1]}(X,T)\Lambda^{-1}+C^{[-2]}(X,T)\Lambda^{-2}.
\ene

Using Eq.~(\ref{C-taylor}), we have another representation of $C^{[-2]}(X,T)$:
\bee\label{C-2-re}
C^{[-2]}(X,T)=2iD(0;X,T)\sigma_3D(0;X,T)^{-1},
\ene
which shows that
\bee\label{C-2-re1}
\mathrm{tr}\l(C^{[-2]}(X,T)\r)=0,\quad \det\l(C^{[-2]}(X,T)\r)=4.
\ene

Applying Proposition \ref{prop-sch} and Eq.~(\ref{C-2-re1}), we have
\bee\no
C^{[-2]}(X,T)=\left[\!\!\begin{array}{cc}
i\alpha(X,T)& i\beta(X,T)  \vspace{0.05in}\\
i\beta(X,T)^* & -i\alpha(X,T)
\end{array}\!\!\right],
\ene
where $\alpha(X,T)$ is a real-valued function and it satisfies $\alpha(X,T)^2+|\beta(X,T)|^2=4$.

Then, we obtain a new Lax system:
\bee\label{lax-L-1}
H_{\Lambda}(\Lambda; X,T)=C(\Lambda; X,T)H(\Lambda; X,T).
\ene

As we discussed in the proof of Proposition \ref{prop5}, we rewrite $C^{[1]}(X,T)$,$C^{[0]}(X,T)$ and $C^{[-1]}(X,T)$ as follows:
\begin{align}\label{re-C0-C-1}
\begin{aligned}
C^{[1]}(X,T)=&3TB^{[2]}(X,T),\vspace{0.05in}\\
C^{[0]}(X,T)=&-iX\sigma_3+3TB^{[1]}(X,T),\vspace{0.05in}\\
C^{[-1]}(X,T)=&\frac{X}{4}B^{[2]}(X,T)+3TB^{[0]}(X,T),
\end{aligned}
\end{align}
where $B^{[2]}(X,T),B^{[1]}(X,T)$ and $B^{[0]}(X,T)$ are defined by Eq.~(\ref{B2B1B0}).

\subsubsection{Ordinary differential equations in $X$ and proof of Theorem \ref{new-th-1}}

The matrix-valued function $H(\Lambda; X,T)$ possesses the following $X$-Lax pair:
\begin{align}\label{lax-new-1}
\left\{\begin{aligned}
&H_X(\Lambda; X,T)=A(\Lambda;X,T)H(\Lambda; X,T),\vspace{0.05in}\\
&H_\Lambda(\Lambda; X,T)=C(\Lambda;X,T)H(\Lambda; X,T),
\end{aligned}\right.
\end{align}
where $A(\Lambda;X,T)$ and $C(\Lambda;X,T)$ are defined by Eqs.~(\ref{lax1-1}) and (\ref{C-final}) respectively. The compatibility condition of $X$-Lax pair (\ref{lax-new-1}) becomes the zero-curvature condition:
\bee\label{zero-curvature}
C_X(\Lambda; X,T)-A_\Lambda(\Lambda; X,T)+[C(\Lambda; X,T),A(\Lambda; X,T)]=0.
\ene

The left-hand side of Eq.~(\ref{zero-curvature}) is a Laurent polynomial in $\Lambda$ with powers ranging from $\Lambda^{-2}$ through $\Lambda^{3}$, and therefore its coefficients must be equal to zero. Below, we will discuss the equations satisfied by the coefficients of $\Lambda$ respectively.

The coefficient of $\Lambda^{3}$ is
\bee\no
[i\sigma_3,C^{[2]}(X,T)]=0,
\ene
which holds automatically because the off-diagonal element of $C^{[2]}(X,T)$ is zero.

The coefficient of $\Lambda^{2}$ is
\bee\no
C_X^{[2]}(X,T)+[C^{[2]}(X,T),A^{[0]}(X,T)]+[i\sigma_3,C^{[1]}(X,T)]=0,
\ene
which holds automatically because of Eq.~(\ref{q-D-fanyan}).

The coefficient of $\Lambda$ is
\bee\no
C_X^{[1]}(X,T)+[C^{[1]}(X,T),A^{[0]}(X,T)]+[i\sigma_3,C^{[0]}(X,T)].
\ene

We note that
\begin{align}
\begin{aligned}
&C_X^{[1]}(X,T)+[C^{[1]}(X,T),A^{[0]}(X,T)]+[i\sigma_3,C^{[0]}(X,T)]\vspace{0.05in}\\
=&3TB_X^{[2]}(X,T)+[3TB^{[2]}(X,T),\frac14B^{[2]}(X,T)]+i[\sigma_3,-iX\sigma_3+3TB^{[1]}(X,T)]\vspace{0.05in}\\
=&3TB_X^{[2]}(X,T)+i[\sigma_3,3TB^{[1]}(X,T)]\vspace{0.05in}\\
=&0.
\end{aligned}
\end{align}
Then, the coefficient of $\Lambda$ is equal to zero.

The coefficient of $\Lambda^0$ is
\bee\no
C_X^{[0]}(X,T)+i\sigma_3+[C^{[0]}(X,T),A^{[0]}(X,T)]+[i\sigma_3,C^{[-1]}(X,T)].
\ene
We note that
\begin{align}
\begin{aligned}
&C_X^{[0]}(X,T)+i\sigma_3+[C^{[0]}(X,T),A^{[0]}(X,T)]+[i\sigma_3,C^{[-1]}(X,T)]\vspace{0.05in}\\
=&-i\sigma_3+3TB_X^{[1]}(X,T)+i\sigma_3+[-iX\sigma_3+3TB^{[1]}(X,T),\frac14B^{[2]}(X,T)]\vspace{0.05in}\\
&+[i\sigma_3,\frac{X}{4}B^{[2]}(X,T)+3TB^{[0]}(X,T)]\vspace{0.05in}\\
=&\, 3TB_X^{[1]}(X,T)+3T[B^{[1]}(X,T),\frac14B^{[2]}(X,T)]+3T[i\sigma_3,B^{[0]}(X,T)]\vspace{0.05in}\\
=&0.
\end{aligned}
\end{align}
Then, the coefficient of $\Lambda^0$ is equal to zero.

The coefficient of $\Lambda^{-1}$ is
\bee\no
C_X^{[-1]}(X,T)+[C^{[-1]}(X,T),A^{[0]}(X,T)]+[i\sigma_3,C^{[-2]}(X,T)]=0,
\ene
which is equivalent to
\begin{align}\label{eq1}
\begin{aligned}
&\widehat{q}+X\widehat{q}_X-3T\widehat{q}_{XXX}-18T\widehat{q}_X|\widehat{q}|^2-2\beta=0,\vspace{0.05in}\\
&-\widehat{q}^*-X\widehat{q}^*_X+3T\widehat{q}^*_{XXX}+18T|\widehat{q}|^2\widehat{q}^*+2\beta^*=0.
\end{aligned}
\end{align}

The coefficient of $\Lambda^{-2}$ is
\bee\no
C_X^{[-2]}(X,T)+[C^{[-2]}(X,T),A^{[0]}(X,T)]=0,
\ene
which is equivalent to
\begin{align}\label{eq2}
\begin{aligned}
&\alpha_X(X,T)-\beta(X,T) \widehat{q}^*(X,T)-\beta^*(X,T)\widehat{q}(X,T)=0,\vspace{0.05in}\\
&-\alpha_X(X,T)+\beta^*(X,T)\widehat{q}(X,T)+\beta(X,T)\widehat{q}^*(X,T)=0,\vspace{0.05in}\\
&\beta_X(X,T)+2\alpha(X,T)\widehat{q}(X,T)=0,\vspace{0.05in}\\
&\beta_X^*(X,T)+2\alpha(X,T)\widehat{q}^*(X,T)=0.
\end{aligned}
\end{align}

Taking the derivative of $X$ on both sides of Eq.~(\ref{eq1}) simultaneously, we have
\bee\label{eq3}
2\widehat{q}_X+X\widehat{q}_{XX}
-3T(\widehat{q}_{XXXX}-+6\widehat{q}_{XX}|\widehat{q}|^2
+6\widehat{q}^2\widehat{q}^*+6|\widehat{q}|^2\widehat{q})-2\beta_X=0.
\ene

Substituting Eqs.~(\ref{eq1}) and (\ref{eq3}) into Eq.~(\ref{eq2}), we obtain
\begin{align}\label{eq4}
\begin{aligned}
&2\widehat{q}_X+X\widehat{q}_{XX}
-3T(\widehat{q}_{XXXX}+3(\widehat{q}^2)_{XX}\widehat{q}^*+6|\widehat{q}_X|^2\widehat{q})+4\alpha \widehat{q}=0,\vspace{0.05in}\\
&2\alpha_X-2|\widehat{q}|^2
-X(|\widehat{q}|^2)_X
+3T(\widehat{q}^*\widehat{q}_{XXX}+\widehat{q}\widehat{q}_{XXX}^*+3(|\widehat{q}|^4)_X)=0.
\end{aligned}
\end{align}

Eliminating $\alpha(X,T)$, we have
\begin{align}\label{eq5}
\begin{aligned}
&3\widehat{q}\widehat{q}_{XX}
-2\widehat{q}_X^2+4\widehat{q}^2|\widehat{q}|^2+
+X\l(\widehat{q}\widehat{q}_{XXX}-\widehat{q}_X\widehat{q}_{XX}+2\widehat{q}\widehat{q}_X|\widehat{q}|^2+2\widehat{q}^3\widehat{q}_X^*\r)\vspace{0.05in}\\
&+3T\bigg(\widehat{q}_X\widehat{q}_{XXXX}-\widehat{q}\widehat{q}_{XXXXX}+6\widehat{q}_X^3\widehat{q}^*-6\widehat{q}_{XXX}\widehat{q}|\widehat{q}|^2-12\widehat{q}_{XX}\widehat{q}^2\widehat{q}_X^*-12\widehat{q}_X\widehat{q}_{XX}|\widehat{q}|^2\vspace{0.05in}\\
&-6\widehat{q}^2\widehat{q}_X\widehat{q}_{XX}^*-6|\widehat{q}_X|^2\widehat{q}_X\widehat{q}-2\widehat{q}^2\widehat{q}^*\widehat{q}_{XXX}-2\widehat{q}^3\widehat{q}_{XXX}^*-12|\widehat{q}|^4\widehat{q}\widehat{q}_X-12|\widehat{q}|^2\widehat{q}^3\widehat{q}_X^*\bigg)=0.
\end{aligned}
\end{align}

Using the conservation law $\alpha(X,T)^2+|\beta(X,T)|^2=4$ and Eq.~(\ref{eq2}), we have
\bee\label{eq6}
|\beta_X|^2+4|\widehat{q}|^2|\beta|^2-16|\widehat{q}|^2=0.
\ene

Substituting Eq.~(\ref{eq1}) into Eq.~(\ref{eq6}), we obtain
\bee\label{eq7}
\left|\left(\widehat{q}+X\widehat{q}_X-3T\widehat{q}_{XXX}-18T\widehat{q}_X|\widehat{q}|^2\right)_X\right|^2+4|\widehat{q}|^2\left|\widehat{q}+X\widehat{q}_X-3T\widehat{q}_{XXX}-18T\widehat{q}_X|\widehat{q}|^2\right|^2-64|\widehat{q}|^2=0.
\ene

According to Proposition \ref{prop-real}, we know that $\widehat{q}^{\pm}(X,T)=\widehat{q}^{\pm}(X,T)^*$. Substituting it into Eqs.~(\ref{eq5}) and (\ref{eq7}), we obtain
\begin{align}\label{eq8}
\begin{aligned}
&3\widehat{q}\widehat{q}_{XX}-X\widehat{q}_X\widehat{q}_{XX}+X\widehat{q}\widehat{q}_{XXX}-2\widehat{q}_X^2+4\widehat{q}^4+4X\widehat{q}^3\widehat{q}_X\vspace{0.05in}\\
&+3T\bigg(\widehat{q}_X\widehat{q}_{XXXX}-\widehat{q}\widehat{q}_{XXXXX}-6\widehat{q}_{XXX}\widehat{q}^3-24\widehat{q}^2\widehat{q}_X
\widehat{q}_{XX}-6\widehat{q}_X\widehat{q}_{XX}\widehat{q}^2-4\widehat{q}^3\widehat{q}_{XXX}-24\widehat{q}^5\widehat{q}_X\bigg)=0,\vspace{0.05in}\\
&\bigg(\left(\widehat{q}+X\widehat{q}_X-3T\widehat{q}_{XXX}-18T\widehat{q}_X\widehat{q}^2\right)_X\bigg)^2+4\widehat{q}^2\left(\widehat{q}+X\widehat{q}_X-3T\widehat{q}_{XXX}-18T\widehat{q}_X\widehat{q}^2\right)^2-64\widehat{q}^2=0.
\end{aligned}
\end{align}

When $T=0$, Eq.~(\ref{lax-new-1}) becomes
\begin{align}\label{lax1-2}
\begin{aligned}
&H_X(\Lambda; X,0)=\left[\!\!\begin{array}{cc}
-i\Lambda& \widehat{q}(X,0)  \vspace{0.05in}\\
-\widehat{q}(X,0) & i\Lambda
\end{array}\!\!\right]H(\Lambda; X,0),\vspace{0.05in}\\
&H_\Lambda(\Lambda; X,0)=\left[\!\!\begin{array}{cc}
-iX-\dfrac{i(2\widehat{q}_X+X\widehat{q}_{XX})}{4\widehat{q}}\Lambda^{-2}& X\widehat{q}\Lambda^{-1}+\dfrac{i}{2}(\widehat{q}+X\widehat{q}_X)\Lambda^{-2}  \vspace{0.05in}\\
-X\widehat{q}\Lambda^{-1}+\dfrac{i}{2}(\widehat{q}+X\widehat{q}_X)\Lambda^{-2} & iX+\dfrac{i(2\widehat{q}_X+X\widehat{q}_{XX})}{4\widehat{q}}\Lambda^{-2}
\end{array}\!\!\right]H(\Lambda; X,0),
\end{aligned}
\end{align}
which is the Lax pair of fist member of the Painlev\'{e}-III hierarchy.

Eq.~(\ref{eq8}) with $T=0$ becomes
\begin{align}\label{eq9}
\begin{aligned}
&3\widehat{q}\widehat{q}_{XX}
+X(\widehat{q}\widehat{q}_{XXX}
-\widehat{q}_X\widehat{q}_{XX}
+4\widehat{q}^3\widehat{q}_X)
-2\widehat{q}_X^2+4\widehat{q}^4
=0,\vspace{0.05in}\\
&\bigg(\left(\widehat{q}+X\widehat{q}_X\right)_X\bigg)^2+4\widehat{q}^2\left(\widehat{q}+X\widehat{q}_X\right)^2-64\widehat{q}^2=0.
\end{aligned}
\end{align}

Introducing $f:=\partial_X\ln\widehat{q}(X,0)$, and dividing both of these equations (\ref{eq9}) through by $\widehat{q}^2$, they can be written as
\begin{align}\label{eq10}
\begin{aligned}
&Xf_{XX}+3f_X+2Xff_X+f^2+4X\widehat{q}\widehat{q}_X+4\widehat{q}^2=0,\vspace{0.05in}\\
&(2f+Xf_X+Xf^2)^2+(Xf+1)(4X\widehat{q}\widehat{q}_X+4\widehat{q}^2)=0.
\end{aligned}
\end{align}

Eliminating $4X\widehat{q}\widehat{q}_X+4\widehat{q}^2$, we get a second-order equation about $X$:
\bee
X^2(ff_{XX}-f_X^2-f^4)+X(f_{XX}+ff_X-3f^3)+3f_X-3f^2+64=0.
\ene
Thus the proof of Theorem \ref{new-th-1} is completed.

\subsubsection{Ordinary differential equations in $T$ and proof of Theorem \ref{new-th-2}}

The matrix-valued function $H(\Lambda; X,T)$ possesses the following $T$-Lax pair:
\begin{align}\label{lax-new-2}
\left\{\begin{aligned}
&H_T(\Lambda; X,T)=B(\Lambda;X,T)H(\Lambda; X,T),\vspace{0.05in}\\
&H_\Lambda(\Lambda; X,T)=C(\Lambda;X,T)H(\Lambda; X,T),
\end{aligned}\right.
\end{align}
where $B(\Lambda;X,T)$ and $C(\Lambda;X,T)$ are defined by Eqs.~(\ref{lax1-1}) and (\ref{C-final}) respectively. The compatibility condition of $T$-Lax pair (\ref{lax-new-2}) becomes the zero-curvature condition:
\bee\label{zero-curvature1}
C_T(\Lambda; X,T)-B_\Lambda(\Lambda; X,T)+[C(\Lambda; X,T),B(\Lambda; X,T)]=0.
\ene

The left-hand side of Eq.~(\ref{zero-curvature1}) is a Laurent polynomial in $\Lambda$ with powers ranging from $\Lambda^{-2}$ through $\Lambda^{5}$, and therefore its coefficients must be equal to zero. Below, we will discuss the equations satisfied by the coefficients of $\Lambda$ respectively.

The coefficient of $\Lambda^{5}$ is
\bee\no
[4i\sigma_3,C^{[2]}(X,T)]=0,
\ene
which holds automatically because the off-diagonal element of $C^{[2]}(X,T)$ is zero.

The coefficient of $\Lambda^{4}$ is
\bee\no
[C^{[2]}(X,T),B^{[2]}(X,T)]+[4i\sigma_3,C^{[1]}(X,T)]=0,
\ene
which holds automatically because of formula (\ref{q-D-fanyan}).

The coefficient of $\Lambda^{3}$ is
\bee\no
[C^{[2]}(X,T),B^{[1]}(X,T)]+[C^{[1]}(X,T),B^{[2]}(X,T)]+[4i\sigma_3,C^{[0]}(X,T)].
\ene

We note that
\begin{align}\no
\begin{aligned}
&[C^{[2]}(X,T),B^{[1]}(X,T)]+[C^{[1]}(X,T),B^{[2]}(X,T)]+[4i\sigma_3,C^{[0]}(X,T)]\vspace{0.05in}\\
=&[-12iT\sigma_3,B^{[1]}(X,T)]+[3TB^{[2]}(X,T),B^{[2]}(X,T)]+[4i\sigma_3,-iX\sigma_3+3TB^{[1]}(X,T)]\vspace{0.05in}\\
=&[-12iT\sigma_3,B^{[1]}(X,T)]+[4i\sigma_3,3TB^{[1]}(X,T)]\vspace{0.05in}\\
=&0.
\end{aligned}
\end{align}

Then, the coefficient of $\Lambda^{3}$ is equal to zero.

The coefficient of $\Lambda^{2}$ is
\begin{align}\no
\begin{aligned}
&C^{[2]}_T(X,T)+12i\sigma_3+[C^{[2]}(X,T),B^{[0]}(X,T)]+[C^{[1]}(X,T),B^{[1]}(X,T)]\vspace{0.05in}\\
&+[C^{[0]}(X,T),B^{[2]}(X,T)]+[4i\sigma_3,C^{[-1]}(X,T)].
\end{aligned}
\end{align}

We note that
\begin{align}\no
\begin{aligned}
&C^{[2]}_T(X,T)+12i\sigma_3+[C^{[2]}(X,T),B^{[0]}(X,T)]+[C^{[1]}(X,T),B^{[1]}(X,T)]\vspace{0.05in}\\
&+[C^{[0]}(X,T),B^{[2]}(X,T)]+[4i\sigma_3,C^{[-1]}(X,T)]\vspace{0.05in}\\
=&[-12iT\sigma_3,B^{[0]}(X,T)]+[3TB^{[2]}(X,T),B^{[1]}(X,T)]+[-iX\sigma_3+3TB^{[1]}(X,T),B^{[2]}(X,T)]\vspace{0.05in}\\
&+[4i\sigma_3,\frac{X}{4}B^{[2]}(X,T)+3TB^{[0]}(X,T)]\vspace{0.05in}\\
=&0.
\end{aligned}
\end{align}

Then, the coefficient of $\Lambda^{2}$ is equal to zero.

The coefficient of $\Lambda^{1}$ is
\begin{align}\no
\begin{aligned}
&C^{[1]}_T(X,T)-2B^{[2]}(X,T)+[C^{[1]}(X,T),B^{[0]}(X,T)]+[C^{[0]}(X,T),B^{[1]}(X,T)]\vspace{0.05in}\\
&+[C^{[-1]}(X,T),B^{[2]}(X,T)]+[4i\sigma_3,C^{[-2]}(X,T)],
\end{aligned}
\end{align}
which is equivalent to
\begin{align}\label{1eq}
\begin{aligned}
&\widehat{q}+X\widehat{q}_X+3T\widehat{q}_T-2\beta=0,\vspace{0.05in}\\
&-\widehat{q}^*-X\widehat{q}^*_X-3T\widehat{q}_T^*+2\beta^*=0.
\end{aligned}
\end{align}

The coefficient of $\Lambda^{0}$ is
\begin{align}\no
\begin{aligned}
&2B^{[1]}(X,T)+3TB_T^{[1]}(X,T)+[-iX\sigma_3,B^{[0]}(X,T)]+[\frac{X}{4}B^{[2]}(X,T),B^{[1]}(X,T)]\vspace{0.05in}\\
&+[C^{[-2]}(X,T),B^{[2]}(X,T)],
\end{aligned}
\end{align}
which is equivalent to
\begin{align}\label{2eq}
\begin{aligned}
&2|\widehat{q}|^2+X(\widehat{q}\widehat{q}_X^*+\widehat{q}_X\widehat{q}^*)-2(\beta \widehat{q}^*+\widehat{q}\beta^*)+3T(\widehat{q}_T\widehat{q}^*+\widehat{q}\widehat{q}_T^*)=0,\vspace{0.05in}\\
&2\widehat{q}_X
+X\widehat{q}_{XX}+3T\widehat{q}_{XT}
+4\alpha \widehat{q}=0.
\end{aligned}
\end{align}

The coefficient of $\Lambda^{-1}$ is
\begin{align}\no
\begin{aligned}
&\frac{X}{4}B_T^{[2]}(X,T)+3B^{[0]}(X,T)+3TB_T^{[0]}(X,T)+\frac{X}{4}[B^{[2]}(X,T),B^{[0]}(X,T)]\vspace{0.05in}\\
&+[C^{[-2]}(X,T),B^{[1]}(X,T)],
\end{aligned}
\end{align}
which is equivalent to
\begin{align}\label{3eq}
\begin{aligned}
&3(\widehat{q}_X^*\widehat{q}-\widehat{q}_X\widehat{q}^*)+3T(\widehat{q}_{XT}^*\widehat{q}+\widehat{q}_X^*\widehat{q}_T-\widehat{q}_{XT}\widehat{q}^*-\widehat{q}_X\widehat{q}_T^*)+X(\widehat{q}_{XX}^*\widehat{q}-\widehat{q}_{XX}\widehat{q}^*)+2\widehat{q}_X\beta^*-2\beta \widehat{q}_X^*=0,\vspace{0.05in}\\
&X\widehat{q}_T-3(\widehat{q}_{XX}+2|\widehat{q}|^2\widehat{q})-3T(\widehat{q}_{XXT}+4\widehat{q}_T|\widehat{q}|^2+2\widehat{q}^2\widehat{q}_T^*)+X(2|\widehat{q}|^2\widehat{q}_X-2\widehat{q}_X^*\widehat{q}^2)+4\beta|\widehat{q}|^2-4\alpha \widehat{q}_X=0.
\end{aligned}
\end{align}

The coefficient of $\Lambda^{-2}$ is
\bee\no
C_T^{[-2]}(X,T)+[C^{[-2]}(X,T),B^{[0]}(X,T)],
\ene
which is equivalent to
\begin{align}\label{4eq}
\begin{aligned}
&\alpha_T+\beta(\widehat{q}_{XX}^*+2|\widehat{q}|^2\widehat{q}^*)+\beta^*(\widehat{q}_{XX}+2|\widehat{q}|^2\widehat{q})=0,\vspace{0.05in}\\
&\beta_T-2\alpha(\widehat{q}_{XX}+2|\widehat{q}|^2\widehat{q})+2\beta(-\widehat{q}_X^*\widehat{q}+\widehat{q}_X\widehat{q}^*)=0.
\end{aligned}
\end{align}

According to Proposition
\ref{prop-real}, we have
\begin{align}\label{5eq}
\begin{aligned}
&\widehat{q}+X\widehat{q}_X+3T\widehat{q}_T-2\beta=0,\vspace{0.05in}\\
&\widehat{q}^2+X\widehat{q}\widehat{q}_X-(\beta \widehat{q}+\widehat{q}\beta^*)+3T\widehat{q}_T\widehat{q}=0,\vspace{0.05in}\\
&2\widehat{q}_X+X\widehat{q}_{XX}+3T\widehat{q}_{XT}+4\widehat{q}\alpha=0,\vspace{0.05in}\\
&\widehat{q}_X\beta^*-\beta \widehat{q}_X=0,\vspace{0.05in}\\
&X\widehat{q}_T-3(\widehat{q}_{XX}+2\widehat{q}^3)-3T(\widehat{q}_{XXT}+6\widehat{q}_T\widehat{q}^2)+4\beta \widehat{q}^2-4\alpha \widehat{q}_X=0,\vspace{0.05in}\\
&\alpha_T+(\beta^*+\beta)(\widehat{q}_{XX}+2\widehat{q}^3)=0,\vspace{0.05in}\\
&\beta_T-2\alpha(\widehat{q}_{XX}+2\widehat{q}^3)=0.
\end{aligned}
\end{align}

Simplifying Eqs.~(\ref{5eq}), yields
\begin{align}\label{6eq}
\begin{aligned}
&\widehat{q}+X\widehat{q}_X+3T\widehat{q}_T-2\beta=0,\vspace{0.05in}\\
&4\alpha \widehat{q}+2\beta_X=0,\vspace{0.05in}\\
&X\widehat{q}_T-3(\widehat{q}_{XX}+2\widehat{q}^3)-3T(\widehat{q}_{XXT}+6\widehat{q}_T\widehat{q}^2)+4\beta \widehat{q}^2-4\alpha \widehat{q}_X=0,\vspace{0.05in}\\
&\alpha_T+2\beta(\widehat{q}_{XX}+2\widehat{q}^3)=0,\vspace{0.05in}\\
&\beta_T-2\alpha(\widehat{q}_{XX}+2\widehat{q}^3)=0,\vspace{0.05in}\\
&\alpha(X,T)^2+\beta(X,T)^2=4.
\end{aligned}
\end{align}
Thus the proof of Theorem \ref{new-th-2} is completed.

\section{Asymptotic behaviors of the near-field limit in $(X, T)$-space}

\subsection{Asymptotic behavior
of $\widehat{q}^{\pm}(X,T)$ for large $X$}

We will study the asymptotic behavior of $\widehat{q}^{\pm}$ when $X$ is large. To this end, let's make the following transformation:
\bee\no
X:=\sigma|X|\,\, (\sigma=\pm 1),\quad T:=a|X|^2,\quad \Lambda:=|X|^{-\frac12}z.
\ene

The phase conjugating the jump matrix for $D^{\pm}(\Lambda;X,T)$ can be rewritten as
\bee
\Lambda X+4\Lambda^3T\pm2\Lambda^{-1}=|X|^{\frac12}(\sigma z+4az^3\pm2z^{-1}).
\ene

We need to study four situations, namely $|X|^{\frac12}(z+4az^3+2z^{-1})$, $|X|^{\frac12}(z+4az^3-2z^{-1})$, $|X|^{\frac12}(-z+4az^3+2z^{-1})$, and $|X|^{\frac12}(-z+4az^3-2z^{-1})$. According to Propositions \ref{sym1} and \ref{sym-prop2}, it is sufficient to consider $\widehat{q}^+(X,T)$ as $X\rightarrow\infty$, namely, $\sigma=1$. Defining
\bee
Y(z;X,a):=D^{+}(X^{-\frac12}z;X,aX^2),\quad X>0.
\ene

\begin{prop} Riemann-Hilbert Problem 6.

Let $(X,T)$ be in compact subsets of $\mathbb{R}^2$ and $\Sigma_2$ is an arbitrary Jordan curve surrounding $z=0$. Find a $2\times 2$ matrix $Y(z;X,a)$ that satisfies the following properties:

\begin{itemize}

\item{} Analyticity: $Y(z;X,a)$ is analytic in $\mathbb{C}\setminus\Sigma_2$ and takes continuous boundary values on $\Sigma_2$.

\item{} Jump conditions: Assuming clockwise orientation of $\Sigma_2$, the boundary values on the jump contour $\Sigma_2$ are related as:
\bee\label{RH45}
Y_+(z;X,a)=Y_-(z;X,a)e^{-iX^{\frac12}(z+4az^3+2z^{-1})\sigma_3}\widehat{Q}^{-1}e^{iX^{\frac12}(z+4az^3+2z^{-1})\sigma_3},\quad z\in\Sigma_2.
\ene

\item{} Normalization: $Y(z;X,a)$ tends to the identity matrix as $z\rightarrow\infty$.

\end{itemize}
\end{prop}

Using Eq.~(\ref{fanyan-q2}), we have
\bee\label{fanyan-Y}
\widehat{q}^{+}(X,aX^2)=2iX^{-\frac12}\lim_{z\rightarrow\infty}z Y_{12}(z;X,a),\quad X>0.
\ene

\begin{figure}[!t]
    \centering
 \vspace{-0.15in}
  {\scalebox{0.68}[0.68]{\includegraphics{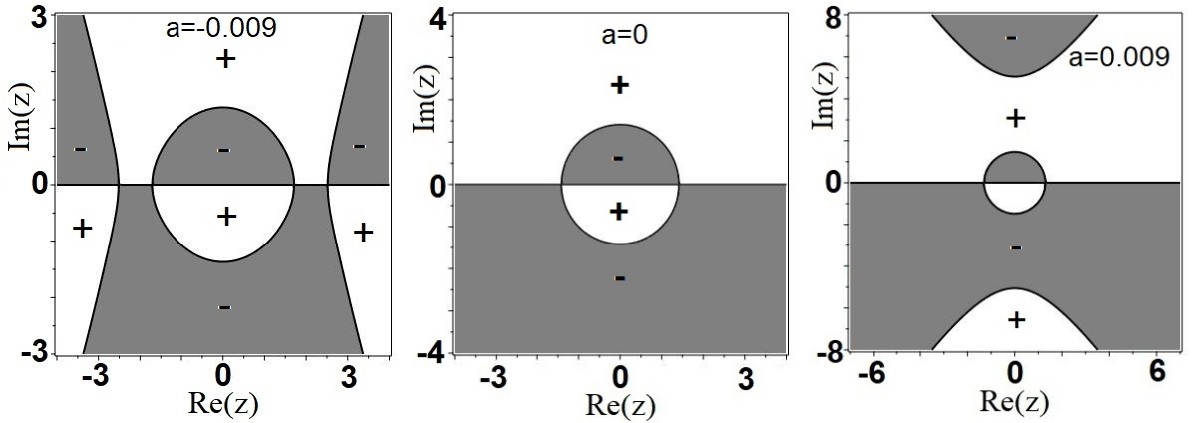}}}%
  \hspace{-0.15in}
\vspace{0.1in}
\caption{Sign charts of $\mathrm{Im}(\Theta(z;a)$ when $a=-0.009,0,0.009$.}
   \label{Sign charts}
\end{figure}

\subsubsection{Exponent analysis and steepest descent}

Let
\bee
\Theta(z;a):=z+4az^3+2z^{-1}.
\ene
Then one has
\begin{itemize}

\item[i)] As $-\frac{1}{96}<a<0$, the function $\Theta(z;a)$ has four different real critical points, namely,
\bee
b_1(a):=\sqrt{\frac{-1-\sqrt{1+96a}}{24a}}, \quad -b_1(a),\quad
 b(a):=\sqrt{\frac{-1+\sqrt{1+96a}}{24a}},\quad -b(a).
\ene

\item[ii)] As $a=0$, the function $\Theta(z;a)$ has double simple real critical points, namely,
\bee b(a):=\sqrt{2}, \qquad
  -b(a);
  \ene

\item[iii)] As
$0<a<\frac{1}{32}$, the function $\Theta(z;a)$ has double simple real critical points, namely,
\bee
b(a):=\sqrt{\frac{-1+\sqrt{1+96a}}{24a}},\qquad -b(a).
\ene
Moreover, the critical points of $\Theta(z;a)$ are its real roots.
\end{itemize}

There exists a component of the curve $\mathrm{Im}(\Theta(z;a))=0$ that is a Jordan curve surrounding the origin in the $z$-plane and that passes through two different real critical points, namely, $-b(a)$ and $b(a)$. We select this curve as the jump contour $\Sigma_2$ for $Y(z;X,a)$. The real axis divides $\Sigma_2$ into an arc $\Sigma_2^+$ in the upper half-plane and an arc $\Sigma_2^-$ in the lower half-plane. We introduce thin lens-shaped domains $L^{\pm}$ and $R^{\pm}$ on the left-hand and right-hand sides, respectively, of $\Sigma_2^{\pm}$ whose outer boundary arcs $\Sigma_3^{\pm}$ and $\Sigma_4^{\pm}$ meet the real axis at $\frac{\pi}{4}$ angles as shown in the left-hand panel of Figure \ref{Sign charts}. The region between $\Sigma_4^{\pm}$ and the real axis is denoted $\Omega^{\pm}$. Let the line segment $I:=[-b(a),b(a)]$. Then, we will make the following transformation:
\begin{eqnarray}\label{WY-1}
&&W(z;X,a):=Y(z;X,a)\left[\!\!\begin{array}{cc} 1& 0  \vspace{0.05in}\\
e^{2iX^{\frac12}\Theta(z;a)} & 1
\end{array}\!\!\right],\quad z\in L^+,\vspace{0.05in} \no \\
&&W(z;X,a):=Y(z;X,a)\left[\!\!\begin{array}{cc}
\sqrt{2}& \frac{\sqrt{2}}{2}e^{-2iX^{\frac12}\Theta(z;a)}  \vspace{0.05in}\\
0 & \frac{\sqrt{2}}{2}
\end{array}\!\!\right],\quad z\in R^+,\vspace{0.05in} \no\\
&&W(z;X,a):=Y(z;X,a)2^{\frac{\sigma_3}{2}},\quad z\in \Omega^+,\vspace{0.05in} \no\\
&&W(z;X,a):=Y(z;X,a)2^{-\frac{\sigma_3}{2}},\quad z\in \Omega^-,\vspace{0.05in} \label{WY-1}\\
&&W(z;X,a):=Y(z;X,a)\left[\!\!\begin{array}{cc}
\frac{\sqrt{2}}{2}& 0  \vspace{0.05in}\\
-\frac{\sqrt{2}}{2}e^{2iX^{\frac12}\Theta(z;a)} & \sqrt{2}
\end{array}\!\!\right],\quad z\in R^-,\vspace{0.05in} \no\\
&&W(z;X,a):=Y(z;X,a)\left[\!\!\begin{array}{cc}
1& -e^{-2iX^{\frac12}\Theta(z;a)}  \vspace{0.05in}\\
0 & 1
\end{array}\!\!\right],\quad z\in L^-,\vspace{0.05in}\no\\
&&W(z;X,a):=Y(z;X,a),\quad \mathrm{otherwise}. \no
\end{eqnarray}

\begin{figure}[!t]
    \centering
 \vspace{-0.15in}
  {\scalebox{0.3}[0.3]{\includegraphics{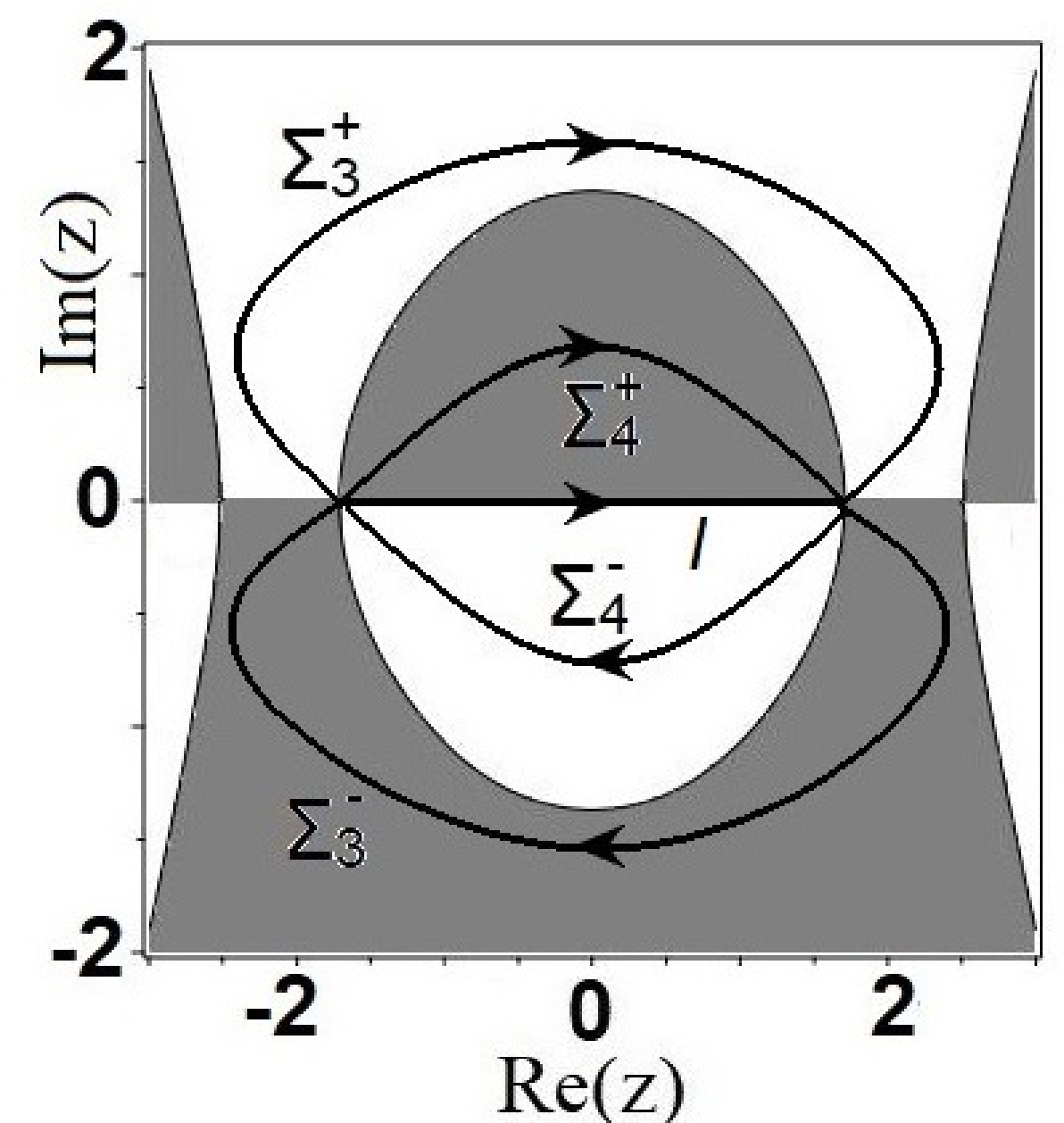}}}\hspace{-0.35in}
\vspace{-0.05in}
\caption{The jump contour $\Sigma_2=\Sigma_2^+\cup\Sigma_2^-$ for $Y(z;X,a)$ and the regions $L^{\pm},R^{\pm}$ and $\Omega^{\pm}$.}
   \label{Big-X}
\end{figure}

We can easily obtain that $W(z;X,a)$ has no jump on arc $\Sigma_2$, namely, $W_+(z;X,a)=W_-(z;X,a),\quad z\in\Sigma_2$. Below we give the jump properties of $W(z;X,a)$ on arcs $\Sigma_3,\Sigma_4$ and line segment $I$.
\begin{eqnarray}
&&W_+(z;X,a)=W_-(z;X,a)\left[\!\!\begin{array}{cc}
1& 0  \vspace{0.05in}\\
-e^{2iX^{\frac12}\Theta(z;a)} & 1
\end{array}\!\!\right],\quad z\in \Sigma_3^+,\vspace{0.05in} \no\\
&&W_+(z;X,a)=W_-(z;X,a)\left[\!\!\begin{array}{cc}
1& \frac12e^{-2iX^{\frac12}\Theta(z;a)} \vspace{0.05in}\\
0 & 1
\end{array}\!\!\right],\quad z\in \Sigma_4^+,\vspace{0.05in} \no\\
&&W_+(z;X,a)=W_-z;X,a)2^{\sigma_3},\qquad\qquad\qquad
\qquad\qquad z\in I,\vspace{0.05in}
\label{jump-W} \\
&&W_+(z;X,a)=W_-(z;X,a)\left[\!\!\begin{array}{cc}
1& 0 \vspace{0.05in}\\
-\frac12e^{2iX^{\frac12}\Theta(z;a)} & 1
\end{array}\!\!\right],\quad z\in \Sigma_4^-,\vspace{0.05in} \no\\
&&W_+(z;X,a)=W_-(z;X,a)\left[\!\!\begin{array}{cc}
1& e^{-2iX^{\frac12}\Theta(z;a)}  \vspace{0.05in}\\
0 & 1
\end{array}\!\!\right],\quad z\in \Sigma_3^-. \no
\end{eqnarray}

Since $\mathrm{Im}(\Theta(z;a))>0$ holds on $\Sigma_3^{+},\Sigma_4^-$ while $\mathrm{Im}(\Theta(z;a))<0$ holds on $\Sigma_3^-,\Sigma_4^{+}$,  the jump matrices are exponentially decreasing on these four contour arcs except near the endpoints $b(a)$ and $-b(a)$. Next, we will find local matrix functions defined near $z=b(a),z=-b(a)$ that exactly satisfy the jump conditions (\ref{jump-W}).

\subsubsection{Parametrix construction}

Let
\bee\no
\overline{W}^o(z;a)=\left(\frac{z+b(a)}{z-b(a)}\right)^{ip\sigma_3},\quad z\in\mathbb{C}\setminus I,
\ene
where $p=\frac{\ln 2}{2\pi}$. Then, we get the jump condition of the function $\overline{W}^o(z;a)$ on the line segment $I$:
\bee\no
\overline{W}_+^o(z;a)=\overline{W}_-^o(z;a)2^{\sigma_3},\quad z\in I.
\ene

Note that
\bee\no
\Theta'(-b(a);a)=\Theta'(b(a);a)=0,\qquad \Theta''(-b(a);a)<0,\qquad
\Theta''(b(a);a)>0,
\ene
we define the following conformal mappings $g_{-b}(z;a)$ and $g_b(z;a)$:
\begin{align}\no
\begin{aligned}
&g_{-b}(z;a)^2=-2\left(\Theta(b(a);a)+\Theta(z;a)\right),\vspace{0.05in}\\
&g_{b}(z;a)^2=2\left(\Theta(z;a)-\Theta(b(a);a)\right)
\end{aligned}
\end{align}
and we choose the case that $g'_{b}(b(a);a)>0$ and $g'_{-b}(-b(a);a)<0$. Let $\xi_{-b}:=X^{\frac14}g_{-b}(z;a)$ and $\xi_{b}:=X^{\frac14}g_{b}(z;a)$. Then the jump conditions are satisfied by
\bee\no
U^{-b}:=iW(z;X,a)e^{iX^{\frac12}\Theta(b(a);a)\sigma_3}\sigma_2,\quad~\mathrm{near}~z=-b
\ene
and by
\bee\no
U^{b}:=W(z;X,a)e^{-iX^{\frac12}\Theta(b(a);a)\sigma_3}\quad~\mathrm{near}~z=b.
\ene
Below we give the jump properties of $U^{-b}$ and $U^{b}$:
\begin{eqnarray}
&&U^{b}_+=U^{b}_-\left[\!\!\begin{array}{cc}
1& 0  \vspace{0.05in}\\
e^{i\xi_{b}^2} & 1
\end{array}\!\!\right],\quad z\in \Sigma_3^+ \,(\mathrm{away~from}~b(a))\vspace{0.05in} \no\\
&&U^{b}_+=U^{b}_-\left[\!\!\begin{array}{cc}
1& \frac12e^{-i\xi_{b}^2} \vspace{0.05in}\\
0 & 1
\end{array}\!\!\right],\quad z\in \Sigma_4^+\,(\mathrm{toward}~b(a))\vspace{0.05in} \no \\
&&U^{b}_+=U^{b}_-2^{\sigma_3},\qquad\qquad\qquad z\in I \,(\mathrm{toward}~b(a))\vspace{0.05in} \no\\
&&U^{b}_+=U^{b}_-\left[\!\!\begin{array}{cc}
1& 0 \vspace{0.05in}\\
\frac12e^{i\xi_{b}^2} & 1
\end{array}\!\!\right],\quad z\in \Sigma_4^- \,(\mathrm{toward}~b(a))\vspace{0.05in}\\
&&U^{b}_+
=U^{b}_-\left[\!\!\begin{array}{cc}
1& e^{-i\xi_{b}^2}  \vspace{0.05in}\\
0 & 1
\end{array}\!\!\right],\quad z\in \Sigma_3^- \,(\mathrm{away~from}~b(a)).\no
\end{eqnarray}
and
\begin{align}
\begin{aligned}
&U^{-b}_+=U^{-b}_-\left[\!\!\begin{array}{cc}
1& e^{-i\xi_{-b}^2}  \vspace{0.05in}\\
0 & 1
\end{array}\!\!\right],\quad z\in \Sigma_3^+\,(\mathrm{away~from}-b(a))\vspace{0.05in}\\
&U^{-b}_+=U^{-b}_-\left[\!\!\begin{array}{cc}
1& 0 \vspace{0.05in}\\
\frac12e^{i\xi_{-b}^2} & 1
\end{array}\!\!\right],\quad z\in \Sigma_4^+ \,(\mathrm{toward}-b(a))\vspace{0.05in}\\
&U^{b}_+=U^{b}_-2^{\sigma_3},\qquad\qquad\qquad z\in I \,(\mathrm{toward}-b(a))\vspace{0.05in}\\
&U^{-b}_+=U^{-b}_-\left[\!\!\begin{array}{cc}
1& 0 \vspace{0.05in}\\
e^{i\xi_{-b}^2}& 1
\end{array}\!\!\right],\quad z\in \Sigma_3^- \,(\mathrm{away~from}-b(a))\vspace{0.05in}\\
&U^{-b}_+=U^{-b}_-\left[\!\!\begin{array}{cc}
1& \frac12e^{-i\xi_{-b}^2}  \vspace{0.05in}\\
0 & 1
\end{array}\!\!\right],\quad z\in \Sigma_4^-,(\mathrm{toward}-b(a)).
\end{aligned}
\end{align}

We can normalize the jump matrix near point $-b(a)$ and the jump matrix near point $b(a)$ into the jump matrix about $\xi$.
\begin{figure}[!t]
    \centering
 \vspace{-0.15in}
  {\scalebox{0.35}[0.35]{\includegraphics{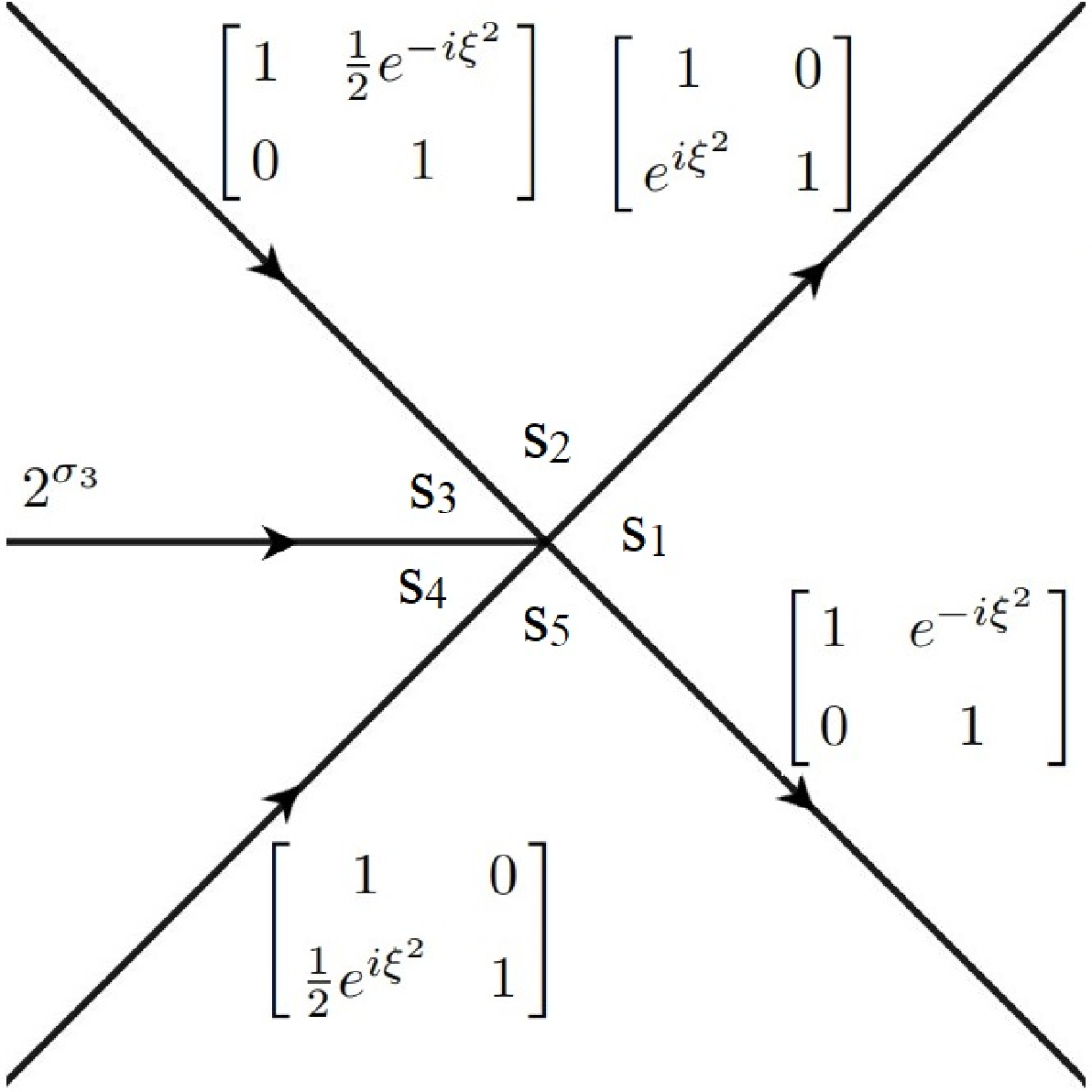}}}  
\vspace{0.05in}
\caption{Regional division for $U(\xi)$.}
   \label{xi}
\end{figure}

\begin{prop}\label{RH7} Riemann-Hilbert Problem 7\\
Find a $2\times 2$ matrix $U(\xi)$ that satisfies the following properties:

\begin{itemize}

\item{} Analyticity: $U(\xi)$ is analytic for $\xi$ in the five regions shown in Figure \ref{xi}, namely, $S_1: |\mathrm{arg}(\xi)|<\frac{\pi}{4}$, $S_2: \frac{\pi}{4}<\mathrm{arg}(\xi)<\frac{3\pi}{4}$, $S_3: \frac{3\pi}{4}<\mathrm{arg}(\xi)<\pi$, $S_4:-\pi<\mathrm{arg}(\xi)<-\frac{3\pi}{4}$, $S_5:-\frac{3\pi}{4}<\mathrm{arg}(\xi)<-\frac{\pi}{4}$ where $-\pi<\mathrm{arg}(\xi)\leq\pi$. It takes continuous boundary values on the excluded rays and at the origin from each sector.

\item{} Jump conditions: The boundary values on the jump contour are related as:
\begin{eqnarray}
&&U_+(\xi)=U_-(\xi)\left[\!\!\begin{array}{cc}
1& 0  \vspace{0.05in}\\
e^{i\xi^2} & 1
\end{array}\!\!\right],\quad \mathrm{arg}(\xi)=\frac{\pi}{4},\vspace{0.05in} \no\\
&&U_+(\xi)=U_-(\xi)\left[\!\!\begin{array}{cc}
1& \frac12e^{-i\xi^2} \vspace{0.05in}\\
0 & 1
\end{array}\!\!\right],\quad \mathrm{arg}(\xi)=\frac{3\pi}{4},\vspace{0.05in} \no\\
&&U_+(\xi)=U_-(\xi)2^{\sigma_3},\qquad \qquad \quad \mathrm{arg}(\xi)=\pi,\vspace{0.05in}\\
&&U_+(\xi)=U_-(\xi)\left[\!\!\begin{array}{cc}
1& 0 \vspace{0.05in}\\
\frac12e^{i\xi^2} & 1
\end{array}\!\!\right],\quad \mathrm{arg}(\xi)=-\frac{3\pi}{4},\vspace{0.05in}\no\\
&&U_+(\xi)=U_-(\xi)\left[\!\!\begin{array}{cc}
1& e^{-i\xi^2}  \vspace{0.05in}\\
0 & 1
\end{array}\!\!\right], \quad \mathrm{arg}(\xi)=-\frac{\pi}{4}.\no
\end{eqnarray}

\item{} Normalization: $U(\xi)\xi^{ip\sigma_3}$ tends to identity matrix as $\xi\rightarrow\infty$, where $p=\frac{\ln(2)}{2\pi}$.

\end{itemize}
\end{prop}

The solution of the Riemann-Hilbert Problem 7 can be expressed explicitly according to the parabolic cylinder function.  In particular, the solution satisfies:
\bee\label{RH7-solu}
U(\xi)\xi^{ip\sigma_3}=\mathbb{I}+\frac{1}{2i\xi}\left[\!\!\begin{array}{cc}
0& \gamma  \vspace{0.05in}\\
\gamma^* & 0
\end{array}\!\!\right]+\mathcal{O}(\xi^{-2}),\quad \xi\rightarrow\infty,
\ene
where
\bee\no
\gamma=\dfrac{2^{\frac54}\pi^{\frac12}e^{\frac{i\pi}{4}}e^{\frac{i(\ln(2))^2}{2\pi}}}{\Gamma(\frac{i\ln(2)}{2\pi})},\quad |\gamma|=\sqrt{2p}.
\ene

Let
\begin{align}\label{T-bTb}
\begin{aligned}
&T^{-b}(z;a):=i(b-z)^{-ip\sigma_3}\left(-\frac{b+z}{g_{-b}(z;a)}\right)^{ip\sigma_3}\sigma_2,\vspace{0.05in}\\
&T^{b}(z;a):=(b+z)^{ip\sigma_3}\left(\frac{g_{b}(z;a)}{z-b}\right)^{ip\sigma_3}.
\end{aligned}
\end{align}

According to Proposition \ref{RH7}, we can obtain the following local functions:
\begin{align}
\begin{aligned}
\overline{W}^{-b}(z;X,a)=&-iX^{-\frac{ip\sigma_3}{4}}e^{iX^{\frac12}\Theta(b;a)\sigma_3}T^{-b}(z;a) U(X^{\frac14}g_{-b}(z;a))\sigma_2e^{-iX^{\frac12}\Theta(b;a)\sigma_3},\quad z\in\delta_{-b}(r),\vspace{0.05in}\\
\overline{W}^{b}(z;X,a)=&X^{\frac{ip\sigma_3}{4}}e^{-iX^{\frac12}\Theta(b;a)\sigma_3}T^{b}(z;a)U(X^{\frac14}g_{b}(z;a))e^{iX^{\frac12}\Theta(b;a)\sigma_3},\quad z\in\delta_{b}(r),
\end{aligned}
\end{align}
where $\delta_{-b}(r)$ and $\delta_{b}(r)$ are small circles with $-b(a)$ as the center and $r$ as the radius, and small circles with $b(a)$ as the center and $r$ as the radius, respectively. When $-\frac{1}{96}<a<0$, we define the following piecewise function:
\bee\label{piecewise-W}
\overline{W}(z;X,a):=\begin{cases}
\overline{W}^{-b}(z;X,a),\quad z\in\delta_{-b}(r),\vspace{0.05in}\\
\overline{W}^{b}(z;X,a),\quad z\in\delta_{b}(r),\vspace{0.05in}\\
\overline{W}^o(z;a),\quad z\in\mathbb{C}\setminus(\overline{\delta_{-b}(r)}\cup\overline{\delta_{b}(r)}\cup I).
\end{cases}
\ene

\subsubsection{Error analysis}

We consider the following error function:
\bee\label{EW-1}
E(z;X,a):=W(z;X,a)\overline{W}(z;X,a)^{-1}.
\ene

Since $W(z;X,a)$ and $\overline{W}(z;X,a)$ have the same jump property on the line segment $I$, $E(z;X,a)$ has no jump across $I$. Error function $E(z;X,a)$ is analytic in $\mathbb{C}\setminus\Sigma_5$ and takes continuous boundary values on $\Sigma_5$, where $\Sigma_5=\left(\Sigma_3^{\pm}\cup\Sigma_4^{\pm}\cup\partial\delta_{-b}(r)\cup\partial\delta_{b}(r)\right)\cap\complement_{\mathbb{C}}{(\delta_{-b}(r)
\cup\delta_{b}(r))}$. Assuming clockwise orientation of $\delta_{-b}(r)$ and $\delta_{b}(r)$, the boundary values on the jump contour $\Sigma_5$ are related as:
\bee\no
E_+(z;X,a):=E_-(z;X,a)V^E(z;X,a).
\ene

Next, we will analyze the error factors of $V^E(z;X,a)$ on each arcs of $\Sigma_5$.

\begin{prop}\label{error1}
The following identity holds:\\
(I)\bee
\sup_{z\in\left(\Sigma_3^{\pm}\cup\Sigma_4^{\pm}\right)\cap\complement_{\mathbb{C}}{(\overline{\delta_{-b}(r)}\cup\overline{\delta_{b}(r)})}}
||V^E(z;X,a)-\mathbb{I}||=\mathcal{O}(e^{-C(a)X^{\frac12}}),\quad X\rightarrow+\infty.
\ene
where $C(a)$ is a constant about $a$ and $||\cdot||$ denotes the matrix norm.

(II)\bee
\sup_{z\in\partial\delta_{-b}(r)\cup\partial\delta_{b}(r)}||V^E(z;X,a)-\mathbb{I}||=\mathcal{O}(X^{-\frac14}),\quad X\rightarrow+\infty.
\ene

\end{prop}

\begin{proof}
When $z\in\left(\Sigma_3^{\pm}\cup\Sigma_4^{\pm}\right)\cap\complement_{\mathbb{C}}{(\overline{\delta_{-b}(r)}\cup\overline{\delta_{b}(r)})}$, the error term depends on the jump matrix of Eqs.~(\ref{jump-W}), then we get formula (I). Note that
\begin{align}\label{W-bwb}
\begin{aligned}
\overline{W}^{-b}(z;X,a)\overline{W}^{o}(z;a)^{-1}=&X^{-\frac{ip\sigma_3}{4}}e^{iX^{\frac12}\Theta(b(a);a)\sigma_3}T^{-b}(z;a)U(\xi_{-b})\vspace{0.05in}\\
&\times\xi_{-b}^{ip\sigma_3}T^{-b}(z;a)^{-1}e^{-iX^{\frac12}\Theta(b(a);a)\sigma_3}X^{\frac{ip\sigma_3}{4}},\quad z\in\partial\delta_{-b}(r),\vspace{0.05in}\\
\overline{W}^{b}(z;X,a)\overline{W}^{o}(z;a)^{-1}=&X^{\frac{ip\sigma_3}{4}}e^{-iX^{\frac12}\Theta(b(a),a)\sigma_3}T^{b}(z;a)U(\xi_{b})\vspace{0.05in}\\
&\times\xi_{b}^{ip\sigma_3}T^{b}(z;a)^{-1}e^{iX^{\frac12}\Theta(b(a);a)\sigma_3}X^{-\frac{ip\sigma_3}{4}},\quad z\in\partial\delta_{b}(r).
\end{aligned}
\end{align}

When $z\in\partial\delta_{-b}(r)\cup\partial\delta_{b}(r)$, the error term depends on $U(\xi_{-b})\xi_{-b}^{ip\sigma_3}$ and $U(\xi_{b})\xi_{b}^{ip\sigma_3}$. According to Eq.~(\ref{RH7-solu}), we get formula (II). Thus the proof is completed.

\end{proof}

Using Plemelj formula, we have
\bee\label{E-ple}
E(z;X,a)=\mathbb{I}+\frac{1}{2\pi i}\int_{\Sigma_5}\dfrac{E_-(v;X,a)(V^E(v;X,a)-\mathbb{I})}{v-z}dv,\quad z\in\mathbb{C}\setminus\Sigma_5.
\ene

According to Proposition \ref{error1} and Eq.~(\ref{E-ple}), we have
\bee\label{E-err}
E_-(z;X,a)-\mathbb{I}=\mathcal{O}(X^{-\frac14}),\quad X\rightarrow+\infty.
\ene

We can get the following Laurent expansion of $E(z;X,a)$ convergent for sufficiently large $|z|$:
\bee\label{E-lau}
E(z;X,a)=\mathbb{I}-\frac{1}{2\pi i}\sum_{j=1}^{\infty}z^{-j}\int_{\Sigma_5}E_-(v;X,a)(V^E(v;X,a)-\mathbb{I})v^{j-1}dv.
\ene

Substituting Eq.~(\ref{EW-1}) into Eq.~(\ref{fanyan-Y}), we have
\bee\label{fanyan-E}
\widehat{q}^+(X,aX^2)=2iX^{-\frac12}\lim_{z\rightarrow\infty}z E_{12}(z;X,a),\quad X>0.
\ene

Substituting Eq.~(\ref{E-lau}) into Eq.~(\ref{fanyan-E}), we have
\bee\label{q-E-simi}
\widehat{q}^+(X,aX^2)=-\frac{1}{\pi X^{\frac12}}\int_{\Sigma_5}E_{11-}(v;X,a)V^{E}_{12}(v;X,a)+E_{12-}(v;X,a)(V^{E}_{22}(v;X,a)-1)dv.
\ene

According to Eq.~(\ref{E-err}), we have
\bee\label{fanyan-VE12}
q^+(X,aX^2)=-\frac{1}{\pi X^{\frac12}}\int_{\Sigma_5}V^E_{12}(v;X,a)dv+\mathcal{O}(X^{-1}),\quad X\rightarrow+\infty.
\ene

Using Eqs.~(\ref{RH7-solu}) and (\ref{W-bwb}), we can obtain the asymptotic expression of $V^{E}_{12}(z;X,a)$:
\begin{align}\label{VE12}
\begin{aligned}
&V^{E}_{12}(z;X,a)=\dfrac{X^{-\frac{ip}{2}}e^{2iX^{\frac12}\Theta(b;a)}(\gamma T^{-b}_{11}(z;a)^2-\gamma^*T^{-b}_{12}(z;a)^2)}{2iX^{\frac14}g_{-b}(z;a)}+\mathcal{O}(X^{-\frac12}),\quad z\in\partial\delta_{-b}(r),\vspace{0.05in}\\
&V^{E}_{12}(z;X,a)=\dfrac{X^{\frac{ip}{2}}e^{-2iX^{\frac12}\Theta(b;a)}(\gamma T^{b}_{11}(z;a)^2-\gamma^*T^{b}_{12}(z;a)^2)}{2iX^{\frac14}g_{b}(z;a)}+\mathcal{O}(X^{-\frac12}),\quad z\in\partial\delta_{b}(r).
\end{aligned}
\end{align}

According to above analysis, we can provide the proof of Proposition \ref{Big-X-q}.

\begin{proof}
Note that $g_{-b}^{'}(-b(a);a)=-g_{b}^{'}(b(a);a)=-\sqrt{24ab+4b^{-3}}$. Substituting Eq.~(\ref{VE12}) into Eq.~(\ref{fanyan-VE12}), we have
\begin{align}
\begin{aligned}
\widehat{q}^+(X,aX^2)=&\frac{1}{\sqrt{24ab+4b^{-3}}X^{\frac34}}\bigg(X^{\frac{ip}{2}}e^{-2iX^{\frac12}\Theta(b;a)}(\gamma T^{b}_{11}(b(a);a)^2-\gamma^*T^{b}_{12}(b(a);a)^2)\vspace{0.05in}\\
&-X^{-\frac{ip}{2}}e^{2iX^{\frac12}\Theta(b;a)}(\gamma T^{-b}_{11}(-b(a);a)^2-\gamma^*T^{-b}_{12}(-b(a);a)^2\bigg)+\mathcal{O}(X^{-1}),\quad X\rightarrow+\infty.
\end{aligned}
\end{align}

According to Eq.~(\ref{T-bTb}), we have
\begin{align}
\begin{aligned}
&T^{-b}(-b(a);a)=i(2b\sqrt{24ab+4b^{-3}})^{-ip\sigma_3}\sigma_2,\vspace{0.05in}\\
&T^{b}(b(a);a)=(2b\sqrt{24ab+4b^{-3}})^{ip\sigma_3}.
\end{aligned}
\end{align}

Therefore,
\begin{align}
\begin{aligned}
\widehat{q}^+(X,aX^2)=&\frac{1}{\sqrt{24ab+4b^{-3}}X^{\frac34}}\bigg(X^{\frac{ip}{2}}e^{-2iX^{\frac12}\Theta(b;a)}(\gamma T^{b}_{11}(b(a);a)^2-\gamma^*T^{b}_{12}(b(a);a)^2)\vspace{0.05in}\\
&-X^{-\frac{ip}{2}}e^{2iX^{\frac12}\Theta(b;a)}(\gamma T^{-b}_{11}(-b(a);a)^2-\gamma^*T^{-b}_{12}(-b(a);a)^2)\bigg)+\mathcal{O}(X^{-1})\vspace{0.05in}\\
=&\frac{1}{\sqrt{24ab+4b^{-3}}X^{\frac34}}\bigg(X^{\frac{ip}{2}}e^{-2iX^{\frac12}\Theta(b;a)}\gamma T^{b}_{11}(b(a);a)^2\vspace{0.05in}\\
&+X^{-\frac{ip}{2}}e^{2iX^{\frac12}\Theta(b;a)}\gamma^*T^{-b}_{12}(-b(a);a)^2\bigg)+\mathcal{O}(X^{-1})\vspace{0.05in}\\
=&\frac{\sqrt{2p}}{\sqrt{24ab+4b^{-3}}X^{\frac34}}\bigg(X^{\frac{ip}{2}}e^{-2iX^{\frac12}\Theta(b;a)}e^{i\mathrm{arg}(\gamma)}
(2b\sqrt{24ab+4b^{-3}})^{2ip}\vspace{0.05in}\\
&+X^{-\frac{ip}{2}}e^{2iX^{\frac12}\Theta(b;a)}e^{-i\mathrm{arg}(\gamma)}(2b\sqrt{24ab+4b^{-3}})^{-2ip}\bigg)+\mathcal{O}(X^{-1})\vspace{0.05in}\\
=&\frac{\sqrt{2p}\cos(\phi(X,a))}{X^{\frac34}\sqrt{6ab+b^{-3}}}+\mathcal{O}(X^{-1}),\quad X\rightarrow+\infty.
\end{aligned}
\end{align}

Thus the proof is completed.
\end{proof}

\begin{corollary}
Let $a=0$ and $\Theta(z;a):=z+2z^{-1}$. Then as $X\rightarrow+\infty$, we have
\begin{align}\label{T=0}
\begin{aligned}
\widehat{q}^+(X,0)=&\frac{2^{\frac34}
\sqrt{\ln 2}}{\sqrt{\pi}X^{\frac34}}
\cos\bigg(2^{\frac52}X^{\frac12}-\frac{\ln 2}{4\pi}\ln(X)+\mathrm{arg}\l(\Gamma\l(\frac{i\ln 2}{2\pi}\r)\r)-\frac{9(\ln 2)^2}{4\pi}-\frac{\pi}{4}\bigg)+\mathcal{O}(X^{-1}).
\end{aligned}
\end{align}

\end{corollary}

\begin{proof}
When $a=0$, the critical points $b(0)=\sqrt{2}$. Using Theorem \ref{Big-X-q}, we obtain Eq.~(\ref{T=0}). Thus the proof is completed.
\end{proof}

\subsection{Asymptotic behavior of $\widehat{q}^{\pm}(X,T)$ for large $T$: part results}

We will study the asymptotic behavior of $\widehat{q}^{\pm}(X,T)$ when $T$ is large. To this end, let us make the following transforms:
\bee
T:=\sigma|T|,\quad X:=w|T|^{\frac12},\quad \Lambda:=|T|^{-\frac14}z.
\ene

The phase conjugating the jump matrix for $D^{\pm}(\Lambda;X,T)$ can be rewritten as
\bee
\Lambda X+4\Lambda^3T\pm2\Lambda^{-1}=|T|^{\frac14}
(4\sigma z^3+wz\pm2z^{-1}).
\ene

Similarly to the asymptotic behavior for large $X$, we only need to consider the case that $|T|^{\frac14}(4z^3+wz+2z^{-1})$. Setting
\bee
Y(z;T,w):=D^{+}(T^{-\frac14}z;wT^{\frac12},T),\quad T>0.
\ene

Using Eq.~(\ref{fanyan-q2}), we have
\bee\label{fanyan-Y-T}
\widehat{q}^{+}(wT^{\frac12},T)=2iT^{-\frac14}\lim_{z\rightarrow\infty}z Y_{12}(z;T,w),\quad T>0.
\ene

Let $\Sigma_2$ be an arbitrary Jordan curve surrounding $z=0$. Assuming clockwise orientation of $\Sigma_2$, the boundary values on the jump contour $\Sigma_2$ are related as:
\bee\label{Jump-T}
Y_+(z;T,w)=Y_-(z;T,w)e^{-iT^{\frac14}\theta(z;w)\sigma_3}Q^{-1}e^{iT^{\frac14}\theta(z;w)\sigma_3},\quad z\in\Sigma_2.
\ene
where $\theta(z;w)=4z^3+wz+2z^{-1}$. When $|w|>\sqrt{32}$, the function $\theta(z;w)$ has double simple real roots. This case is equivalent to the case which two simple real root in the asymptotic behavior for large $X$. When $0<|w|<\sqrt{32}$, the function $\theta(z;a)$ has double simple real roots, namely, $\sqrt{\frac{-w+\sqrt{w^2+96}}{24}}$ and $-\sqrt{\frac{-w+\sqrt{w^2+96}}{24}}$. When $w=0$, the function $\Theta(z;a)$ has double simple real roots, namely, $6^{-\frac14}$ and $-6^{-\frac14}$. When $w=0$ and $w=2$, the sign chart of $\theta(z;w)$ are illustrated in Figure \ref{theta-w}.

\begin{figure}[!t]
    \centering
 \vspace{-0.15in}
  {\scalebox{0.5}[0.5]{\includegraphics{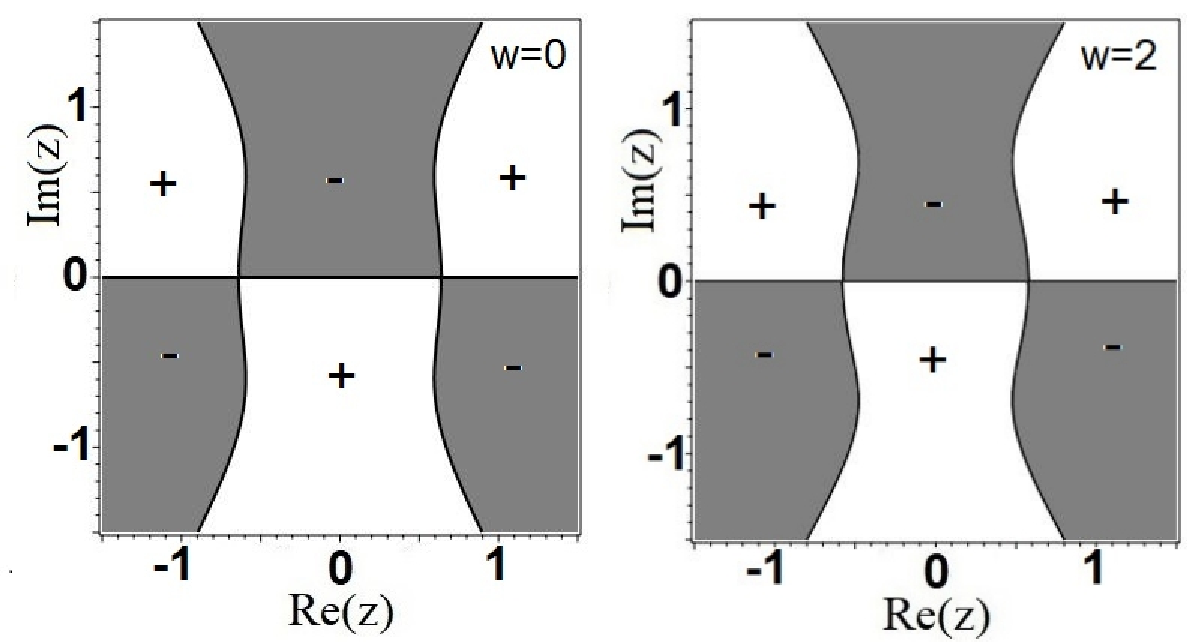}}}  
\vspace{0.1in}
\caption{Sign charts of $\theta(z;w)$ when $w=0,2$.}
   \label{theta-w}
\end{figure}

\subsubsection{The $g$-function}

Suppose that $g(z;w)$ is a scalar function, that satisfies $g(z;w)=\mathcal{O}(\frac{1}{z})$ as $z\rightarrow\infty$, and the boundary values on the jump contour are related as:
\bee
g_+(z;w)+g_-(z;w)+2\theta(z;w)=\mathrm{constant}.
\ene

It is easy to check that the function $(g_z(z;w)+\theta_z(z;w))^2$ is analytic for $z\in\mathbb{C}\setminus\!\{0\}$. Then we can obtain the asymptotic expression of the function $(g_z(z;w)+\theta_z(z;w))^2$ at $z=0$ and $z=\infty$:
\begin{align}\label{g}
\begin{aligned}
&(g_z(z;w)+\theta_z(z;w))^2=4z^{-4},\quad z\rightarrow0,\vspace{0.05in}\\
&(g_z(z;w)+\theta_z(z;w))^2=144z^4+24wz^2+\mathcal{O}(1),\quad z\rightarrow\infty.
\end{aligned}
\end{align}

According to Liouville's theorem, we can rewrite Eq.~(\ref{g}) as
\begin{align}
\begin{aligned}
&(g_z(z;w)+\theta_z(z;w))^2=z^{-4}Q(z;w),\vspace{0.05in}\\
&Q(z;w)=144z^8+24wz^6+d_4(w)z^4+d_3(w)z^3+d_2(w)z^2+4.
\end{aligned}
\end{align}

Here we will be interested in the main case in which $Q(z;w)$ has two double roots and four simple roots:
\bee\label{Q-root}
Q(z;w)=144(z^2+a_1z+a_0)^2(z^2+b_1z+b_0)(z^2+c_1z+c_0).
\ene

Expanding out the right-hand side of Eq.~(\ref{Q-root}) and comparing with the coefficients of $z^7,z^6,z^5,z^1$ and $z^0$, we have
\begin{align}\label{a1b1c1}
\begin{aligned}
&288a_1+144b_1+144c_1=0,\vspace{0.05in}\\
&144a_1^2+288a_1b_1+288a_0+144b_0+144(2a_1+b_1)c_1+144c_0=24w,\vspace{0.05in}\\
&288a_0a_1+144(a_1^2+2a_0)b_1+288a_1b_0+144(a_1^2+2a_1b_1+2a_0+b_0)c_1+144(2a_1+b_1)c_0=0,\vspace{0.05in}\\
&144a_0^2b_0c_1+144(a_0^2b_1+2a_0a_1b_0)c_0=0,\vspace{0.05in}\\
&144a_0^2b_0c_0=4.
\end{aligned}
\end{align}

To solve the above equation (\ref{a1b1c1}), we choose a constraint
\bee\no
a_1=0,\quad b_1=0,\quad c_1=0.
\ene

For convenience, we  add constraints on simple zeros of $Q(z;w)$, and the function $Q(z;w)$ can be rewritten as
\bee\label{Q-root1}
Q(z;w)=144(z^2-d_0^2)^2(z^2-(d_1+id_2)^2)(z^2-(d_1-id_2)^2).
\ene

Expanding out the right-hand side of Eq.~(\ref{Q-root1}) and comparing with the coefficients of $z^6$ and $z^0$, we have
\begin{align}\label{d0d1d2}
\begin{aligned}
&-288d_0^2-288d_1^2+288d_2^2=24w,\vspace{0.05in}\\
&144d_0^4d_1^4+288d_0^4d_1^2d_2^2+144d_0^4d_2^4=4.
\end{aligned}
\end{align}

Note that the following two constraints must be met:
\begin{align}\label{d0d1d2-1}
\begin{aligned}
&144d_0^4+(576d_1^2-576d_2^2)d_0^2+144(d_1^2+d_2^2)^2=w^2+24(\lim_{z=\infty}z^2g_z(z;w)-2),\vspace{0.05in}\\
&-288((d_1^2-d_2^2)d_0^2+(d_1^2+d_2^2)^2)d_0^2=2w(\lim_{z=\infty}z^2g_z(z;w)-2).
\end{aligned}
\end{align}

By solving Eqs.~(\ref{d0d1d2}) and (\ref{d0d1d2-1}), we obtain
\bee\no
d_0=\frac{1}{12}\sqrt{6\sqrt{w^2+96}-6w},\quad d_1=0,\quad d_2=\frac{1}{2d_0}\sqrt{\frac{12d_0^4+d_0^2w+2}{6}}.
\ene

Specifically, when $w=0$, we have
\bee\no
d_0=6^{-\frac14},\quad d_1=0,\quad d_2=6^{-\frac14}.
\ene

Then, the function $Q(z;w)$ has the following four double roots:
$d_0(w),~-d_0(w),~id_2(w),~-id_2(w).$

\begin{figure}[!t]
    \centering
 \vspace{-0.15in}
  {\scalebox{0.5}[0.5]{\includegraphics{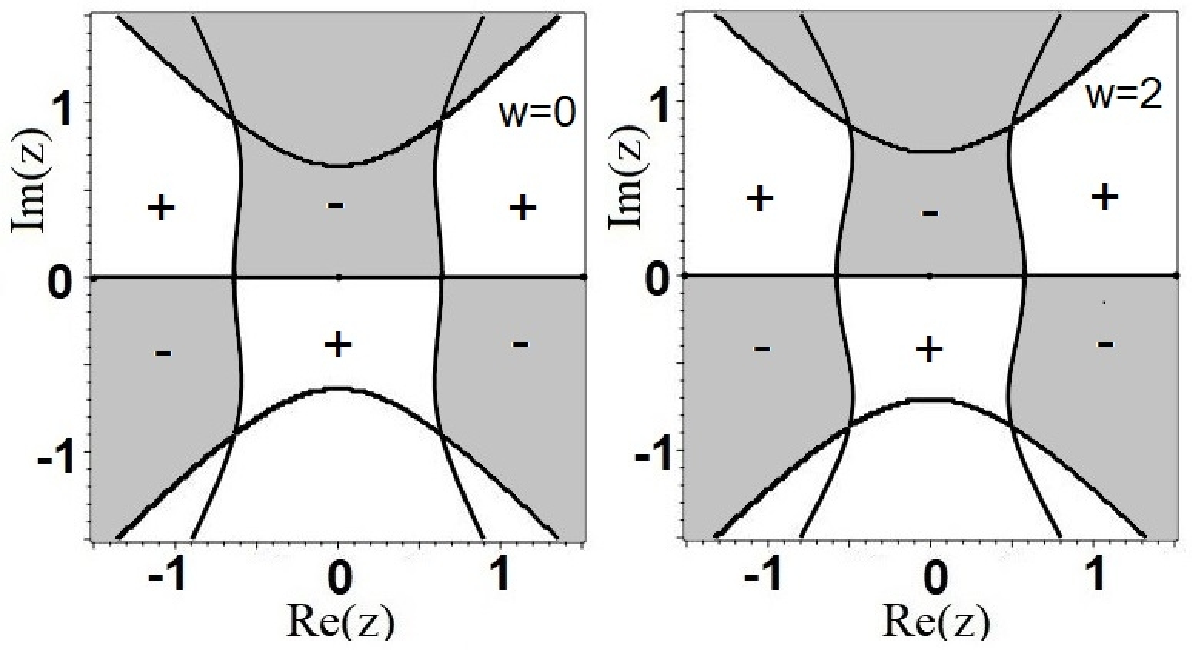}}}  
\vspace{-0.05in}
\caption{Sign charts of $h(z;w)$ when $w=0,2$.}
   \label{w=2}
\end{figure}

\subsubsection{Exponent analysis and steepest descent}

Let $h(z;w)=g(z;w)+\theta(z;w)$. For convenience, we denote $b(w):=|d_0(w)|,d(w):=id_2(w)$. There exists a component of the curve $\mathrm{Im}(h(z;a))=0$ that is a Jordan curve surrounding the origin in the $z$-plane and that passes through two real simple critical points, namely, $-b(w)$ and $b(w)$. We select this curve as the jump contour $\Sigma_2$ for $Y(z;T,w)$. The real axis divides $\Sigma_2$ into an arc $\Sigma_2^+$ in the upper half-plane and an arc $\Sigma_2^-$ in the lower half-plane. We introduce thin lens-shaped domains $L^{\pm}$ and $R^{\pm}$ on the left-hand and right-hand sides, respectively, of $\Sigma_2^{\pm}$ whose outer boundary arcs $\Sigma_3^{\pm}$ and $\Sigma_4^{\pm}$ meet the real axis at $\frac{\pi}{4}$ angles as shown in the left-hand panel of Figure \ref{fig2}. The region between $\Sigma_4^{\pm}$ and the real axis is denoted $\Omega^{\pm}$. Let the line segment $I:=[-b(w),b(w)]$. Then, we will make the following transformations:
\begin{eqnarray}
&&W(z;T,w):=Y(z;T,w)\left[\!\!\begin{array}{cc}
1& 0  \vspace{0.05in}\\
e^{2iT^{1/4}\theta(z;w)} & 1
\end{array}\!\!\right]e^{iT^{1/4}g(z;w)\sigma_3},\quad z\in L^+_{l,r},\vspace{0.05in} \no\\
&&W(z;T,w):=Y(z;T,w)\left[\!\!\begin{array}{cc}
\sqrt{2}& \frac{\sqrt{2}}{2}e^{-2iT^{1/4}\theta(z;w)}  \vspace{0.05in}\\
0 & \frac{\sqrt{2}}{2}
\end{array}\!\!\right]e^{iT^{1/4}g(z;w)\sigma_3},\quad z\in R^+_{l,r},\vspace{0.05in} \no\\
&&W(z;T,w):=Y(z;T,w)2^{\frac{\sigma_3}{2}}e^{iT^{1/4}g(z;w)\sigma_3},\quad z\in \Omega^+,\vspace{0.05in} \no\\
&&W(z;T,w):=Y(z;T,w)2^{-\frac{\sigma_3}{2}}e^{iT^{1/4}g(z;w)\sigma_3},\quad z\in \Omega^-,\vspace{0.05in}\no\\
&&W(z;T,w):=Y(z;T,w)\left[\!\!\begin{array}{cc}
\frac{\sqrt{2}}{2}& 0  \vspace{0.05in}\\
-\frac{\sqrt{2}}{2}e^{2iT^{1/4}\theta(z;w)} & \sqrt{2}
\end{array}\!\!\right]e^{iT^{1/4}g(z;w)\sigma_3},\quad z\in R^-_{l,r},\vspace{0.05in}\no\\
&&W(z;T,w):=Y(z;T,w)\left[\!\!\begin{array}{cc}
1& -e^{-2iT^{1/4}\theta(z;w)}  \vspace{0.05in}\\
0 & 1
\end{array}\!\!\right]e^{iT^{1/4}g(z;w)\sigma_3},\quad z\in L^-_{l,r},\vspace{0.05in} \\
&&W(z;T,w):=Y(z;T,w)2^{\frac{\sigma_3}{2}}\left[\!\!\begin{array}{cc}
1& -\frac12e^{-2iT^{1/4}\theta(z;w)}  \vspace{0.05in}\\
0 & 1
\end{array}\!\!\right]e^{iT^{1/4}g(z;w)\sigma_3},\quad z\in R_{1,l}^+\cup R_{1,r}^+,\vspace{0.05in}\no\\
&&W(z;T,w):=Y(z;T,w)\left[\!\!\begin{array}{cc}
1& e^{-2iT^{1/4}\theta(z;w)}  \vspace{0.05in}\\
0 & 1
\end{array}\!\!\right]e^{iT^{1/4}g(z;w)\sigma_3},\quad z\in L_{1,l}^+\cup L_{1,r}^+,\vspace{0.05in}\no\\
&&W(z;T,w):=Y(z;T,w)2^{-\frac{\sigma_3}{2}}\left[\!\!\begin{array}{cc}
1& 0  \vspace{0.05in}\\
\frac12e^{2iT^{1/4}\theta(z;w)} & 1
\end{array}\!\!\right]e^{iT^{1/4}g(z;w)\sigma_3},\quad z\in R_{1,l}^-\cup R_{1,r}^-,\vspace{0.05in}\no\\
&&W(z;T,w):=Y(z;T,w)\left[\!\!\begin{array}{cc}
1& 0  \vspace{0.05in}\\
-e^{2iT^{1/4}\theta(z;w)} & 1
\end{array}\!\!\right]e^{iT^{1/4}g(z;w)\sigma_3},\quad z\in L_{1,l}^-\cup L_{1,r}^-,\vspace{0.05in}\no\\
&&W(z;T,w):=Y(z;T,w)e^{iT^{1/4}g(z;w)\sigma_3},\quad \mathrm{otherwise}. \no
\end{eqnarray}

\begin{figure}[!t]
    \centering
 \vspace{-0.15in}
  {\scalebox{0.35}[0.35]{\includegraphics{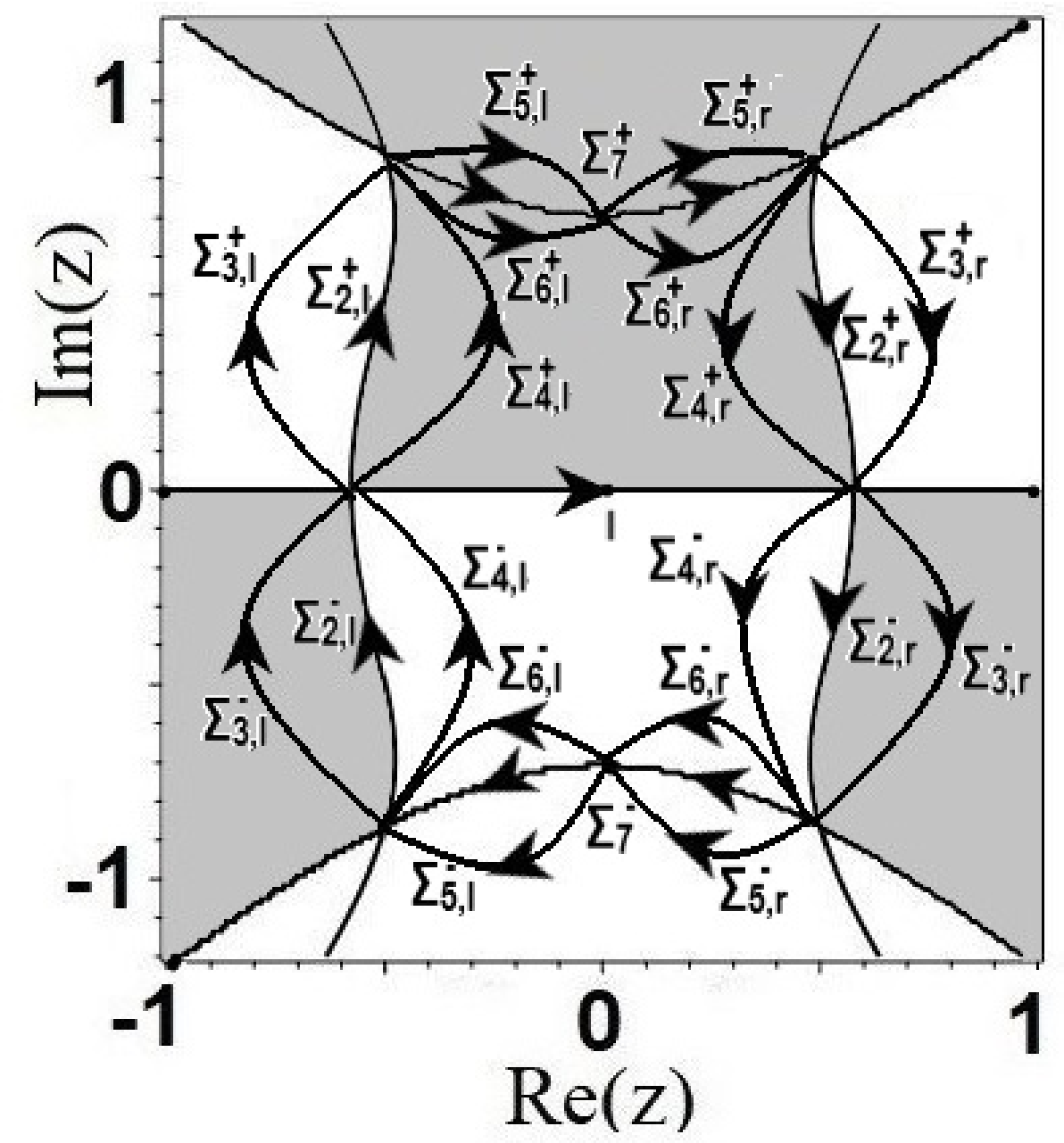}}}
\vspace{-0.05in}
\caption{Regional division for large $T$.}
   \label{fig2}
\end{figure}

We can easily find that $W(z;T,w)$ has no jump on the arc $\Sigma_2$, namely, $W_+(z;T,w)=W_-(z;T,w),\quad z\in\Sigma_2$. Below we give the jump properties of $W(z;T,w)$ on the arcs $\Sigma_j,\, (j=3,4,5,6,7)$ and line segment $I$.
\begin{eqnarray} 
&&W_+(z;T,w)=W_-(z;T,w)\left[\!\!\begin{array}{cc}
1& 0  \vspace{0.05in} \\
-e^{2iT^{\frac14}h(z;w)} & 1
\end{array}\!\!\right],\quad z\in \Sigma_{3,l,r}^+,\vspace{0.05in} \no\\
&&W_+(z;T,w)=W_-(z;T,w)\left[\!\!\begin{array}{cc}
1& \frac12e^{-2iT^{\frac14}h(z;w)} \vspace{0.05in}\\
0 & 1
\end{array}\!\!\right],\quad z\in \Sigma_{4,l,r}^+,\vspace{0.05in} \no\\
&&W_+(z;T,w)=W_-(z;T,w)2^{\sigma_3},\qquad z\in I,\vspace{0.05in} \no\\
&&W_+(z;T,w)=W_-(z;T,w)\left[\!\!\begin{array}{cc}
1& 0 \vspace{0.05in}\\
-\frac12e^{2iT^{\frac14}h(z;w)} & 1
\end{array}\!\!\right],\quad z\in \Sigma_{4,l,r}^-,\vspace{0.05in} \no\\
&&W_+(z;T,w)=W_-(z;T,w)\left[\!\!\begin{array}{cc}
1& e^{-2iT^{\frac14}h(z;w)}  \vspace{0.05in}\\
0 & 1
\end{array}\!\!\right],\quad z\in \Sigma_{3,l,r}^-,\vspace{0.05in}
\label{jump-W-1} \\
&&W_+(z;T,w)=W_-(z;T,w)\left[\!\!\begin{array}{cc}
1& -e^{-2iT^{\frac14}h(z;w)}  \vspace{0.05in}\\
0 & 1
\end{array}\!\!\right],\quad z\in \Sigma_{5,l,r}^+,\vspace{0.05in} \no\\
&&W_+(z;T,w)=W_-(z;T,w)\left[\!\!\begin{array}{cc}
1& -\frac12e^{-2iT^{\frac14}h(z;w)}  \vspace{0.05in}\\
0 & 1
\end{array}\!\!\right],\quad z\in \Sigma_{6,l,r}^+,\vspace{0.05in}\no\\
&&W_+(z;T,w)=W_-(z;T,w)\left[\!\!\begin{array}{cc}
1& 0  \vspace{0.05in}\\
\frac12e^{2iT^{\frac14}h(z;w)} & 1
\end{array}\!\!\right],\quad z\in \Sigma_{6,l,r}^-,\vspace{0.05in} \no\\
&&W_+(z;T,w)=W_-(z;T,w)\left[\!\!\begin{array}{cc}
1& 0  \vspace{0.05in}\\
e^{2iT^{\frac14}h(z;w)} & 1
\end{array}\!\!\right],\quad z\in \Sigma_{5,l,r}^-,\vspace{0.05in} \no\\
&&W_+(z;T,w)=W_-(z;T,w)\left[\!\!\begin{array}{cc}
0& e^{-iT^{\frac14}c(w)}  \vspace{0.05in}\\
-e^{iT^{\frac14}c(w)} & 0
\end{array}\!\!\right],\quad z\in\Sigma_7^+\cup\Sigma_7^-, \no
\end{eqnarray}
where $c(w)=g_+(z;w)+g_-(z;w)+2\theta(z;w),~z\in\Sigma_7^+\cup\Sigma_7^-$. Since $\mathrm{Im}(h(z;w))>0$ holds on $\Sigma_{3,r}^{+},\Sigma_{3,l}^{+},\Sigma_{4,l}^-,\Sigma_{4,r}^-,\Sigma_{5,l,r}^-,\Sigma_{6,l,r}^-$ while $\mathrm{Im}(h(z;w))<0$ holds on $\Sigma_{3,l}^-,\Sigma_{3,r}^-,\Sigma_{4,l}^{+},\Sigma_{4,r}^{+},\Sigma_{5,l,r}^+,\Sigma_{6,l,r}^+$,  the jump matrices are exponentially decreasing on these four contour arcs except near the endpoints $b(w),-b(w),d(w)$ and $-d(w)$. 
\textcolor{red}{At these four endpoints, the new models are needed, which cannot be provided in this paper now. This will be further studied in our future work.}

\subsection{Transitional asymptotic behavior: $a\rightarrow- 1/96$}

Let $\Theta(z;a):=z+4az^3+2z^{-1}$. When $a\rightarrow-\frac{1}{96}:=a_c$, there is a pair of critical points (real for $-\frac{1}{96}<a<0$ and complex-conjugate for $a<-\frac{1}{96}$) near the double critical point $b_c:=2$ and a pair of critical points (real for $-\frac{1}{96}<a<0$ and complex-conjugate for $a<-\frac{1}{96}$) near the double critical point $-b_c$. Sign chart of $\Theta(z;a)$ when $a=-\frac{1}{96}$ is shown in Figure \ref{double}. Note that $\Theta(b_c;a_c):=\frac83$ and $\Theta(-b_c;a_c):=-\frac83$. The Taylor expansion of $\Theta(z;a)$ about $z=-b_c,a=a_c$ as:
\begin{align}\label{bcac}
\begin{aligned}
\Theta(z;a):=&-\frac83-32(a-a_c)+48(a-a_c)(z+b_c)-24(a-a_c)(z+b_c)^2-\frac16(z+b_c)^3\vspace{0.05in}\\
&+4(a-a_c)(z+b_c)^3-\frac{1}{16}(z+b_c)^4-\frac{1}{32}(z+b_c)^5+\mathcal{O}((z+b_c)^6),\quad z\rightarrow-b_c,
\end{aligned}
\end{align}
and when $a=a_c$, we have
\begin{align}\label{bcac2}
\begin{aligned}
\Theta(z;a_c):=&-\frac83-\frac16(z+b_c)^3-\frac{1}{16}(z+b_c)^4-\frac{1}{32}(z+b_c)^5+\mathcal{O}((z+b_c)^6),\quad z\rightarrow-b_c.
\end{aligned}
\end{align}

\begin{figure}[!t]
    \centering
 \vspace{-0.15in}
  {\scalebox{0.27}[0.27]{\includegraphics{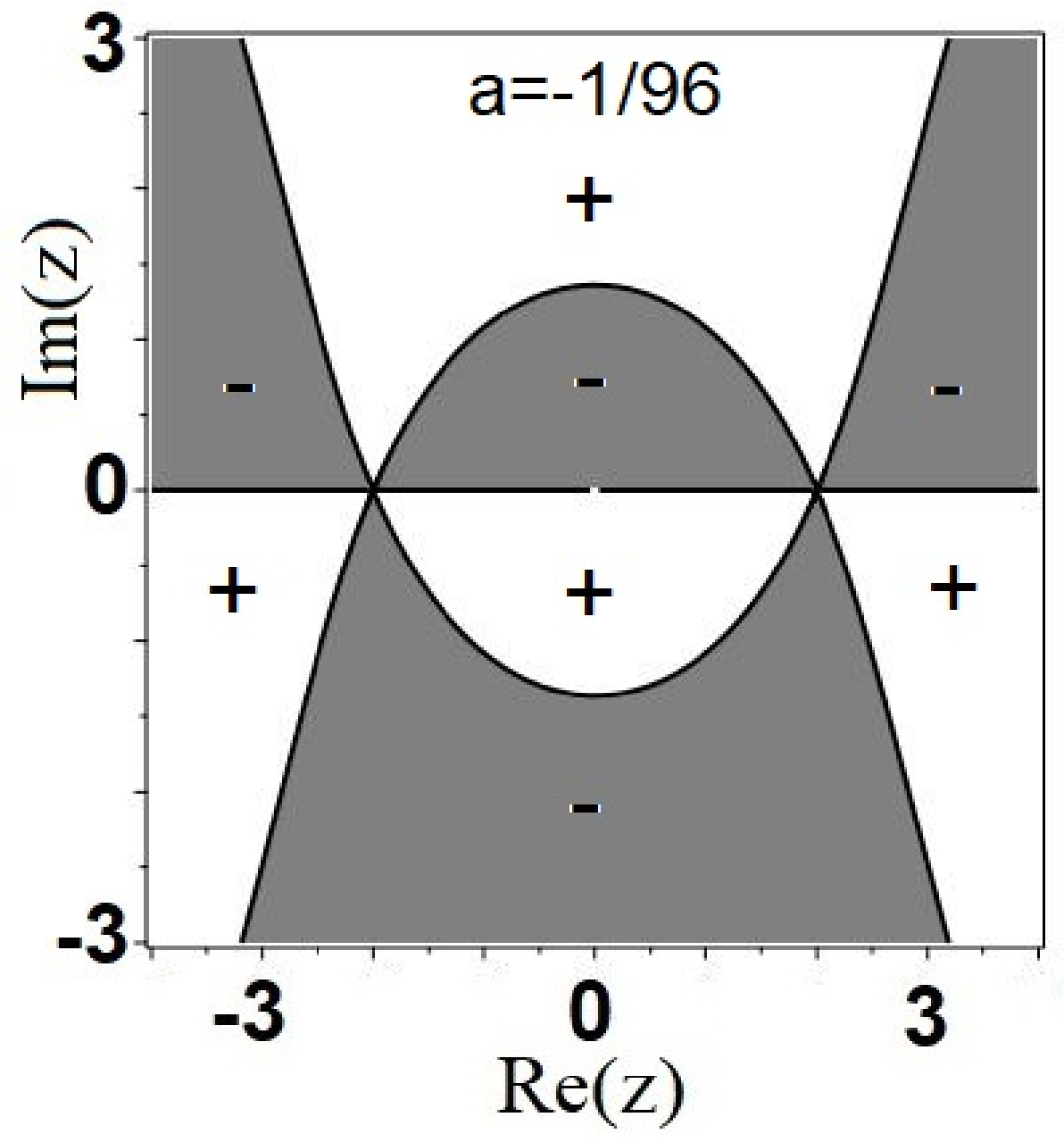}}}  
\vspace{-0.03in}
\caption{Sign charts of $\Theta(z;a)$ when $a=-\frac{1}{96}$.}
   \label{double}
\end{figure}

According to~\cite{Chester1957}, we can define a Schwarz-symmetric conformal mapping $z\mapsto M_1(z;a)$ in the neighborhood of $z=-b_c,a=a_c$ by the equation:
\bee\no
2\Theta(z;a)=M_1(z;a)^3+r_1(a)M_1(z;a)-s_1(a),
\ene
where $r_1(a)$ and $s_1(a)$ are real analytic functions of $a$ near $a_c$. And they satisfy the following properties:
\bee\label{r-s-1}
r_{1}(a_c)=0,\quad s_1(a_c)=\frac{16}{3},\quad M_1^{'}(-b_c;a_c)=-3^{-\frac13}<0,\quad s_1^{'}(a_c)=64,\quad r_1^{'}(a_c)=-96\cdot3^{\frac13}.
\ene

Moreover, we denote by $z_1(a)$ the preimage of $M_1(z;a)=0$. It is an analytic function of $a$ that satisfies $z_1(a_c)=-b_c$.

On the other hand, the Taylor expansion of $\Theta(z;a)$ about $z=b_c,a=a_c$ as:
\begin{align}\label{bcac3}
\begin{aligned}
\Theta(z;a):=&\frac83+32(a-a_c)+48(a-a_c)(z-b_c)+24(a-a_c)(z-b_c)^2-\frac16(z-b_c)^3\vspace{0.05in}\\
&+4(a-a_c)(z-b_c)^3+\frac{1}{16}(z-b_c)^4-\frac{1}{32}(z-b_c)^5+\mathcal{O}((z-b_c)^6),\quad z\rightarrow b_c,
\end{aligned}
\end{align}
and when $a=a_c$, we have
\begin{align}\label{bcac4}
\begin{aligned}
\Theta(z;a_c):=&\frac83-\frac16(z-b_c)^3+\frac{1}{16}(z-b_c)^4-\frac{1}{32}(z-b_c)^5+\mathcal{O}((z-b_c)^6),\quad z\rightarrow b_c.
\end{aligned}
\end{align}

Similarly, we can define a Schwarz-symmetric conformal mapping $z\mapsto M_2(z;a)$ in the neighborhood of $z=b_c,a=a_c$ by the equation:
\bee\no
-2\Theta(z;a)=M_2(z;a)^3+r_2(a)M_2(z;a)-s_2(a),
\ene
where $r_2(a)$ and $s_2(a)$ are real analytic functions of $a$ near $a_c$. And they satisfy the following properties:
\bee\label{r-s-2}
r_{2}(a_c)=0,\quad s_2(a_c)=\frac{16}{3},\quad M_2'(b_c;a_c)=3^{-\frac13}>0,\quad s_2'(a_c)=64,\quad r_2'(a_c)=-96\cdot3^{\frac13}.
\ene
Moreover, we denote by $z_2(a)$ the preimage of $M_2(z;a)=0$. It is an analytic function of $a$ that satisfies $z_2(a_c)=b_c$

\subsubsection{Parametrix modification}

The matrix function $W(z;X,a)$ and its jumping properties are given by Eqs.~(\ref{jump-W}) and (\ref{WY-1}). We define the following function:
\bee\no
\overline{W}^o(z;a)=\left(\frac{z-z_1(a)}{z-z_2(a)}\right)^{ip\sigma_3},\quad z\in\mathbb{C}\setminus [z_1(a),z_2(a)],
\ene
where $p=\frac{\ln(2)}{2\pi}$. Then, we get the jump condition of the function $\overline{W}^o(z;a)$ on the line segment $[z_1(a),z_2(a)]$:
\bee\no
\overline{W}_+^o(z;a)=\overline{W}_-^o(z;a)2^{\sigma_3},\quad z\in [z_1(a),z_2(a)].
\ene

Let $\xi_{-b}:=X^{\frac16}M_1,~y_1:=X^{\frac13}r_1$ and $\xi_{b}:=X^{\frac16}M_2,~y_2:=X^{\frac13}r_2$. Then the jump conditions satisfied by
\bee\no
U^{-b}:=iW(z;X,a)e^{\frac{i}{2}X^{\frac12}s_1(a)\sigma_3}\sigma_2,\quad~\mathrm{near}~z=-b_c
\ene
and by
\bee\no
U^{b}:=W(z;X,a)e^{-\frac{i}{2}X^{\frac12}s_2(a)\sigma_3}\quad~\mathrm{near}~z=b_c.
\ene

The outer parametrix can also be expressed near $z=-b_c$ or $z=b_c$ in terms of the relevant conformal coordinate:
\begin{align}
\begin{aligned}
&i\overline{W}^o(z;a)e^{\frac{i}{2}X^{\frac12}s_1(a)\sigma_3}\sigma_2=X^{-\frac{i}{6}p\sigma_3}e^{\frac{i}{2}X^{\frac12}s_1(a)\sigma_3}T^{-b}(z;a)
\xi_{-b}^{-ip\sigma_3},\vspace{0.05in}\\
&\overline{W}^o(z;a)e^{-\frac{i}{2}X^{\frac12}s_2(a)\sigma_3}=X^{\frac{i}{6}p\sigma_3}e^{-\frac{i}{2}X^{\frac12}s_2(a)\sigma_3}T^{b}(z;a)
\xi_{b}^{-ip\sigma_3},
\end{aligned}
\end{align}
where
\begin{align}
\begin{aligned}
&T^{-b}(z;a):=i(z_2(a)-z)^{-ip\sigma_3}\left(\frac{z_1(a)-z}{M_1(z;a)}\right)^{ip\sigma_3}\sigma_2,\vspace{0.05in}\\
&T^{b}(z;a):=(z-z_1(a))^{ip\sigma_3}\left(\frac{M_2(z;a)}{z-z_2(a)}\right)^{ip\sigma_3}.
\end{aligned}
\end{align}

Below we give the jump properties of $U^{-b}$ and $U^{b}$:
\begin{align}
\begin{aligned}
&U^{b}_+=U^{b}_-\left[\!\!\begin{array}{cc}
1& 0  \vspace{0.05in}\\
e^{-i(\xi_{b}^3+y_2\xi_b)} & 1
\end{array}\!\!\right],\quad z\in \Sigma_3^+\,(\mathrm{away~from}~z_2(a)),\vspace{0.08in}\\
&U^{b}_+=U^{b}_-\left[\!\!\begin{array}{cc}
1& \frac12e^{i(\xi_{b}^3+y_2\xi_b)} \vspace{0.08in}\\
0 & 1
\end{array}\!\!\right],\quad z\in \Sigma_4^+\,(\mathrm{toward}~z_2(a)),\vspace{0.08in}\\
&U^{b}_+=U^{b}_-2^{\sigma_3},\qquad\qquad\qquad\qquad  z\in I\,(\mathrm{toward}~z_2(a)),\vspace{0.08in}\\
&U^{b}_+=U^{b}_-\left[\!\!\begin{array}{cc}
1& 0 \vspace{0.08in}\\
\frac12e^{-i(\xi_{b}^3+y_2\xi_b)} & 1
\end{array}\!\!\right],\quad z\in \Sigma_4^-\,(\mathrm{toward}~z_2(a)),\vspace{0.08in}\\
&U^{b}_+=U^{b}_-\left[\!\!\begin{array}{cc}
1& e^{i(\xi_{b}^3+y_2\xi_b)}  \vspace{0.05in}\\
0 & 1
\end{array}\!\!\right],\qquad z\in \Sigma_3^-\,(\mathrm{away~from}~z_2(a)).
\end{aligned}
\end{align}
and
\begin{align}
\begin{aligned}
&U^{-b}_+=U^{-b}_-\left[\!\!\begin{array}{cc}
1& e^{i(\xi_{-b}^3+y_1\xi_{-b})}  \vspace{0.05in}\\
0 & 1
\end{array}\!\!\right],\quad z\in \Sigma_3^+,(\mathrm{away~from}z_1(a))\vspace{0.05in}\\
&U^{-b}_+=U^{-b}_-\left[\!\!\begin{array}{cc}
1& 0 \vspace{0.05in}\\
\frac12e^{-i(\xi_{-b}^3+y_1\xi_{-b})} & 1
\end{array}\!\!\right],\quad z\in \Sigma_4^+\,(\mathrm{toward}~z_1(a)),\vspace{0.05in}\\
&U^{b}_+=U^{b}_-2^{\sigma_3},\qquad\qquad\qquad\qquad\qquad  z\in I\,(\mathrm{toward}~z_1(a)),\vspace{0.05in}\\
&U^{-b}_+=U^{-b}_-\left[\!\!\begin{array}{cc}
1& 0 \vspace{0.05in}\\
e^{-i(\xi_{-b}^3+y_1\xi_{-b})}& 1
\end{array}\!\!\right],\quad z\in \Sigma_3^-\,(\mathrm{away~from}~z_1(a)),\vspace{0.05in}\\
&U^{-b}_+=U^{-b}_-\left[\!\!\begin{array}{cc}
1& \frac12e^{i(\xi_{-b}^3+y_1\xi_{-b})}  \vspace{0.05in}\\
0 & 1
\end{array}\!\!\right],\quad z\in \Sigma_4^-\,(\mathrm{toward}~z_1(a)).
\end{aligned}
\end{align}

We can normalize the jump matrix near point $z_1(a)$ and the jump matrix near point $z_2(a)$ into the jump matrix about $\xi$.
\begin{figure}[!t]
    \centering
 \vspace{-0.15in}
  {\scalebox{0.35}[0.35]{\includegraphics{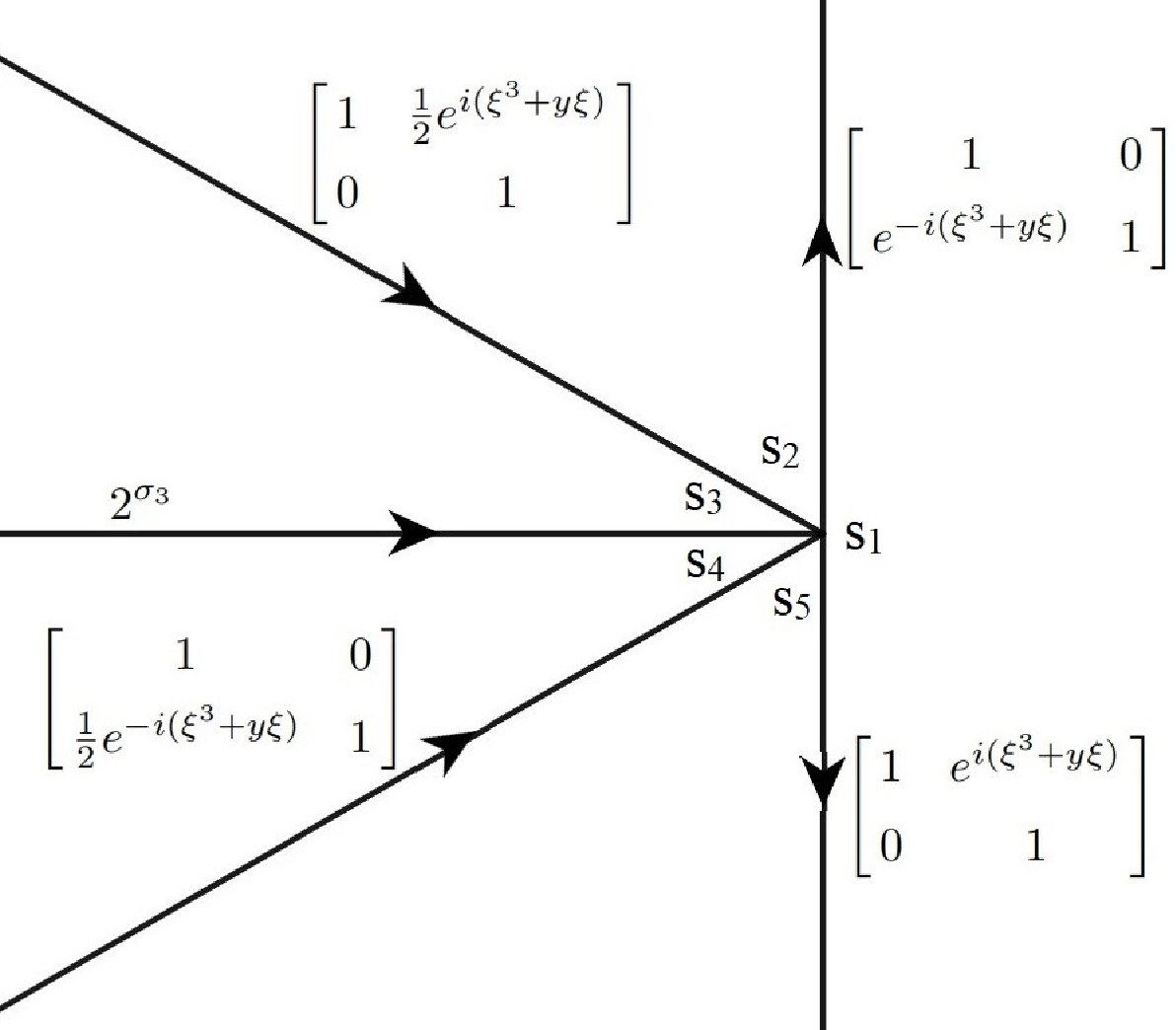}}}  
\vspace{0.05in}
\caption{Regional division for $U(\xi)$.}
   \label{guodu-1}
\end{figure}

\begin{prop}\label{RH8} Riemann-Hilbert Problem 8

Find a $2\times 2$ matrix $U(\xi;y)$ that satisfies the following properties:

\begin{itemize}

\item{} Analyticity: $U(\xi;y)$ is analytic for $\xi$ in the five regions shown in Figure \ref{guodu-1}, namely, $S_1: |\mathrm{arg}(\xi)|<\frac{\pi}{2}$, $S_2: \frac{\pi}{2}<\mathrm{arg}(\xi)<\frac{5\pi}{6}$, $S_3: \frac{5\pi}{6}<\mathrm{arg}(\xi)<\pi$, $S_4:-\pi<\mathrm{arg}(\xi)<-\frac{5\pi}{6}$, $S_5:-\frac{5\pi}{6}<\mathrm{arg}(\xi)<-\frac{\pi}{2}$ where $-\pi<\mathrm{arg}(\xi)\leq\pi$. It takes continuous boundary values on the excluded rays and at the origin from each sector.

\item{} Jump conditions: The boundary values on the jump contour are related as:
\begin{align}
\begin{aligned}
&U_+(\xi;y)=U_-(\xi;y)\left[\!\!\begin{array}{cc}
1& 0  \vspace{0.05in}\\
e^{-i(\xi^3+y\xi)} & 1
\end{array}\!\!\right],\quad \mathrm{arg}(\xi)=\frac{\pi}{2},\vspace{0.05in}\\
&U_+(\xi;y)=U_-(\xi;y)\left[\!\!\begin{array}{cc}
1& \frac12e^{i(\xi^3+y\xi)} \vspace{0.05in}\\
0 & 1
\end{array}\!\!\right],\quad \mathrm{arg}(\xi)=\frac{5\pi}{6},\vspace{0.05in}\\
&U_+(\xi;y)=U_-(\xi;y)2^{\sigma_3},\qquad\qquad \mathrm{arg}(\xi)=\pi,\vspace{0.05in}\\
&U_+(\xi;y)=U_-(\xi;y)\left[\!\!\begin{array}{cc}
1& 0 \vspace{0.05in}\\
\frac12e^{-i(\xi^3+y\xi)} & 1
\end{array}\!\!\right],\quad \mathrm{arg}(\xi)=-\frac{5\pi}{6},\vspace{0.05in}\\
&U_+(\xi;y)=U_-(\xi;y)\left[\!\!\begin{array}{cc}
1& e^{i(\xi^3+y\xi)}  \vspace{0.05in}\\
0 & 1
\end{array}\!\!\right],\quad \mathrm{arg}(\xi)=-\frac{\pi}{2}.
\end{aligned}
\end{align}

\item{} Normalization: $U(\xi;y)\xi^{ip\sigma_3}$ tends to the identity matrix as $\xi\rightarrow\infty$, where $p=\frac{\ln(2)}{2\pi}$.

\end{itemize}
\end{prop}

The solution of the Riemann-Hilbert Problem 8 can be expressed explicitly according to the Painlev\'{e}-II differential equation~\cite{miller-sigma}. In particular, the solution satisfies:
\bee
U(\xi;y)\xi^{ip\sigma_3}\sim\mathbb{I}+\sum\limits_{j=1}^{\infty}U^{[j]}(y)\xi^{-j},\quad \xi\rightarrow\infty,
\ene
uniformly in all directions of the complex $\xi$-plane. Let
\bee\no
\mathcal{V}_1(y):=U_{21}^{[1]}(y)=\lim\limits_{\xi\rightarrow\infty}\xi U_{21}(\xi;y)\xi^{ip},
\ene
and $\mathcal{Q}(y):=\frac{\mathcal{V}_1^{'}(y)}{\mathcal{V}_1(y)}$. Then the function $\mathcal{Q}(y)$ satisfies the Painlev\'{e}-II differential equation:
\bee\label{27-painleve-1}
\frac{d^2\mathcal{Q}}{dy^2}+\frac{2}{3}y\mathcal{Q}-2\mathcal{Q}^3-\frac23ip-\frac13=0,
\ene
with the following asymptotic behavior
\bee\no
\mathcal{Q}(y)=i(-\frac{y}{3})^{\frac12}-(\frac14+\frac{ip}{2})\frac{1}{y}+\mathcal{O}(|y|^{-\frac52}),\quad |\mathrm{arg}(-y)|<\frac{2\pi}{3}.
\ene

The alternate formula for the function $\mathcal{V}_1(y)$ is given as follows:
\bee
\mathcal{V}_1(y)
=\begin{cases}
\d -\frac{i\gamma^*}{2}e^{-\frac{2\sqrt{3}i}{9}(-y)^{\frac32}}(-3y)^{-(\frac14+\frac{ip}{2})}\exp\l(\int_{-\infty}^y\mathcal{Q}(k)
-i(-\frac{k}{3})^\frac12+\frac{1+2ip}{4k}dk\r),\quad y<0,\vspace{0.05in} \\
\mathcal{V}_1(-1)\exp\l(\d\int_{-1}^y\mathcal{Q}(k)dk\r),\quad y\geq0.
\end{cases}
\ene

Then, $\mathcal{V}_1(y)$ has the asymptotic behavior:
\bee
\mathcal{V}_1(y)=-\sqrt{\frac{y}{6}}\left(\frac{y}{6}\right)^{ip}+\mathcal{O}(y^{-\frac14}),\quad y\rightarrow\infty.
\ene

And the function $\mathcal{V}_2(y)$ defined by the formula:
\bee\no
\mathcal{V}_2(y):=U_{12}^{[1]}(y)=\lim\limits_{\xi\rightarrow\infty}\xi U_{12}(\xi;y)\xi^{-ip},
\ene
and $\mathcal{P}(y):=\frac{\mathcal{V}_2^{'}(y)}{\mathcal{V}_2(y)}$. Then the function $\mathcal{P}(y)$ satisfies the Painlev\'{e}-II differential equation:
\bee\label{27-painleve-2}
\frac{d^2\mathcal{P}}{dy^2}+\frac{2}{3}y\mathcal{P}-2\mathcal{P}^3+\frac23ip-\frac13=0,
\ene
with the following asymptotic behavior
\bee\no
\mathcal{P}(y)=i(-\frac{y}{3})^{\frac12}-(\frac14-\frac{ip}{2})\frac{1}{y}+\mathcal{O}(|y|^{-\frac52}),\quad |\mathrm{arg}(-y)|<\frac{2\pi}{3}.
\ene
Finally, $\mathcal{V}_2(y)$ has the asymptotic behavior:
\bee
\mathcal{V}_2(y)=\sqrt{\frac{y}{6}}\left(\frac{y}{6}\right)^{-ip}+\mathcal{O}(y^{-\frac14}),\quad y\rightarrow\infty.
\ene
Then, we can get $\mathcal{V}_2(y)=-\mathcal{V}_1(y)^*$. According to Proposition \ref{RH8}, we can obtain the following local functions:
\begin{align}
\begin{aligned}
\overline{W}^{-b}(z;X,a)=&-iX^{-\frac{ip\sigma_3}{6}}e^{\frac{i}{2}X^{\frac12}s_1(a)\sigma_3}T^{-b}(z;a)\vspace{0.05in}\\
&\times U(X^{\frac16}M_1(z;a);X^\frac{1}{3}r_1(a))\sigma_2e^{-\frac{i}{2}X^{\frac12}s_1(a)\sigma_3},\quad z\in\delta_{-b_c}(r),\vspace{0.05in}\\
\overline{W}^{b}(z;X,a)=&X^{\frac{ip\sigma_3}{6}}e^{-\frac{i}{2}X^{\frac12}s_2(a)\sigma_3}T^{b}(z;a)\vspace{0.05in}\\
&\times U(X^{\frac16}M_2(z;a);X^\frac{1}{3}r_2(a))e^{\frac{i}{2}X^{\frac12}s_2(a)\sigma_3},\quad z\in\delta_{b_c}(r),
\end{aligned}
\end{align}
where $\delta_{-b}(r)$ and $\delta_{b}(r)$ are small circles with $-b(a)$ as the center and $r$ as the radius, and small circles with $b(a)$ as the center and $r$ as the radius, respectively. When $-\frac{1}{96}<a<0$, we define the following piecewise function:
\bee
\overline{W}(z;X,a):=\begin{cases}
\overline{W}^{-b}(z;X,a),\quad z\in\delta_{-b_c}(r),\vspace{0.05in}\\
\overline{W}^{b}(z;X,a),\quad z\in\delta_{b_c}(r),\vspace{0.05in}\\
\overline{W}^o(z;a),\quad z\in\mathbb{C}\setminus(\overline{\delta_{-b_c}(r)}\cup\overline{\delta_{b_c}(r)}\cup [z_1(a),z_2(a)]).
\end{cases}
\ene
Then, we will conduct the error analysis.

\subsubsection{Error analysis}

We consider the following error function:
\bee\label{EW-double}
E(z;X,a):=W(z;X,a)\overline{W}(z;X,a)^{-1}.
\ene

Since $W(z;X,a)$ and $\overline{W}(z;X,a)$ have the same jump property on the line segment $[z_1(a),z_2(a)]$, $E(z;X,a)$ has no jump across $[z_1(a),z_2(a)]$. The error function $E(z;X,a)$ is analytic in $\mathbb{C}\setminus\Sigma_5$ and takes continuous boundary values on $\Sigma_5$, where $\Sigma_5=\left(\Sigma_3^{\pm}\cup\Sigma_4^{\pm}\cup\partial\delta_{-b_c}(r)\cup\partial\delta_{b_c}(r)\right)\cap\complement_{\mathbb{C}}{(\delta_{-b_c}(r)
\cup\delta_{b_c}(r))}$. Assuming the clockwise orientation of $\delta_{-b_c}(r)$ and $\delta_{b_c}(r)$, the boundary values on the jump contour $\Sigma_5$ are related as:
\bee\no
E_+(z;X,a):=E_-(z;X,a)V^E(z;X,a).
\ene

Next, we will analyze the error factors of $V^E(z;X,a)$ on each arc of $\Sigma_5$.
\begin{prop}\label{error-double}
The following identity holds:\\
{\rm (I)}
\bee
\sup_{z\in\left(\Sigma_3^{\pm}\cup\Sigma_4^{\pm}\right)\cap\complement_{\mathbb{C}}{(\overline{\delta_{-b_c}(r)}\cup\overline{\delta_{b_c}(r)})}}
||V^E(z;X,a)-\mathbb{I}||=\mathcal{O}(e^{-C(a)X^{\frac12}}),\quad X\rightarrow+\infty,
\ene
where $C(a)$ is a constant about $a$ and $||\cdot||$ denotes the matrix norm.

{\rm (II)}
\bee
\sup_{z\in\partial\delta_{-b_c}(r)\cup\partial\delta_{b_c}(r)}||V^E(z;X,a)-\mathbb{I}||=\mathcal{O}(X^{-\frac16}),\quad X\rightarrow+\infty.
\ene

\end{prop}

\begin{proof}
When $z\in\left(\Sigma_3^{\pm}\cup\Sigma_4^{\pm}\right)\cap\complement_{\mathbb{C}}{(\overline{\delta_{-b_c}(r)}\cup\overline{\delta_{b_c}(r)})}$, the error term depends on the jump matrix of Eq.~(\ref{jump-W}), then we get formula (I). Note that
\begin{align}
\begin{aligned}
\overline{W}^{-b}(z;X,a)\overline{W}^{o}(z;a)^{-1}=&X^{-\frac{ip\sigma_3}{6}}e^{\frac{i}{2}X^{\frac12}s_1(a)\sigma_3}T^{-b}(z;a)U(\xi_{-b};y_1)\vspace{0.05in}\\
&\times\xi_{-b}^{ip\sigma_3}T^{-b}(z;a)^{-1}e^{-\frac{i}{2}X^{\frac12}s_1(a)\sigma_3}X^{\frac{ip\sigma_3}{6}},\quad z\in\partial\delta_{-b_c}(r)\vspace{0.05in}\\
\overline{W}^{b}(z;X,a)\overline{W}^{o}(z;a)^{-1}=&X^{\frac{ip\sigma_3}{6}}e^{-\frac{i}{2}X^{\frac12}s_2(a)\sigma_3}T^{b}(z;a)U(\xi_{b};y_2)\vspace{0.05in}\\
&\times\xi_{b}^{ip\sigma_3}T^{b}(z;a)^{-1}e^{\frac{i}{2}X^{\frac12}s_2(a)\sigma_3}X^{-\frac{ip\sigma_3}{6}},\quad z\in\partial\delta_{b_c}(r).
\end{aligned}
\end{align}

When $z\in\partial\delta_{-b}(r)\cup\partial\delta_{b}(r)$, the error term depends on $U(\xi_{-b};y_1)\xi_{-b}^{ip\sigma_3}$ and $U(\xi_{b};y_2)\xi_{b}^{ip\sigma_3}$. According to Eq.~(\ref{RH7-solu}), we get the formula (II). Thus the proof is completed.

\end{proof}

Using the Plemelj formula and Proposition \ref{error-double}, we have
\bee\label{E-err-double}
E_-(z;X,a)-\mathbb{I}=\mathcal{O}(X^{-\frac16}),\quad X\rightarrow+\infty.
\ene

Similarly to Eq.~(\ref{q-E-simi}), we have
\bee\label{q-E-simi-double}
\widehat{q}^+(X,aX^2)=-\frac{1}{\pi X^{\frac12}}\int_{\Sigma_5}E_{11-}(v;X,a)V^{E}_{12}(v;X,a)+E_{12-}(v;X,a)(V^{E}_{22}(v;X,a)-1)dv.
\ene

According to Eqs.~(\ref{E-err-double}) and (\ref{q-E-simi-double}), we have
\begin{align}\label{fanyan-VE12-double}
\begin{aligned}
\widehat{q}^+(X,aX^2)=&-\frac{1}{\pi X^{\frac12}}\int_{\delta_{-b_c}(r)\cup\delta_{b_c}(r)}V^E_{12}(v;X,a)dv+\mathcal{O}(X^{-\frac56})\vspace{0.05in}\\
=&-\frac{1}{\pi X^{\frac12}}\int_{\delta_{-b_c}(r)}X^{-\frac16}X^{-\frac{ip}{3}}e^{iX^{\frac12}s_1(a)}M_1(v;a)^{-1}\vspace{0.05in}\\
&\times\left(T^{-b}(v;a)U^{[1]}(X^{\frac13}r_1(a))T^{-b}(v;a)^{-1}\right)_{12}dv\vspace{0.05in}\\
&-\frac{1}{\pi X^{\frac12}}\int_{\delta_{b_c}(r)}X^{-\frac16}X^{\frac{ip}{3}}e^{-iX^{\frac12}s_2(a)}M_2(v;a)^{-1}\vspace{0.05in}\\
&\times\left(T^{b}(v;a)U^{[1]}(X^{\frac13}r_2(a))T^{b}(v;a)^{-1}\right)_{12}dv+\mathcal{O}(X^{-\frac56})\vspace{0.05in}\\
=&\frac{2i}{X^{\frac{2}{3}}}\bigg(\frac{X^{-\frac{ip}{3}}e^{iX^{\frac12}s_1(a)}}{M_1^{'}(z_1(a);a)}\left(T^{-b}(z_1(a);a)U^{[1]}(X^{\frac13}r_1(a))
T^{-b}(z_1(a);a)^{-1}\right)_{12}\vspace{0.05in}\\
&+\frac{X^{\frac{ip}{3}}e^{-iX^{\frac12}s_2(a)}}{M_2^{'}(z_2(a);a)}\left(T^{b}(z_2(a);a)U^{[1]}(X^{\frac13}r_2(a))T^{b}(z_2(a);a)^{-1}\right)_{12}\bigg)+\mathcal{O}(X^{-\frac56}).
\end{aligned}
\end{align}

Using the L'H\^{o}pital's rule, we know
\begin{align}\label{Hopital}
\begin{aligned}
&T^{-b}(z_1(a);a):=i(z_2(a)-z_1(a))^{-ip\sigma_3}\left(-{M_1^{'}(z_1(a);a)}\right)^{-ip\sigma_3}\sigma_2,\vspace{0.05in}\\
&T^{b}(z_2(a);a):=(z_2(a)-z_1(a))^{ip\sigma_3}M_2^{'}(z_2(a);a)^{ip\sigma_3}.
\end{aligned}
\end{align}

Substituting Eq.~(\ref{Hopital}) into Eq.~(\ref{fanyan-VE12-double}), we obtain
\begin{align}\label{fanyan-VE12-double-2}
\begin{aligned}
\widehat{q}^+(X,aX^2)=&\frac{2i}{X^{\frac{2}{3}}}\bigg(-\frac{X^{-\frac{ip}{3}}e^{iX^{\frac12}s_1(a)}}{M_1'(z_1(a);a)}
\mathcal{V}_1(X^{\frac13}r_1(a))\left(-M_1'(z_1(a);a)(z_2(a)-z_1(a))\right)^{-2ip}\vspace{0.05in}\\
&+\frac{X^{\frac{ip}{3}}e^{-iX^{\frac12}s_2(a)}}{M_2'(z_2(a);a)}\mathcal{V}_2(X^{\frac13}r_2(a))\left((z_2(a)-z_1(a))M_2'(z_2(a);a)\right)^{2ip}\bigg)+\mathcal{O}(X^{-\frac56}).
\end{aligned}
\end{align}

According to the above analysis, we can provide the proof of Theorem \ref{new-prop}.

\begin{proof}
If $a-a_c=\mathcal{O}(X^{-\frac13})$, we have
\begin{align}\label{rs-tayel}
\begin{aligned}
&r_1(a)=r_1'(a_c)(a-a_c)+\mathcal{O}(X^{-\frac23}),\quad s_1(a)=s_1(a_c)+s_1'(a_c)(a-a_c)+\mathcal{O}(X^{-\frac23}),\vspace{0.05in}\\
&r_2(a)=r_2'(a_c)(a-a_c)+\mathcal{O}(X^{-\frac23}),\quad s_2(a)=s_2(a_c)+s_2'(a_c)(a-a_c)+\mathcal{O}(X^{-\frac23}).
\end{aligned}
\end{align}

Substituting Eq.~(\ref{rs-tayel}) into Eq.~(\ref{fanyan-VE12-double-2}), we obtain
\begin{align}\label{fanyan-VE12-double-3}
\begin{aligned}
\widehat{q}^+(X,aX^2)=&\frac{2i}{X^{\frac{2}{3}}}\bigg(-\frac{X^{-\frac{ip}{3}}e^{iX^{\frac12}s_1(a_c)}e^{iX^{\frac12}s_1^{'}(a_c)(a-a_c)}}
{M_1^{'}(-b_c;a_c)}\mathcal{V}_1(X^{\frac13}r_1^{'}(a_c)(a-a_c))\left(-2M_1^{'}(-b_c;a_c)b_c\right)^{-2ip}\vspace{0.05in}\\
&+\frac{X^{\frac{ip}{3}}e^{-iX^{\frac12}s_2(a_c)}e^{-iX^{\frac12}s_2^{'}(a_c)(a-a_c)}}{M_2^{'}(b_c;a_c)}\mathcal{V}_2(X^{\frac13}r_2^{'}(a_c)(a-a_c))
\left(2b_cM_2^{'}(b_c;a_c)\right)^{2ip}\bigg)+\mathcal{O}(X^{-\frac56}).
\end{aligned}
\end{align}

Substituting Eqs.~(\ref{r-s-1}) and (\ref{r-s-2}) into Eq.~(\ref{fanyan-VE12-double-3}), we have
\begin{eqnarray}
\widehat{q}^+(X,aX^2)=&&\!\!\!\!\!\!\frac{4}{X^{\frac{2}{3}}}\mathrm{Im}\bigg(\frac{X^{-\frac{ip}{3}}
e^{iX^{\frac12}s_1(a_c)}e^{iX^{\frac12}s_1'(a_c)(a-a_c)}}
{M_1'(-b_c;a_c)}\mathcal{V}_1(X^{\frac13}r_1'(a_c)(a-a_c))\vspace{0.05in} \no\\
&&\!\!\!\!\!\!\times\left(-2M_1^{'}(-b_c;a_c)b_c\right)^{-2ip}\bigg)+\mathcal{O}(X^{-\frac56})\vspace{0.05in}\no\\
=&&\!\!\!\!\!\! -\frac{4\cdot3^{\frac13}|\mathcal{V}_1(-96\cdot3^\frac13(a+\frac{1}{96}) X^{\frac13})|}{X^{\frac{2}{3}}}\mathrm{Im}\bigg(
e^{-\frac{i\ln(2)}{6\pi}\ln(X)+i\frac{16}{3}X^{\frac12}+i64(a+\frac{1}{96})X^{\frac12}}\vspace{0.05in}\no\\
&&\!\!\!\!\!\!\times e^{i\mathrm{arg}(\mathcal{V}_1(-96\cdot3^\frac13(a+\frac{1}{96}) X^{\frac13}))-i\frac{\ln(2)}{\pi}\ln(4\cdot3^{-\frac13})}\bigg)+\mathcal{O}(X^{-\frac56})\vspace{0.05in} \no\\
=&&\!\!\!\!\!\! -\frac{4\cdot3^{\frac13}|\mathcal{V}_1(-96\cdot3^\frac13(a+\frac{1}{96}) X^{\frac13})|}{X^{\frac{2}{3}}}\sin(\phi(X,a))+\mathcal{O}(X^{-\frac56}).
\label{fanyan-VE12-double-4}
\end{eqnarray}

Thus the proof is completed.
\end{proof}

\section{Conclusions and discussions}

 In summary, we have studied the multi-rational solitons of the c-mKdV equation  with nonzero background in the limit of large order. We construct an infinite-order rational solitons corresponding to a solvable Riemann-Hilbert problem of the c-mKdV equation by a series of deformations. Then, we construct a Riemann-Hilbert problem corresponding to the limit function, which is a new solution of the c-mKdV equation with respect to the new space $X$ and time $T$, and prove the existence and uniqueness of the Riemann-Hilbert problem's solution. Moreover, we also find that the limit functions satisfy ordinary differential equations with respect to $X$ and $T$, respectively. We also analyze the asymptotic behaviors of near-field limit solutions with respect to the large $X$ and transitional asymptotics. {\color{red} For the asymptotics of the large $T$, we give some part results due to the complicated arcs passing through the endpoints and additional requirement of new models}.
 In future, we will try to analyze the asymptotics of the large $T$, and use the RHP-numerical method~\cite{Olver-1,Olver-2} to study the near-field limits of the c-mKdV equation, in order to compare them with the results of this paper. Moreover, the used method can also be extended to other integrable nonlinear wave equations, such as the higher-order c-mKdV euqations, the mKdV equation hierarchy, and the (2+1)-dimensional KP equation.

\vspace{0.2in}
\noindent {\bf Acknowledgments}

\vspace{0.05in}
This work was supported by the National Natural Science Foundation of China (Grant Nos. 11925108 and 12201615).

\end{document}